\PassOptionsToPackage{hypertexnames=false}{hyperref}
\documentclass{theoretics}

  \allowdisplaybreaks

\usepackage{amsthm}
\usepackage{thm-restate}
\usepackage{xfrac}

\newcommand{\RestateRemark}[1]{{\normalfont\bfseries #1}}
\newcommand{\RestateInit}[1]{\newcommand{#1}{}}
\newcommand{\RestateGo}[1]{\renewcommand{#1}{(Restated)}}

\usepackage{tikz}
\usepackage[capitalise,noabbrev,nameinlink]{cleveref}

\crefname{prop}{property}{properties}
\creflabelformat{prop}{#2{#1}#3}
\crefname{modification}{modification}{modifications}

\crefname{modification}{modification}{modifications}

\newcommand*{\B}{\mathcal{B}}

\let\oldnl\nl
\newcommand{\nonl}{\renewcommand{\nl}{\let\nl\oldnl}}

\DeclareMathOperator{\Pb}{P}
\DeclareMathOperator{\Ev}{E}
\DeclareMathOperator{\pred}{pred}
\DeclareMathOperator{\tr}{tr}
\DeclareMathOperator{\seg}{seg}

\newcommand*{\Otilde}{\widetilde{O}}
\newcommand*{\poly}{\textsf{poly}}
\newcommand*{\polylog}{\textsf{polylog}}

\newcommand*{\length}[1]{|#1|}

\newcommand*{\ball}{\mbox{\rm ball}}

\newcommand*{\nwspace}{\hspace*{.1em}} 

\renewcommand{\leq}{\leqslant}
\renewcommand{\geq}{\geqslant}
\renewcommand{\le}{\leqslant}
\renewcommand{\ge}{\geqslant}
\renewcommand{\epsilon}{\varepsilon}
\newcommand{\eps}{\varepsilon}

\let\oldsqrt\sqrt
\def\hksqrt{\mathpalette\DHLhksqrt}
\def\DHLhksqrt#1#2{\setbox0=\hbox{$#1\oldsqrt{#2\,}$}\dimen0=\ht0
   \advance\dimen0-0.2\ht0
   \setbox2=\hbox{\vrule height\ht0 depth -\dimen0}%
   {\box0\lower0.4pt\box2}}
\renewcommand\sqrt\hksqrt

\newcommand*{\gwidth}{3.7}
\newcommand*{\gdist}{1.5}
\newcommand*{\gpad}{0.5}
\newcommand*{\noderadius}{7.5pt}
\usetikzlibrary{arrows.meta, quotes}
\tikzset{
node distance={\trlength cm}, vert/.style = {draw, circle, inner sep = 0 pt, minimum size = 2 * \noderadius}, every path/.style = {-Latex}, every label/.append style={rectangle, font = {\normalsize}}, fac/.style = {circle,fill,inner sep=1.5pt}, every node/.style = {font = {\normalsize}}, every edge quotes/.style = {auto, sloped, font = {\normalsize}, inner sep = 0.5pt}
}

\definecolor{cred}{HTML}{D81B60}
\definecolor{cblue}{HTML}{1E88E5}
\definecolor{cyellow}{HTML}{D09C00}
\definecolor{cgreen}{HTML}{5B8600}

\title{Approximate Distance Sensitivity Oracles in Subquadratic Space}

\ThCSshortnames{D.~Bil\`o, S.~Chechik, K.~Choudhary, S.~Cohen, T.~Friedrich, S.~Krogmann, M.~Schirneck} 
 \ThCSshorttitle{Approximate Distance Sensitivity Oracles in Subquadratic Space} 
 
\ThCSauthor[1]{Davide Bil\`o}{davide.bilo@univaq.it}[https://orcid.org/0000-0003-3169-4300]
\ThCSauthor[2]{Shiri Chechik}{shiri.chechik@gmail.com}[]
\ThCSauthor[3]{Keerti Choudhary}{keerti@iitd.ac.in}[https://orcid.org/0000-0002-8289-5930]
\ThCSauthor[4]{Sarel Cohen}{sarelco@mta.ac.il}[https://orcid.org/0000-0003-4578-1245]
\ThCSauthor[5]{Tobias Friedrich}{tobias.friedrich@hpi.de}[https://orcid.org/0000-0003-0076-6308]
\ThCSauthor[5]{Simon Krogmann}{simon.krogmann@hpi.de}[https://orcid.org/0000-0001-6577-6756]
\ThCSauthor[6]{Martin Schirneck}{martin.schirneck@univie.ac.at}[https://orcid.org/0000-0001-7086-5577]

\ThCSaffil[1]{Department of Information Engineering, Computer Science and Mathematics, University of L'Aquila}
\ThCSaffil[2]{Department of Computer Science, Tel Aviv University}
\ThCSaffil[3]{Department of Computer Science and Engineering, Indian Institute of Technology Delhi}
\ThCSaffil[4]{School of Computer Science, The Academic College of Tel Aviv-Yaffo}
\ThCSaffil[5]{Hasso Plattner Institute, University of Potsdam}
\ThCSaffil[6]{Faculty of Computer Science, University of Vienna}

\ThCSyear{2024}
\ThCSarticlenum{15}
\ThCSreceived{Aug 2, 2023}
\ThCSrevised{Mar 8, 2024}
\ThCSaccepted{Mar 29, 2024}
\ThCSpublished{Jun 5, 2024}
\ThCSdoicreatedtrue

\ThCSthanks{An extended abstract of this work appeared at the 55th Symposium on Theory of Computing (STOC 2023)~\cite{Bilo23SubquadraticSpaceSTOC}. The authors thank Merav Parter for raising the question of designing distance sensitivity oracles that require only subquadratic space. This project received funding from the European Research Council (ERC) under the European Union's Horizon 2020 research and innovation program, grant agreement No.~803118 ``The Power of Randomization in Uncertain Environments (UncertainENV)'' and grant agreement No.~101019564 ``The Design and Evaluation of Modern Fully Dynamic Data Structures (MoDynStruct)''.}

\ThCSkeywords{data structures, distance sensitivity oracles, fault tolerance, shortest paths, subquadratic space} 

\begin{document}
\maketitle

\begin{abstract}
	An \emph{\(f\)-edge fault-tolerant distance sensitive oracle} (\(f\)-DSO) with stretch \(\sigma \ge 1\)
	is a data structure that preprocesses a given undirected, unweighted graph \(G\) with \(n\) vertices and \(m\) edges, 
	and a positive integer \(f\).
	When queried with a pair of vertices \(s, t\) and a set \(F\) of at most \(f\) edges,
	it returns a \(\sigma\)-approximation of the \(s\)-\(t\)-distance in \(G-F\).

	We study \(f\)-DSOs that take subquadratic space.
	Thorup and Zwick [JACM~2005] showed that this is only possible for \(\sigma \ge 3\).
	We present, for any constant \(f \ge 1\) and \(\alpha \in (0, \frac{1}{2})\), 
	and any \(\varepsilon  > 0\), 
	a randomized \mbox{\(f\)-DSO} with stretch \( 3 + \varepsilon\)
	that w.h.p.\ takes 
	\(\widetilde{O}(n^{2-\frac{\alpha}{f+1}}) \cdot O(\log n/\varepsilon)^{f+2}\)
	space and has an \(O(n^\alpha/\varepsilon^2)\) query time.
	The time to build the oracle is \(\widetilde{O}(mn^{2-\frac{\alpha}{f+1}}) \cdot O(\log n/\varepsilon)^{f+1}\).
	We also give an improved construction for graphs with diameter at most \(D\).
	For any positive integer \(k\), we devise an \(f\)-DSO with stretch \(2k-1\) that w.h.p.\ takes
	\(O(D^{f+o(1)} n^{1+1/k})\) space and has \(\widetilde{O}(D^{o(1)})\) query time,
	with a preprocessing time of \(O(D^{f+o(1)} mn^{1/k})\).
  
	Chechik, Cohen, Fiat, and Kaplan [SODA 2017] devised an \(f\)-DSO with stretch \(1{+}\varepsilon\) 
	and preprocessing time $O(n^{5+o(1)}/\varepsilon^f)$, 
	albeit with a super-quadratic space requirement.
	We show how to reduce their preprocessing time to $O(mn^{2+o(1)}/\varepsilon^f)$.
\end{abstract}

\section{Introduction}
\label{sec:intro}

\emph{Distance Oracles} (DOs) are fundamental data structures 
that store information about the distances of an input graph $G=(V,E)$.\footnote{%
	Throughout, we assume the graph $G$ to be undirected, unweighted, have $n$ vertices,
	and $m$ edges.
}
These oracles are used in applications where one cannot afford to store the entire input,
but still wants to quickly retrieve the graph distances upon query.
Therefore, DOs should provide reasonable  trade-offs between space consumption, query time,
and \emph{stretch}, that is, the quality of the estimated distance. 

We are interested in the design of DOs that additionally can tolerate multiple failures of edges in $G$. 
An \emph{$f$-edge fault-tolerant distance sensitivity oracles} ($f$-DSO) is able to report an estimate $\widehat{d}_{G-F}(s,t)$ of the distance $d_{G-F}(s,t)$ between $s$ and $t$ in the graph $G-F$, where $F \subseteq E$ is a set of at most $f$ failing edges, when queried with the triple $(s,t,F)$. 
The parameter $f$ is the \emph{sensitivity} of the DSO.
We say that the \emph{stretch} of the $f$-DSO is $\sigma \geq 1$ if $d_{G-F}(s,t) \leq \widehat{d}_{G-F}(s,t) \leq \sigma \,{\cdot}\, d_{G-F}(s,t)$ holds for every query $(s,t,F)$. 

Several $f$-DSOs with different size-stretch-time trade-offs have been proposed in the last decades.
Some of them can only deal with a very small number $f \in \{1,2\}$ of failures~\cite{BaswanaK13, BeKa08, BeKa09, BCFS21SingleSourceDSO_ESA, ChCo20, DemetrescuT02, DeThChRa08, DuanP09a, GrandoniVWilliamsFasterRPandDSO_journal, GuRen21, Ren22Improved}. In the following, we focus on $f$-DSOs that deal with multiple failures $f \geq 3$. The randomized $f$-DSO of Weimann and Yuster~\cite{WY13} computes exact distances w.h.p.\footnote{%
  An event occurs \emph{with high probability} (w.h.p.) if it has probability at least $1- n^{-c}$ for some constant $c >0$.
}
and gives adjustable trade-offs depending on some parameter $\alpha \in (0,1)$.
More precisely, the data structure can be built in $\Otilde(mn^{2-\alpha})$ time, has a query time of 
$\Otilde(n^{2-2(1-\alpha)/f})$, and uses $\Otilde(n^{3-\alpha})$ space.\footnote{%
  The space is measured in the number of machine words on $O(\log n)$ bits.
} 
For any constant $f$, the $f$-DSO of Duan and Ren~\cite{DuRe22} requires $O(fn^4)$ space and
returns exact distances in $f^{O(f)}$ query time, but the preprocessing algorithm 
takes $n^{\Omega(f)}$ time. The $f$-DSO of Chechik, Cohen, Fiat, and Kaplan~\cite{ChCoFiKa17}
can handle up to $f = o(\log n / \log \log n)$ failures
but has a stretch of $1+\eps$, for any approximation parameter $\eps \ge 1/n$.
In turn, the oracle is more compact requiring $O_{\varepsilon}(n^{2+o(1)}\log W)$ space,\footnote{%
	For any positive function $g$ of the input and parameters, we use 
	$\Otilde(g)$ to denote $O(g \cdot \polylog(n))$,
	and $O_{\varepsilon}(g)$ for $O(g \cdot \poly(1/\varepsilon))$.
	$\Otilde_{\varepsilon}(g)$ combines the two.
}
where $W$ is the weight of the heaviest edge of $G$,
has query time $O_{\varepsilon}(f^5\log n \log \log W)$,
and can be built in $O_{\varepsilon}(n^{5+o(1)}\log W)$ preprocessing time.
Note that the aforementioned $f$-DSOs all have a super-quadratic space requirement,
that is, they take up \emph{more} space than the original input graph,
which is prohibitive in settings where we cannot even afford to store $G$.
The $f$-DSO of Chechik, Langberg, Peleg, and Roditty~\cite{CLPR12} addresses this issue
with a space requirement of $O(fkn^{1+1/k}\log(nW))$, 
where $k\geq 1$ is an integer parameter.
Their data structure has a fast query time of $\Otilde(|F| \log\log d_{G-F}(s,t))$
but guarantees only a stretch of $(8k-2)(f+1)$,
that is, depending on the sensitivity $f$.

Another way to provide approximate pairwise replacement distances under edge failures is that of fault-tolerant spanners~\cite{Levcopoulos98FaultTolerantGeometricSpanners}.
An \emph{($f$-edge) fault-tolerant $\sigma$-spanner} 
is a subgraph $H$ of~$G$ such that
$d_{H-F}(s,t) \allowbreak\le \sigma \cdot d_{G-F}(s,t)$, for every triple $(s,t,F)$, with $s,t \in V$ and $F\subseteq E$, $|F| \le f$.
There is a simple algorithm by 
Chechik, Langberg, Peleg, and Roditty~\cite{Chechik10FaultTolerantSpannersGeneralGraphs}
that computes, for any positive integer $k$,
a fault-tolerant $(2k{-}1)$-spanner with $O(f n^{1+1/k})$ edges.
Constructions by Bodwin, Dinitz, and Robelle~\cite{BodwinDinitzRobelle21OptimalVFTSpanners,BodwinDinitzRobelle22PartiallyOptimalEFTSpanners} recently reduced the size
to $f^{1/2}n^{1+1/k} \cdot \poly(k)$ for even $k$,
and $f^{1/2-1/(2k)}n^{1+1/k} \cdot \poly(k)$ for odd 
$k$.
They also showed an almost matching lower bound of $\Omega(f^{1/2-1/(2k)} n^{1+1/k} + fn)$ for $k > 2$, and $\Omega(f^{1/2}n^{3/2})$ for $k = 2$,
assuming the Erdős girth conjecture~\cite{Erd64extremal}.
The space is also the main problem with this approach
as it translates to a high query time. 
Currently, the most efficient way to retrieve the approximate distance 
between a given pair of vertices is to compute the single-source distance from one of the endpoint
taking at least linear in the size of the spanner.

All the results above for multiple failures either require $\Omega(n^2)$ space, 
have a stretch depending on $f$, or superlinear query time.
If one wants a truly constant stretch and fast query time simultaneously,
one currently has to pay $\Omega(n^2)$ space.
Thorup and Zwick~\cite{ThorupZ05} showed that, even when not supporting a single failure,
breaking the quadratic barrier is impossible for directed graphs; 
and for undirected graphs this requires a stretch of at least $3$.
In this paper, we discuss the case of unweighted graphs and constant sensitivity.
We give a subquadratic-space DSO with near-optimal stretch $3+\varepsilon$
and an arbitrarily small polynomial query time.
\RestateInit{\restateoraclelongpaths}
\begin{restatable}{theorem}{oraclelongpaths}
\label{thm:oracle_long_paths}\RestateRemark{\restateoraclelongpaths}
  Let $f \ge 2$ be an integer constant and $0 < \alpha < \sfrac{1}{2}$.
  For any  undirected, unweighted graph $G$
  with unique shortest paths
  and any $\varepsilon = \varepsilon(m,n,f) > 0$,
  there is a $(3{+}\eps)$-approximate randomized \mbox{$f$-DSO} for $G$
  that w.h.p.\ takes space 
  $\Otilde(n^{2-\frac{\alpha}{f+1}}) \,{\cdot}\, O(\log n/\varepsilon)^{f+2}$,
  has query time $O(n^{\alpha}/\varepsilon^2)$,
  and preprocessing time
  $\Otilde(mn^{2-\frac{\alpha}{f+1}}) \cdot O(\log n/\varepsilon)^{f+1}$.
\end{restatable}

Note that the literature on $f$-DSOs generally assumes that the only valid queries
to the oracle are triples $(s,t,F)$ where $F \subseteq E$, $|F| \le f$, is a set of edges
actually present in $G$.
\emph{Verifying} whether a given query is valid requires $\Omega(n^2)$ space for dense graphs.
This would make the goal of a subquadratic-space $f$-DSO impossible.
However, our query algorithm never uses the assumption $F \subseteq E$.
This allows it to process any triplet $(s,t,F)$ 
where $F \subseteq \binom{V}{2}$ is a set of at most $f$ pairs of vertices.
It then returns the approximate distance for the valid query~$(s,t, F \cap E)$.\footnote{
	See \Cref{lemma:approx_parent_to_child,lem:weaker_guarantee} for the technical details.
}

The assumption in \Cref{thm:oracle_long_paths} of shortest paths being unique
in the base graph $G$ can be achieved w.h.p.\ by randomly perturbing the edge weights of the input,
while keeping the characteristics of an essentially unweighted graph.
For an unweighted graph, this results in edge weights of $1 \pm o(1/n)$
not affecting the graph distances.
This is a sufficient alternative condition in all places of the paper
where we assume $G$ to be unweighted.
The change on the preprocessing time is negligible.
Computing all-pairs shortest paths via Dijkstra's algorithm from every vertex of a weighted graph
takes time $\Otilde(mn)$, while using breath-first searches 
in an unweighted graph takes time $O(mn)$.
As an alternative,
we can also precompute a set of unique paths via so-called 
\emph{lexicographic perturbation}~\cite{CaChEr13} in time $O(mn + n^2 \log^2 n)$.

Very recently, Bilò, Choudhary, Cohen, Friedrich, Krogmann, and Schirneck~\cite{Bilo23CompactDOLargeSensitivity} addressed the same problem as we do in this work.
Their setting is more general in that they obtain an $f$-DSO for graphs with 
edge weights that are non-negative and polynomially bounded,
and support a sensitivity of $f = o(\log n/\log\log n)$.
For any integer $k \ge 2$ and constant \mbox{$0 < \alpha < 1$}
their construction achieves a stretch of $2k-1$ with space $O(n^{1+\frac{1}{k}+\alpha+o(1)})$
and an $\Otilde(n^{1+\frac{1}{k}-\frac{\alpha}{k(f+1)}})$ query time.
In comparison, their space requirement is smaller for the price of a query time that is always
at least linear (albeit smaller than running a single-source shortest path algorithm on a spanner).
The construction by Bilò et al.~\cite{Bilo23CompactDOLargeSensitivity} also differentiates between long and short paths, which is common for fault-tolerant data structures, 
and employs the distance oracle of Thorup and Zwick~\cite{ThorupZ05}.
Besides that, they use techniques that are different from ours.

In order to prove \Cref{thm:oracle_long_paths}, we develop several new ideas.
For the remainder of this section, we highlight the novelties.
A more detailed overview
can be found in \Cref{sec:overview}.
\vspace*{.5em}

\noindent
\textbf{Tree Sampling for Short Paths.}
It is a common approach in the design of fault-tolerant data structures
to first give a solution for short paths and then combine them into one for all distances,
see~\cite{ChCo20, KarthikParter21DeterministicRPC, GrandoniVWilliamsFasterRPandDSO_journal, GuRen21, Ren22Improved, WY13}.
We also focus first on $f$-DSOs for short paths.
Let $L$ be the cut-off parameter.\footnote{%
	The cut-off point will eventually turn out to be $L = n^{\alpha/(f+1)}$,
	where $\alpha \in (0,\frac{1}{2})$ is the parameter from \Cref{thm:oracle_long_paths}.
} 
We say a path is \emph{short} if it has at most $L$ edges.
An $f$-DSO for short paths only needs to report the correct answer for a query $(s,t,F)$
if $G-F$ contains a shortest path from $s$ to $t$ with at most $L$ edges. 
	Designing such an oracle with good query-space-preprocessing trade-offs 
	is the first step towards improving general $f$-DSOs. 
    Let $d_{G-F}^{\le L}(s,t)$ be the minimum length over all $s$-$t$-paths in $G-F$ with at most $L$ edges; if there are none, then $d_{G-F}^{\le L}(s,t) = +\infty$.
    Note that $d_{G-F}^{\le L}(s,t) = +\infty$ may hold for pairs $(s,t)$
    that are connected in $G-F$.
\begin{restatable}{theorem}{oracleshortpaths}
\label{thm:oracle_short_paths}
  Let $f, k \ge 1$ be integer constants.\footnote{%
  In principle, $k$ could depend on $n$ or $m$,
  but for $k = \Omega(\log n)$ we do not get further space improvements.
}
  There exists a randomized data structure that,
  when given an undirected, unweighted graph $G = (V,E)$, 
  and a positive integer $L$ (possibly dependent on $n$ and $m$),
  preprocesses $G$ and answers queries $(s,t,F)$ 
  for vertices $s,t \in V$ and sets of edges $F \subseteq E$ with $|F| \le f$.
  W.h.p.\ over all queries,
  the returned value $\widehat{d^{\le L}}(s,t,F)$ satisfies 
  $d_{G-F}(s,t) \le \widehat{d^{\le L}}(s,t,F) \le (2k{-}1) \,{\cdot}\,d_{G-F}^{\le L}(s,t,F)$.
  The data structure takes space $\Otilde(L^{f+o(1)} \nwspace n^{1+1/k})$, 
  has query time $\Otilde(L^{o(1)})$, 
    and preprocessing time $\Otilde(L^{f+o(1)} \nwspace mn^{1/k})$.
\end{restatable}

We compare \Cref{thm:oracle_short_paths} with previous work on $f$-DSOs for short paths.
Weimann and Yuster~\cite{WY13} presented a construction with $\Otilde(L^f mn)$ preprocessing time, 
$\Otilde(L^f n^2)$ space, and $\Otilde(L^f)$ query time. 
It laid the foundation for many subsequent works, see~\cite{AlonChechikCohen19CombinatorialRP,BCFS21DiameterOracle_MFCS,Bilo22Extremal,
KarthikParter21DeterministicRPC,Ren22Improved}.
When using the fault-tolerant trees of Chechik et al.~\cite{ChCoFiKa17}, 
one can reduce the query time of the oracle to $O(f^2)$.
However, storing all of these fault-tolerant trees still requires $\Omega(L^f n^2)$ space.
For small enough $L$, sub-quadratic space suffices for our data structure,
while still providing a better query time than the oracle by Weimann and Yuster~\cite{WY13}. 

We extend their sampling technique~\cite{WY13}
in order to prove~\Cref{thm:oracle_short_paths}.
The technique consists of first constructing $\Otilde(L^f)$ copies of $G$ and then, in each one,
remove edges with probability $\sfrac{1}{L}$.
One can show that w.h.p.\ each short replacement path survives in one of the copies,
where a \emph{replacement path} is the respective shortest path after at most $f$ edge failures.
Instead of having all those graphs be independent of each other, 
we develop hierarchical tree sampling.
This allows us to quickly find
the copies that are relevant for a given query, 
reducing the query time to $\Otilde(L^{o(1)})$.
We further sparsify the resulting graphs for a better space complexity.

From~\Cref{thm:oracle_short_paths}, we immediately get an $f$-DSO for graphs with bounded diameter. Afek, Bremler-Barr, Kaplan, Cohen, and Merritt~\cite{Afek02RestorationbyPathConcatenation_journal}
proved that for undirected, unweighted graphs $G$ any shortest path in $G-F$
is a concatenation of up to $|F|+1$ shortest paths in $G$.
If $G$ has diameter at most $D$ and $|F| \le f$, the diameter of $G-F$ is thus bounded by $(f{+}1)D$.

\begin{corollary}
\label{cor:oracle_bounded_diameter}
  Let $f,k \ge 1$ be integer constants.
  There exists a randomized $(2k{-}1)$-approximate $f$-DSO for undirected, unweighted graphs
  with diameter bounded by $D$,
  that w.h.p.\ takes space $\Otilde(D^{f+o(1)} \nwspace n^{1+1/k})$, 
  has query time $\Otilde(D^{o(1)})$, and preprocessing time $\Otilde(D^{f+o(1)} \nwspace mn^{1/k})$.
\end{corollary}

\noindent
\textbf{Fault-Tolerant Trees with Granularity.}
We employ fault-tolerant trees\footnote{%
  The FT-trees are not related to the tree sampling mentioned before.
}
(FT-trees) introduced by Chechik et al.~\cite{ChCoFiKa17}
to combine the solutions for short paths.
Those are trees in which every node is associated with a path in a subgraph
$G-A$ where $A \subseteq E$ is a set of edges, but possibly much more than the sensitivity $f$.
Each path is partitioned into segments whose sizes increase exponentially towards the middle.
This is done to encode the paths more space efficient than edge-by-edge.
We have to take some additional compression steps to fit them in subquadratic space.
For example, instead of building a tree $FT(s,t)$ for every pair of vertices $s,t$,
we only do so if one of them is from a set of randomly selected \emph{pivots}.
But even this gives only a sublinear query time.
To improve it further to $\Otilde_\eps(n^\alpha)$ 
for an any constant $\alpha \in (0,\frac{1}{2})$,
we generalize the FT-trees by adding what we call \emph{granularity} $\lambda \ge 0$.\footnote{%
	In the proof of \Cref{thm:oracle_long_paths},
	we set $\lambda = \eps L/c$, for an ad-hoc constant $c > 1$.
}
That means the first and last $\lambda$ edges of each path are their own segment
and do not fall into the regime of exponential increase.
The original construction \cite{ChCoFiKa17} corresponds to granularity $0$.
Intuitively, the larger the value of $\lambda$, the better the fault-tolerant tree $FT_\lambda(u,v)$ with granularity $\lambda$ approximates the shortest distance from $u$ to $v$ in $G-F$, but the larger the size of each node of the tree becomes.

The idea to answer a query $(s,t,F)$ is to scan balls of a certain radius around
$s$ and $t$ in $G-F$ for pivots and query the respective FT-tree together with
the oracle for short paths in \Cref{thm:oracle_short_paths}.
W.h.p.\ one of the pivots hits the replacement path from $s$ to $t$ ensuring
that this gives (an approximation of) the right distance.
The bottleneck is the case when there are too many vertices in the vicinity of both $s$ and $t$
since then these balls also receive many pivots.
Instead, we sample a second type of much more scarce pivots, 
which are used to hit the dense neighborhoods.
In that case, we can find a scarce pivot $b_s$ near $s$ and a scarce pivot $b_t$ near~$t$, 
but we can no longer assume that they hit the sought replacement path.
The fault-tolerant tree $FT_\lambda(b_s, b_t)$ with granularity $\lambda$, however,
allows us to get a good approximation,
as long the starting points $b_s$ and $b_t$ are at distance at most $\lambda$
from the real endpoints. 

The trees $FT_\lambda(b_s, b_t)$ are much larger than their classical counterparts $FT(s,t)$.
This is compensated by the fact that we require much fewer of those. 
We verify that several of the key lemmas from~\cite{ChCoFiKa17} transfer to 
fault-tolerant trees with granularity $\lambda >0$. 
\vspace*{.5em}

\noindent
\textbf{Efficient Computation of Expaths.}
Since fault-tolerant trees are crucial for our work,
we revisit the approach used in~\cite{ChCoFiKa17} to construct them (with granularity $0$).
It turns out that their algorithm can be improved. 
The preprocessing in~\cite{ChCoFiKa17} invokes many calls to all-pairs shortest path computations (APSP) in different subgraphs $G-F$, each of which is associated with a node of the fault-tolerant trees. 
They also invoke $O(n)$ calls to Dijkstra's algorithm on suitable dense graphs with $O(fn^2)$ edges. 
We prove that many of those APSP calls can be avoided by instead re-using the distances in the original graph $G$, which can be obtained by a single APSP computation.
More precisely, the paths associated with the nodes of the fault-tolerant trees
(later referred as~$(2f+1)$-expaths) are the concatenation of $O(f \log (nW))$ original shortest paths.
The distances in $G$ can be integrated into a single Dijkstra run on a specially built graph with $\Otilde(fm)$ edges to compute such an expath in time $\Otilde(fm)$.
This technique implies an improved preprocessing time for our own subquadratic $f$-DSO.
Moreover, when applied to the $f$-DSO by Chechik et al.~\cite{ChCoFiKa17}, it improves their preprocessing time from $O_\varepsilon(fn^{5+o(1)})$ to $O_\varepsilon(fmn^{2+o(1)})$. 

\begin{theorem}
\label{thm:improved_preprocessing}
  Let $G$ be an undirected weighted graph with maximum edge weight $W = \emph{\poly}(n)$
  and unique shortest paths.
  For any positive integer $f = o(\log n/\log\log n)$,
  and $\varepsilon \ge 1/(nW)$,
  there exists an $(1{+}\varepsilon)$-approximate $f$-DSO for $G$ that takes
  $\Otilde(fn^2) \cdot O\!\left(\frac{\log(nW)}{\varepsilon}\right)^f 
    = O(\varepsilon^{-f}) \cdot n^{2+o(1)}$ space,
  has query time $O(f^5 \log n)$,
  and preprocessing time 
  $\Otilde(fmn^2) \cdot O\!\left(\frac{\log(nW)}{\varepsilon}\right)^f 
  	= O(\varepsilon^{-f}) \cdot mn^{2+o(1)}$.
\end{theorem} 

\noindent
\textbf{Open Problems.}
As an open question, we ask whether one can further improve the query time from 
$\Otilde_\eps(n^\alpha)$ to poly-logarithmic in $n$ and $\sfrac{1}{\eps}$ 
while keeping the space truly subquadratic.
The converse problem is to further reduce the space without affecting the query time.
Finally, we can currently only handle unweighted graphs
where the length of the path corresponds to the number of edges.
Some sampling-based ideas break down if long paths can consist of only a few heavy edges.
In all cases, the bottleneck is the handling of long paths.
For short distances, our $f$-DSO has asymptotically almost optimal size and very low query time that can easily be adapted to the weighted case.

\section{Overview}
\label{sec:overview}

\textbf{Fault-tolerant Trees.}
Our distance sensitivity oracle is built on the concept of fault-tolerant trees~\cite{ChCoFiKa17}.
This is a data structure that reports,
for a fixed pair of vertices $s,t \in V$ and any set $F \subseteq E$ of up to $f$ edge failures,
the replacement distance $d_{G-F}(s,t)$.
Consider a shortest path~$P$ from $s$ to $t$ in the original graph $G$.
FT-trees draw from the fact that only failures on $P$ can influence the distance from $s$ to $t$.
In its simplest form, the tree $FT(s,t)$ consists of a root node that stores the path $P$
and the distance $d(s,t) = |P|$.
It has a child for each edge $e \in E(P)$ which in turn holds a shortest $s$-$t$-path in $G\,{-}\,e$.
Iterating this construction until depth $f$ ensures 
that all relevant failure sets for the pair $(s,t)$ are covered.
If some set of edge failures disconnect the two vertices,
this is represented by a leaf node that does not store any path.
Let $P_\nu$ denote the path in some node $\nu$.
Given a failure set $F$, the algorithm checks in each node $\nu$ starting with the root 
whether it is a leaf or $F \cap E(P_\nu) = \emptyset$,
with the latter meaning that the path $P_\nu$ exists in $G-F$.
If so, its length $\length{P_\nu}$ is reported;
otherwise, the search recurses on the child node corresponding to 
an (arbitrary) edge $e \in F \cap E(P_\nu)$.
Let $FT(s,t,F)$ be the reported distance.
It is equal to $d_{G-F}(s,t)$ and the query time is $O(f^2)$ since at most $f{+}1$ vertices are visited and
computing the intersection takes time $O(f)$.

The problem is, these trees are huge.
Preprocessing them for all pairs of vertices takes total space $O(n^{f+3})$.
The main technical contribution of~\cite{ChCoFiKa17} is to reduce the space without sacrificing too much of their performance,
that is, the stretch of the reported distance and the query time.
In the first step, the number of vertices in the tree is decreased
by introducing an approximation parameter $\eps > 0$.
Each path $P_\nu$ is split into $O(\log n/\eps)$ \emph{segments}.
Now node $\nu$ only has a child for each segment and the search procedure recursing 
on that child corresponds to failing the whole segment instead of only a single edge.
This reduces the total size of all trees to $O(n^3 \, (c \nwspace \frac{\log n}{\eps})^f)$
for some constant $c > 0$.
However, it leads to some inaccuracies in the answer of the tree.
The failed segments may contain edges that are actually present in $G-F$
and thus the path $P_{\nu^*}$ stored in the last visited node $\nu^*$ may take unnecessary detours.
It is proven in~\cite{ChCoFiKa17} that $FT(s,t,F) = |P_{\nu^*}| = d_{G-F}(s,t)$
is correct if all failing edges are ``far away''\footnote{%
  More formally, a path $P$ being ``far away'' from $F$ means that,
  for every vertex $x$ on $P$ except for $s$ and $t$ 
  and every endpoint $y$ of a failing edge in $F$,
  the distance from $x$ to $y$ is more than
  $\frac{\eps}{9} \cdot \min(\nwspace |P[s,x]|,\, |P[x,t]| \nwspace )$,
  see~\Cref{def:trapezoid}.
}
from the true replacement path $P(s,t,F)$ in $G-F$,
where the required safety distance depends on the distance $d_{G-F}(s,t)$.
To also answer queries for which this condition is violated,
they consult multiple FT-trees.
An auxiliary graph $H^F$ is constructed on the endpoints $V(F)$
of all failing edges, that is, $V(H^F) = \{s,t\} \cup V(F)$.
For each pair of vertices $u,v \in V(H^F)$, the edge $\{u,v\}$ is weighted
with the reported distance $FT(u,v,F)$.
While not all edge weights may be the correct $u$-$v$-replacement distance,
the distance of $s$ and $t$ in $H^F$ can be shown 
to be a $(1{+}\eps)$-approximation of $d_{G-F}(s,t)$.
The idea is that, when going from $s$ to $t$, one can always find a next vertex in $V(H^F)$
that is not too far off the shortest path and
such that the subpath to that vertex is ``far away'' from all failures.
Computing the weights for $H^F$ increases the query time to $O(f^4)$.

The next step is more involved and is concerned with the size of the nodes in the FT-trees.
Originally, each of them stores all edges of a path in (a subgraph of) $G$ 
and therefore may take $O(n)$ space.
Afek et al.~\cite{Afek02RestorationbyPathConcatenation_journal}
showed that every shortest path in $G{-}F$, for $|F| \le f$, is  $f$-\emph{decomposable},
that is, a concatenation of at most $f$ shortest paths in $G$.
Chechik et al.~\cite{ChCoFiKa17} extend this notion to so-called expaths.
For a positive integer $\ell$, a path is said to be an \mbox{$\ell$-\emph{expath}}
if it is the concatenation of $(2\log_2(n)+1)$ \mbox{$\ell$-decomposable} paths
such that the $i$th $\ell$-decomposable path has length at most $\min\{ 2^i,2^{2\log_2(n) -i}\}$.
Consider a node $\nu$ in the tree $FT(u,v)$.
Instead of storing the shortest $u$-$v$-path $P_\nu$ edge by edge,
one would like to represent it
by the endpoints of the constituting shortest paths (in $G$) and edges.
However, the collection $A_\nu$ of edges in all segments
that were failed while descending from the root to $\nu$ may be much larger than $f$
and~$P_{\nu}$ may not be $f$-decomposable.
Instead, the node $\nu$ now holds the shortest $(2f{+}1)$-expath from~$u$ to $v$ in $G-A_\nu$.
It can be represented by $O(f \log n)$ endpoints,
bringing the total space of the trees to $O(fn^2 (\log n)(c \nwspace \frac{\log n}{\eps})^f)$.
It is described in~\cite{ChCoFiKa17} how to navigate the new representation
to obtain a $(1{+}\eps)$-approximation of $d_{G-F}(s,t)$ in time $O(f^5 \log n)$.

In this work, we advance the space reduction further into the subquadratic regime.
Recall that $L$ is the number of edges up to which a path is called short.
When sampling a set $B$ of~$\Otilde_\eps(n/L)$ \emph{pivots} uniformly at random,
then w.h.p.\ every long replacement path contains a pivot. 
Restricting the FT-trees $FT(u,v)$ to only those pairs $u,v$ 
for which at least one vertex is in~$B$ brings the total number of trees to $o(n^2)$.
Unfortunately, it deprives us of the replacement distances for pairs
that are joined by a short path.
\vspace*{.5em}

\noindent
\textbf{Short Paths.}
To make up for this deficit, we design an approximate $f$-DSO for
vertex pairs with short replacement paths.
We extend a technique by Weimann and Yuster~\cite{WY13}
from exact to approximate distances while also reducing the required space and query time.
Let $\{G_i\}_i$ be a collection of spanning subgraphs of $G$.
It is called an 
$(L,f)$-\emph{replacement path covering} (RPC)~\cite{KarthikParter21DeterministicRPC}
if, for every set $F$ of at most $f$ edges and any pair of vertices that is joined in $G{-}F$
by a path of at most $L$ edges, there exists some subgraph $G_i$ that does not contain any edge of $F$
but all edges of the path.
This construction is the basis of many $f$-DSOs for the following reason.
Consider two vertices $s$ and $t$ that have a replacement path on at most $L$ edges.
Scan over all graphs of the RPC that contain no edge of $F$,
and, for each graph, record the $s$-$t$-distance.
By the properties of an RPC, the minimum recorded value is the correct replacement distance $d_{G-F}(s,t)$.
This holds even if $G{-}F$ itself is not in $\{G_i\}_i$.

Weimann and Yuster~\cite{WY13} showed that one can obtain an RPC with high probability by the following process.
Take $\Otilde(L^f)$ copies of $G$ and,
in each one, remove any edge independently with probability $\sfrac{1}{L}$.
We cannot use that approach directly in subquadratic space.
The subgraphs have total size $\Omega(L^{f-1} \nwspace m)$, which is already too large if $G$ is dense.
Also, it is expensive to find the correct members of the RPC for a given query.
In~\cite{WY13}, the solution was to indeed go over all graphs and explicitly check
whether they have the set $F$ removed.
Karthik and Parter~\cite{KarthikParter21DeterministicRPC} derandomized this
construction and reduced the time needed to find the correct subgraphs to $\Otilde(L)$.
Both approaches break down in subquadratic space, since we cannot even store all edges of the graphs.
However, we are only seeking \emph{approximate} replacement distances.
We exploit this fact in a new way of constructing an approximate
$(L,f)$-replacement path covering.
We turn the sampling technique upside down
and combine it with the DO of Thorup and Zwick~\cite{ThorupZ05}.

Instead of sampling the subgraphs directly by removing edges,
we construct them in a hierarchical manner by \emph{adding} connections.
We build a tree\footnote{%
	Again, the sampling trees and fault-tolerant trees are not related.
} 
in which each node is associated with a subset of the edges of $G$,
this set stands for the ``missing'' edges.
We start with the full edge set $E$ in the root,
that is, the graph in the root is empty.
The height of the tree is $h$ and each node has $L^{f/h}$ children.
The associated set of a child node contains any edge of its parent with probability $L^{-1/h}$.
This corresponds to adding any missing edge with probability $1-L^{-1/h}$.
Knowing the missing edges upfront benefits the query algorithm.
At each node starting with the root, if we were to expand all children
in which all failures of $F$ are missing, we would find the suitable subgraphs.
The hierarchical sampling creates some dependencies among the subgraphs associated with 
the leaves of the tree, while the graphs in~\cite{WY13} were independent.
We tackle this issue by always recursing only on one child node
and therefore querying a single leaf.
The process is repeated in several independent trees in order to amplify the success probability.
We prove that there exists a constant $c > 0$ such that
$\Otilde(c^h)$ trees together ensure the property 
we need from an $(L,f)$-replacement path covering w.h.p.
Optimizing the height $h$ gives an~$\Otilde(L^{o(1)})$ query time (assuming constant $f$).

The main challenge is to bring down the size of this construction 
by reducing the number of edges in the graphs associated
with the nodes of the trees.
Thorup and Zwick~\cite{ThorupZ05} devised, for any positive integer $k$,
a $(2k{-}1)$-approximate distance oracle together with a compatible spanner of size $O(kn^{1+1/k})$, 
i.e., the stretched distance returned by the oracle 
is the length of a shortest path in the spanner.
Therefore, we can use the oracles in the leaves of the trees to report distances,
giving a low query time,
and employ the spanners as proxies for the graphs associated with the intermediate nodes.
For this to work, we have to carefully tweak the computation of the spanners
and interleave it with the sampling process 
in order to not blow up the size or stretch too much.
\vspace*{.5em}

\noindent
\textbf{Long Paths.}
We return to the fault-tolerant trees.
By the use of the pivots, we reduced the required number of trees to $\Otilde_{\varepsilon}(n^2/L)$.
But even in the most compact version of FT-trees, this is not enough 
to reach subquadratic space altogether.
The issue is with the representation of expaths as a sequence of $O(f \log n)$ components,
each of which is implicitly represented by its two endpoints.
In~\cite{ChCoFiKa17} this was implemented by storing the original graph distance $d(x,y)$ 
and the predecessor $\pred(x,y)$ of $y$ on the shortest $x$-$y$-path for \emph{all} pairs $x,y$.
This information is used to expand the implicit representation of an expath when needed.
However, the space is again $\Omega(n^2)$.
The key observation to overcome this is that, in our case,
we do not need to encode arbitrary expaths but only those with a particular structure,
e.g., at least one endpoint is a pivot.
This allows us to forgo the need for a quadratic database of all distances. 

We also devise a new procedure to obtain an approximation of $d_{G-F}(s,t)$
by combining the values from the FT-trees with the $f$-DSO for short paths.
Recall that we build one FT-tree for each pair of vertices $(u,v)$ where $u$ or $v$ are pivots.
The main open issue is to find the weight of the edge $\{u,v\}$ in the auxiliary graph $H^F$ 
(see above)
if neither $u$ nor $v$ are pivots and they do not have a short path between them in $G-F$.
Then, w.h.p.\ at least one pivot $b$ hits the $L$-edge prefix of that replacement path.
Therefore, it is sufficient to estimate its length as the sum of an approximation for 
$d_{G-F}^{\le L}(u,b)$ via the $f$-DSO for short paths, 
and an approximation for $d_{G-F}(b,v)$ via the FT-trees.
However, since we do not know the right pivot $b$, we have to scan all of them.
We prove that this results in a stretch of $3+\varepsilon$ and a sublinear query time.

While already being faster than running a shortest-path algorithm 
on a fault-tolerant spanner, this is still not very efficient.
In~\Cref{sec:improved_query_time}, we improve the query time to 
$O_\eps(n^\alpha)$ 
for any constant $0 < \alpha < 1/2$.
We provide an efficient way to check whether the number of pivots in $B$
that are close to $u$ and $v$ in $G-F$ are below the threshold value of $L^{f-1}$
and, if so, find them all. 
If only a few pivots are around $u$ (or $v$),
we can afford to scan them as described above.

The complementary case of many pivots around both endpoints is solved by 
precomputing a set of $\Otilde_\eps(n/L^f)$ \emph{new pivots}, much fewer than before,
and generalizing the FT-trees to granularity $\lambda > 0$.
This ensures that, in any node $\nu$, the first and last $\lambda$ edges 
of the corresponding path $P_\nu$ each form their own segment.
High granularity thus makes the generalized trees much larger.
For comparison, the maximum granularity $\lambda = n$ would unwind \emph{all} the efforts
taken in~\cite{ChCoFiKa17} to reduce their size, as summarized at the beginning of this section.
We can still fit the trees in subquadratic space by building
$FT_\lambda(b,b')$ only for pairs $b,b'$ of new pivots.

The $u$-$v$-distance in $G-F$ in the case of many \emph{original} pivots around $u$ and $v$ is approximated as follows. We compute two \emph{new} pivots $b_u,b_v$,
with $b_u$ close to $u$ in $G-F$ and $b_v$ close to $v$.
The approximate length of the shortest path from $u$ to $v$ in $G-F$ is computed by the overall sum of (\textit{i}) an approximation of the distance from $u$ to $b_u$ in $G-F$, (\textit{ii}) an approximation of the distance from $b_u$ to $b_v$ in $G-F$ computed by querying $FT_\lambda(b_u,b_v)$, and (\textit{iii}) an approximation of the distance from $b_v$ to $v$ in $G-F$.
We make sure to have a granularity $\lambda \le L$ so that we can obtain the terms (\textit{i}) and (\textit{iii})
from our $f$-DSO for short paths.

\section{Preliminaries}
\label{sec:prelims}

We let $G = (V,E)$ denote the undirected and unweighted base graph with $n$ vertices and $m$ edges.
We tacitly assume $m = \Omega(n)$.
For any undirected \mbox{(multi-)graph} $H$, which may differ from the input $G$,
we denote by $V(H)$ and $E(H)$ the set of its vertices and edges, respectively.
Let $P$ be a path in $H$ from a vertex $s \in V(H)$ to $t \in V(H)$,
we say that $P$ is an \emph{$s$-$t$-path} in $H$. 
We denote by $|P| = |E(P)|$ the \emph{length} of $P$.
For vertices $u,v \in V(P)$, we let $P[u..v]$ denote the subpath of $P$ from $u$ to $v$.
Let $P = (u_1, \dots, u_i)$ and $Q = (v_1, \dots, v_j)$ be two paths in $H$.
Their \emph{concatenation} is $P \circ Q = (u_1, \dots, u_i, v_1, \dots, v_j)$,
which is well-defined if $u_i = v_1$ or $\{u_i,v_1\} \in E(H)$.
For $s,t \in V(H)$, the \emph{distance} $d_H(s,t)$ 
is the minimum length of any $s$-$t$-path in $H$;
if $s$ and $t$ are disconnected, we set $d_H(s,t) =+ \infty$.
When talking about the base graph $G$, we drop the subscripts.

A \emph{spanning} subgraph of a graph $H$
is one with the same vertex set as $H$ but possibly any subset of its edges.
This should not be confused with a spanner.
A \emph{spanner of stretch}~$\sigma \ge 1$, or \emph{$\sigma$-spanner},
is a spanning subgraph $S \subseteq H$
such that additionally for any two vertices $s,t \in V(S) = V(H)$, 
it holds that $d_H(s,t) \le d_S(s,t) \le \sigma \cdot d_H(s,t)$.
A \emph{distance oracle} (DO) for $H$ is a data structure 
that reports, upon query $(s,t)$, the distance $d_H(s,t)$.
It has \emph{stretch} $\sigma \ge 1$, or is $\sigma$\emph{-approximate},
if the reported value $\widehat{d}(s,t)$
satisfies $d_H(s,t) \le \widehat{d}(s,t) \le \sigma \cdot d_H(s,t)$.

For a set $F \subseteq E$ of edges,
let $G{-}F$ be the graph obtained from $G$ by removing all edges in $F$.
For any two $s,t \in V$, a \emph{replacement path} $P(s,t,F)$ 
is a shortest path from $s$ to $t$ in $G{-}F$. 
Its length $d_{G-F}(s,t)$ is the \emph{replacement distance}.
Let $L$ be a positive integer.
We call a path in (a subgraph of) $G$ \emph{short} if it has at most $L$ edges,
and \emph{long} otherwise.
Let $d_{G-F}^{\le L}(s,t)$ be the minimum length
of any short $s$-$t$-paths in $G-F$, or $+\infty$ if no such path exists.

For a positive integer $f$, an \mbox{\emph{$f$-distance sensitivity oracle}} (DSO)
answers queries $(s,t,F)$ with $|F| \le f$ by reporting the replacement distance $d_{G-F}(s,t)$.
The stretch of a DSO is defined as for DOs.
The maximum number $f$ of supported failures is called the \emph{sensitivity}.
We measure the space complexity of any data structure in the number of $O(\log n)$-bit machine words.
The size of the input graph $G$ does not count against the space, 
unless it is stored explicitly.

\section{Handling Short Paths}
\label{sec:short_paths}

We develop here our $(2k{-}1)$-approximate solution for short replacement paths.
It will later be used for the general distance sensitivity oracle.
We first review (and slightly modify) 
the DO and spanner of Thorup and Zwick~\cite{ThorupZ05}
to an extent that is needed to present our own construction.

\subsection{The Distance Oracle and Spanner of Thorup and Zwick}
\label{subsec:short_paths_TZ}

For any positive integer $k$, 
Thorup and Zwick~\cite{ThorupZ05} devised a DO
that is computable in time $\Otilde(kmn^{1/k})$, has size $O(kn^{1+1/k})$, query time $O(k)$,
and a stretch of $2k-1$.
Their stretch-space tradeoff is essentially optimal for sufficiently dense graphs,
assuming the Erdős girth conjecture~\cite{ThorupZ05}.
For sparse graphs, better constructions are known~\cite{Abraham11AffineStretch,PatrascuRoditty14BeyondThorupZwick,Patrascu11NewInfinity},
including subquadratic-space DOs with a stretch less than $2$~\cite{Agarwal14SpaceStretchTimeTradeoffDOs,AgarwalBrightenGodfrey13DOsStretchLessThan2}.

We first review the Thorup and Zwick construction before discussing our changes.
First, a family of vertex subsets 
$V = X_0 \supseteq X_1\supseteq \cdots \supseteq X_{k-1}  \supseteq X_k = \emptyset$ is computed.
Each $X_i$ is obtained by sampling the elements of $X_{i-1}$ independently with probability $n^{-1/k}$.
We keep this family fixed and
apply the construction to a variety of subgraphs of $G$.

\setcounter{AlgoLine}{0}
\begin{algorithm}[H]
    $w\gets s$\; 
    $i\gets 0$\; 
    \While{$w\notin \bigcup_{j=0}^{k-1}X_{j,H}(t)$}
    {
        $i\gets i+1$\;
        $(s,t)\gets (t,s)$\;  
        $w\gets p_{i,H}(s)$\;
    }
    \Return $d_H(s,w)+d_H(w,t)$\;
\caption{Original query algorithm \cite{ThorupZ05} of the distance oracle for the pair $(s,t)$.}
\label{alg:query_TZ}
\end{algorithm}

Let $H$ be such a subgraph for which the oracle needs to be computed.
For any $v\in V$ and $0\le i<k$, 
let $p_{i,H}(v)$ be the closest vertex\footnote{%
  We have $p_{i,H}(v) = v$  for all $i$ small enough so that $X_{i}$ still contains $v$.
}
to $v$ in $X_i$ in the graph $H$,
ties are broken in favor of the vertex with smaller label.
The distances from $v$ to all elements in
\begin{equation*}
	X_{i,H}(v)=\{x\in X_i \mid d_H(v,x) < \min_{y\in X_{i+1}}d_H(v,y)\} \cup \{p_{i,H}(v)\}
\end{equation*}
are stored in a hash table.
In other words, $X_{i,H}(v)$ contains those vertices of $X_i{\setminus}X_{i+1}$ that are closer to $v$
than any vertex of $X_{i+1}$.
This completes the construction of the DO for $H$.
Note that while the $X_i$ are fixed,
the sets $X_{i,H}(v)$ and vertices $p_{i,H}(v)$ may differ for the various subgraphs $H \subseteq G$
as the underlying distance function $d_H$ changes.

The oracle is accompanied by a $(2k{-}1)$-spanner with $O(k n^{1+1/k})$ edges.
It stores all those edges of $H$ that lie on
a shortest path between $v$ and a vertex in $\bigcup_{0\le i<k}X_{i,H}(v)$,
again ties between shortest paths are broken using the edge labels.

\Cref{alg:query_TZ} describes how the oracle handles the query $(s,t)$.
The returned distance can be shown to overestimate $d_H(s,t)$ by at most a factor $2k{-}1$~\cite[Lemma~3.3]{ThorupZ05}.
We would like to use this construction in many different subgraphs $H,G'$
with $H \subseteq G' \subseteq G$
and have it satisfy a certain \emph{inheritance property}.
That means, if some approximate shortest $s$-$t$-path in the larger graph $G'$
only uses edges already present in the smaller graph $H$,
then the same path should be used in both $H$ and $G'$ 
to compute the estimate reported by the oracle.
For a precise statement see \Cref{lem:new_query_inheritance}.
This is not the case for \Cref{alg:query_TZ}.
It computes the vertex $p_{i,H}(s)$ with \emph{smallest index} $i$
that is contained in $\bigcup_{j=0}^{k-1} X_{j,H}(t)$
(respectively, the $p_{i,H}(t)$ with the smallest index in~$\bigcup_{j=0}^{k-1} X_{j,H}(s)$).
As more edges are added to get from $H$ to $G'$,
both the sets $X_{j,G'}$ and the first vertex $p_{i,G'}(s)$ (respectively, $p_{i,G'}(t)$) 
satisfying the inclusion relations may change.

We instead use a slightly modified version as presented in~\Cref{alg:query_new}. 
It considers the vertices $p_{i,H}(s)$ (respectively, $p_{i,H}(t)$) for \emph{all indices} $i$,
checks whether they satisfy the inclusion relations,
and chooses the one minimizing the combined distance.
Observe that the estimate $\widehat{d}$ produced by our version is at most the value returned 
by the original one and at least the actual distance between $s$ and $t$.
As before, for any $s$ and $t$, 
the path corresponding to the new estimate is a concatenation of 
at most two original shortest paths in $H$.
The interconnecting vertex is either $p_{i,H}(s)$ or $p_{i,H}(t)$ for some $i$,
we denote it as $u_{s,t,H}$, and the $(2k{-}1)$-approximate
shortest path as $P_{s,t,H}$.
We show next that the adapted query algorithm has the inheritance property.

\setcounter{AlgoLine}{0}
\begin{algorithm}[H]
    $\widehat{d} \gets \infty$\; 
    \For{$i=0$ \textup{\textbf{to}} $k-1$}
    {
        \If{$p_{i}(s)\in \bigcup_{j=0}^{k-1}X_{j,H}(t)$}
          {$\widehat{d} \gets \min\!\big\{\nwspace \widehat{d}, \
            d_H(s,p_{i}(s)) + d_H(p_{i}(s),t) \big\}$}
        \If{$p_{i}(t)\in \bigcup_{j=0}^{k-1}X_{j,H}(s)$}
          {$\widehat{d} \gets \min\!\big\{ \nwspace \widehat{d}, \
            d_H(t,p_{i}(t))+d_H(p_{i}(t),s) \big\}$}
    }
    \Return $\widehat{d}$\;
    \caption{Modified query algorithm of the distance oracle for the pair $(s,t)$.}
    \label{alg:query_new}
\end{algorithm}

 \begin{lemma}[Inheritance property]
\label{lem:new_query_inheritance}
  Let $H \subseteq G' \subseteq G$ be two spanning subgraphs of $G$, $s,t \in V$ two vertices,
  and $P_{s,t,G'}$ the approximate shortest path underlying the value returned by the
  (modified) distance oracle for $G'$.
  If $P_{s,t,G'}$ also exists in $H$, then $P_{s,t,H} = P_{s,t,G'}$,
  Moreover, the oracle for $H$ returns $|P_{s,t,G'}|$ upon query $(s,t)$.
\end{lemma}

\begin{proof}
  Recall that $P_{s,t,G'}$ is a concatenation of two shortest paths in $G'$,
  say, $P(s,u)$ and $P(u,t)$, where $u = u_{s,t,G'}$ is the interconnecting
  vertex in $\bigcup_{j<k}X_{j,G'}(s) \nwspace \cup \nwspace \bigcup_{j<k}X_{j,G'}(t)$  
  that minimizes the sum of distances $d_{G'}(s,u)+d_{G'}(u,t)$.
  Without losing generality, we have $u = p_{i,G'}(s)$ for some $0 \le i < k$;
  otherwise, we swap the roles of $s$ and $t$.
  Let $0 \le j < k$ be such that $u \in X_{j,G'}(t)$.
  
  For any spanning subgraph $H \subseteq G'$ that contains the path $P_{s,t,G'}$,
  it holds that $u=p_{i,H}(s)$ and $u \in X_{j,H}(t)$.
  Here, we use that the tie-breaking for the $p_{i,H}(s)$ does not depend on the edge set of $H$.
  Moreover, the shortest $s$-$u$-path and $u$-$t$-path in the spanner for $H$
  are the same as in $G$, that is, $P(s,u)$ and $P(u,t)$.
  As a result, we have $u = u_{s,t,H}$ and $P_{s,t,G'} = P_{s,t,H}$.
  The second assertion of the lemma follows from $d_H(s,u) = |P(s,u)|$ and $d_H(u,t) = |P(u,t)|$.
\end{proof}

\subsection{Tree Sampling}
\label{subsec:short_paths_tree_sampling}

We present our fault-tolerant oracle construction for short paths.
Recall that a path in $G$ is short if it has at most $L$ edges,
and that $d_{G-F}^{\le L}(s,t)$ is the minimum distance over short $s$-$t$-paths in $G-F$.
Note that, while we assume $f$ and $k$ to be constants, $L$ may depend on $m$ and $n$.
We prove \Cref{thm:oracle_short_paths} in the remainder of the section.
   
We first compute the vertex sets $X_0, \dots,X_k$.
Define {$h=\sqrt{f\ln L}$}, {$K = \big\lceil((2k{-}1)L)^{{f}/{h}}\big\rceil $, $p = K^{-1/f}$, and $I = C \cdot 11^{h} \ln n$} 
for some sufficiently large constant $C > 0$ (independent of $f$ and~$k$).
We build $I$ rooted trees $T_1, \dots, T_I$, each of height $h$,
such that any internal node has {$K$} children.
For the following description, we fix some tree $T_i$ and use $x$ to denote a node in $T_i$.
Let $y$ be the parent of $x$ in case $x$ is not the root.
We associate with each $x$ a subset of edges $A_x \subseteq E$ 
and a spanning subgraph $S_x \subseteq G$ in recursive fashion.
For the root of $T_i$, set $A_x = E$;
otherwise $A_x$ is obtained by selecting each edge of $A_y$ independently with probability $p$.
The random choices here and everywhere else
are made independently of all other choices.

Let $r$ be the depth of $x$ in $T_i$ (where the root has depth $r = 0$).
Define $J_r = 4 \cdot K^{h-r}$ for $r < h$, and $J_h = 1$.
The graph $S_x$ is constructed in $J_r$ rounds.
In each round, we sample a subset $A \subseteq A_x$ 
by independently selecting each edge with probability $p^{h-r}$.
We then compute the Thorup-Zwick spanner of $S_y\,{-}\,A$
using the family $X_0,\dots,X_k$.
Slightly abusing notation, if $x$ is the root, we define $S_y = G$ here.
We set $S_x$ to be the union of all those spanners.
Note that, for a leaf $x$ at depth $r = h$, then $A = A_x$ with probability $1$,
so indeed only $J_h = 1$ iteration is needed.

For each node, we store a dictionary of the edge set $E(S_x)$ 
and (except for the root) $A_x \cap E(S_y)$.
We use the static construction of Hagerup, Bro Miltersen, and Pagh~\cite{HagerupMiltersenPagh01DeterministicDictionaries}
that, for a set $M$, has space $O(|M|)$, preprocessing time $\Otilde(|M|)$,
and query time $O(1)$.
For each leaf of a tree, we store the (modified) distance oracle $D_x$.
At depth $0 \le r \le h$, 
the tree $T_i$ has $K^r$ nodes.
The largest dictionary at depth $r$ is for $A_x \cap E(S_y)$
of size $O(J_{r-1} \cdot kn^{1+1/k}) = O(K^{h-r+1} \nwspace n^{1+1/k})$
(using that $k$ is constant).
Due to $K = O((2k{-}1)^{f/h} L^{f/h})$ and $h = \sqrt{f \ln n}$,
we have $K^{h+1} = O(L^{f+o(1)})$ (using that $f$ is constant as well).
In total, our data structure requires 
$O(I \cdot h\cdot K^{h+1} \nwspace n^{1+1/k}) = \Otilde(L^{f+o(1)} \nwspace n^{1+1/k})$ space
and can be preprocessed in time 
$O(I \cdot h\cdot K^{h+1} (kmn^{1/k} + kn^{1+1/k})) = \Otilde(L^{f+o(1)} \nwspace mn^{1/k})$.

\subsection{Query Algorithm}
\label{subsec:short_paths_query}

\setcounter{AlgoLine}{0}
\begin{algorithm}[H]
    $\widehat{d} \gets \infty$\;
\For{$i=1$ \textbf{to} $I$}
{
    $y\gets$ root of $T_i$\;
    \While{$y$ is not leaf}
    {
        \ForEach{child $x$ of $y$}
        {
            \If{$F \cap E(S_y) \subseteq A_x$}
            {
                $y \gets x$\; 
                \textbf{continue} while-loop\; 
            }
        
        }
        \textbf{break} while-loop\;
    }
    \lIf{$y$ is leaf}
    {$\widehat{d} \gets \min\!\big\{ \nwspace \widehat{d}, \ D_y(s,t) \big\}$}
}
\Return $\widehat{d}$\;
\caption{Algorithm to answer the query $(s,t,F)$. $D_y$ is the distance oracle associated with the leaf $y$.}
\label{alg:query_short_paths}
\end{algorithm}

\Cref{alg:query_short_paths} presents the query algorithm to report approximate distances.
Fix a query $(s,t,F)$ where $s,t \in V$ are two vertices 
and $F \subseteq E$ is a set of at most $f$ edges.
For each of the $I$ trees, we start at the root and recurse on an arbitrary child,
computed in the inner for-loop, that satisfies $F \cap E(S_y) \subseteq A_x$,
where $y$ is parent of $x$.
Note that the set $A_x$ is not stored as it may be too large.
(We have $|A_x| = m$ in the root.)
The test is equivalent to $F \cap E(S_y) \subseteq A_x \cap E(S_y)$
and can be performed in time $O(f)$ using the stored dictionaries.
If at some point no child satisfies the condition, the algorithm resumes with the next tree.
Once a leaf $y$ is reached, we query the associated (modified) distance oracle $D_y$
with the pair $(s,t)$.
Finally, the algorithm returns the minimum of all oracle answers.

We set $h = \sqrt{f \ln L}$.
This gives $I = \Otilde(11^{\sqrt{f \ln L}}) = \Otilde(L^{o(1)})$ sampling trees
and $K = ((2k{-}1)L)^{\sqrt{f}/\sqrt{\ln L}} = \Otilde(L^{o(1)})$ children per node.
The total query time is $I \cdot O(fh K + k) = \Otilde(L^{o(1)})$.

We are left to prove correctness.
That means, we claim that w.h.p.\ the returned estimate is at least as large as 
the replacement distance $d(s,t,F)$ and,
if $s$ and $t$ are joined by a short path in $G\,{-}\,F$,
then this estimate is also at most $(2k{-}1) \nwspace d_{G-F}^{\le L}(s,t)$.
Consider the Thorup-Zwick spanner for $G-F$
and in it the approximate shortest path $P_{s,t,G-F}$ 
(as defined ahead of~\Cref{lem:new_query_inheritance}).
If $s$ and $t$ have a short path in $G-F$, then $P_{s,t,G-F}$ has at most $(2k{-}1)L$ edges.

Let $x$ be a node at depth $r$ in the tree $T_i$ and let $S_y$ be the spanner associated to its parent
(or $S_y = G$ if $x$ is the root).
We say $x$ is \emph{well-behaved} if it satisfies the following three properties.
\begin{enumerate}[(1)]
\item \label[prop]{prop:F_missing} $F\cap E(S_y)\subseteq A_x$.
\item \label[prop]{prop:P_small} Either $x$ is a root or $|E(P_{s,t,G-F})\cap A_x| < K^{\frac{h-r}{f}}$.
\item \label[prop]{prop:P_spanner} The path $P_{s,t,G-F}$ is contained in $S_x$.
\end{enumerate}

Our query algorithm follows a path from the root to a leaf node
such that at each node~\Cref{prop:F_missing} is satisfied. 
We show in the following lemma that any child $x$ of a well-behaved node $y$
that fulfills~\Cref{prop:F_missing} is itself well-behaved with constant probability.

\begin{restatable}{lemma}{approxparenttochild}
\label{lemma:approx_parent_to_child}
  The following statements hold for any non-leaf node $y$ in the tree $T_i$.
  \begin{enumerate}[(i)]
    \item If $y$ satisfies~\Cref{prop:F_missing}, then
      with probability at least $1-\frac{1}{e}$ there exists a child of $y$ 
      that satisfies~\Cref{prop:F_missing}.
    \item If $y$ satisfies~\Cref{prop:P_small}, then any child of $y$
      satisfies~\Cref{prop:P_small} with probability at least $\frac{1}{4}$.
    \item If $y$ is well-behaved
      and a child $x$ of $y$ satisfies~\Cref{prop:F_missing,prop:P_small}, 
      then the probability of $x$ being well-behaved is at least $1-\frac{1}{e}$.
  \end{enumerate}
  The root of $T_i$ is well-behaved with probability at least $1-\frac{1}{e}$.
\end{restatable}

\begin{proof}
  Assume that node $y$ satisfies~\Cref{prop:F_missing},
  that means $F\cap E(S_y)\subseteq A_y$.
  Let $x_1,\ldots,x_K$ be the child nodes of $y$.
  Each edge of $A_y$ is sampled into $A_{x_j}$ with probability $p$.
  The probability that there exists some child $x_j$ with
  $F\cap E(S_y) \subseteq A_{x_j}$ is therefore
  \begin{multline*}
    1-\prod_{j=1}^{K}\Pb\!\Big[F\cap E(S_y) \nsubseteq A_{x_j}\Big]
      = 1-\prod_{i=1}^{K} \left(1-p^{|F\cap E(S_y)|} \right)\\
      = 1-\left(1-\frac{1}{K^{|F\cap E(S_y)|/f}} \right)^K
      \geq  1-\left(1-\frac{1}{K}\right)^{K} 
       \geq 1-\frac{1}{e}.
  \end{multline*} 

  For the second statement, let $r{-}1$ be the depth of $y$ in $T_i$.
  Recall that the path $P = P_{s,t,G-F}$ has at most $(2k{-}1)L$ edges.
  By our assumption of $y$ satisfying~\Cref{prop:P_small}, 
  at most $K^{\frac{h-r+1}{f}}$ of those are in $A_y$. 
  Let $x$ be a child of $y$.
  We first analyze the case that $x$ is a leaf, that is, $r=h$.
  \begin{equation*}
    \Pb\!\Big[E(P)\cap A_x=\emptyset \Big] 
      = (1-p)^{|E(P)\cap A_y|}
       = \left( 1- \frac{1}{K^{1/f}} \right)^{|E(P)\cap A_y|}
      \geq \left( 1- \frac{1}{K^{1/f}}\right)^{K^{1/f}}
       \geq \frac{1}{4}.
  \end{equation*}
  
  Now suppose $r < h$.
  Define the random variable $M = |E(P) \cap A_x|$ to be the number of edges of the path $P$
  that are contained in $A_x$.
  Since $A_x$ is obtained by sampling edges from $A_y$ independently with probability $p$,
  the variable $M$ is binomially distributed with parameters $|E(P) \cap A_y|$ and $p$.
  The parent $y$ satisfies~\Cref{prop:P_small},
  which implies  $\Ev[M] \le p K^{\frac{h-r+1}{f}} = K^{\frac{h-r}{f}}$.
  By the central limit theorem, we have
  $\Pb\!\big[ M \geq K^{\frac{h-r}{f}} \big]
      \le \Pb\!\big[ M \geq \Ev[M] \nwspace \big] \le \frac{3}{4}$.
  In both cases, we see that the child node $x$ also satisfies~\Cref{prop:P_small} 
  with probability at least $\frac{1}{4}$.
  
  We now turn to the third clause of the lemma. 
  Suppose $y$ is well-behaved and its child node $x$ fulfills~\Cref{prop:F_missing,prop:P_small}.
  If $x$ is a leaf, it is well-behaved deterministically.
  Indeed, in this case, subgraph $S_x$ is just the Thorup-Zwick spanner for $S_y-A_x$.
  \Cref{prop:F_missing} for $x$ means that $S_y-A_x$ doesn't contain edges of $F$.
  Likewise, $x$ satisfying~\Cref{prop:P_small} with $r = h$ and
  together with $y$ satisfying~\Cref{prop:P_spanner}
  shows that the path $P$ is contained in $S_y-A_x$.
  The inheritance property (\Cref{lem:new_query_inheritance}) applied to
  $G' = G-F$ and $H = S_y-A_x$ gives that $P$ is contained in~$S_x$.
  
  For $r < h$, the argument goes through only with a certain probability.
  The graph~$S_x$ is obtained in $J_r = 4 K^{h-r}$ iterations;
  in each iteration, a subset $A \subseteq A_x$ is sampled
  by selecting each edge with probability $p^{h-r}$,
  and the spanner $H_A$ of $S_y{-}A$ is computed.
  $S_x$ is the union of all $4 K^{h-r}$ $H_A$'s.
  We estimate the probability that the path 
  $P$ exists in $S_y{-}A$ 
  and no failing edge of $F$ is in $S_y{-}A$.
  By inheritance to $H_A$ and taking the union,
  this will imply that $P$ lies in~$S_x$.
  
  We first claim that $\Pb[F\cap (E(S_y){\setminus}A)=\emptyset] = p^{|F\cap E(S_y)| \cdot (h-r)}$.
  To see this, note that~\Cref{prop:F_missing} holding for $x$ means that $F\cap E(S_y)\subseteq A_x$.
  No failure from $F$ is in $S_y{-}A$ if and only if all the edges in $F\cap E(S_y)$
  are chosen for $A$. 
  Our second claim is $\Pb[E(P) \subseteq E(S_y){\setminus}A] = (1-p^{h-r})^{|E(P) \cap A_x|}$.
  It holds that $E(P)\subseteq E(S_y)$ since $y$ is well-behaved (\Cref{prop:P_spanner}).
  Thus, $E(P)\subseteq E(S_y){\setminus}A$ is true
  if and only if none of the edges in $E(P)\cap A_x$ are selected in $A$.
  
  Using the independence of the events and \Cref{prop:P_small} of the node $x$,
  we arrive at
  \begin{multline*}
    \Pb\!\Big[ (F \cap (E(S_y){\setminus}A)=\emptyset) \,\wedge\,
      (E(P) \subseteq E(S_y){\setminus}A) \Big]
      = p^{|F\cap E(S_y)| \cdot (h-r)} \cdot (1-p^{h-r})^{|E(P) \cap A_x|}\\[-.33em]
      \ge  p^{f(h-r)} \cdot \Big(1-p^{h-r} \Big)^{K^{\frac{h-r}{f}}}
       =  \frac{1}{K^{h-r}} \cdot \left( 1-\frac{1}{K^{\frac{h-r}{f}}} \right)^{K^{\frac{h-r}{f}}}
      \ge \frac{1}{4 \cdot K^{h-r}}.
  \end{multline*}
  Iterating this $J_r$ times gives
  \begin{equation*}
    \Pb\!\big[E(P) \subseteq E(S_x)\big] 
      \ge 1-\left(1-\frac{1}{4K^{h-r}}\right)^{J_r}
      = 1-\left(1-\frac{1}{4K^{h-r}}\right)^{4K^{h-r}}
      \ge 1-\frac{1}{e}.
  \end{equation*}
  
  The assertion about the root follows
  by observing that, for the purpose of this proof, the original graph $G$ is the ``parent''
  of the root, meaning that $A_x = E$ and $S_y = G$ both hold.
\end{proof}

The next lemma shows that the distance oracle computed for a well-behaved leaf
reports a $(2k{-}1)$-approximation of the distance in $G-F$ for short paths. 

\begin{lemma}
\label{lem:approx_well-behaved_child}
  Let $s,t \in V$ be two vertices and $F \subseteq E$ a set of at most $f$ edges.
  Furthermore, let $x$ be a leaf in $T_i$ and $D_x$ be the (modified) distance oracle associated with $x$.
  If $x$ satisfies~\Cref{prop:F_missing} with respect to $F$, then $D_x(s,t) \ge d(s,t,F)$.
  Moreover, if $x$ is well-behaved with respect to the approximate shortest path $P_{s,t,G-F}$,
  then $D_x(s,t) \le (2k{-}1) \, d(s,t,F)$.
\end{lemma}

\begin{proof}
As $x$ is a leaf node, $S_x$ is the spanner of the graph $S_y-A_x$
and $D_x$ reports the distances in $S_x$.
By~\Cref{prop:F_missing}, we have $F\cap E(S_y)\subseteq A_x$ 
whence $S_x \subseteq S_y-A_x \subseteq G-F$.
This implies that~$D_x(s,t) = d_{S_x}(s,t) \ge d_{G-F}(s,t) = d(s,t,F)$.
If $x$ is even well-behaved then, by~\Cref{prop:P_spanner}, the path $P_{s,t,G-F}$ lies in $S_x$
and thus by inheritance, $D_x(s,t) \le |P_{s,t,G-F}| \le (2k{-}1) \cdot d(s,t,F)$.
\end{proof}

Our algorithm only ever queries leaves that fulfill~\Cref{prop:F_missing},
it therefore never underestimates the distance $d(s,t,F)$.
Additionally assume that $s$ and $t$ are connected in $G{-}F$ via a path with at most $L$ edges.
To complete the proof of~\Cref{thm:oracle_short_paths},
we need to show that, under this condition and with high probability over all queries,
our algorithm queries at least one well-behaved leaf.
If there is a short $s$-$t$-path in $G{-}F$ then $P_{s,t,G-F}$ has at most $(2k{-}1)L$ edges.
\Cref{lemma:approx_parent_to_child} shows that the root of each tree $T_i$
is well-behaved with probability $1-\frac{1}{e}$, 
and that in each stage the query algorithm finds a well-behaved child node with constant probability.
More precisely, we arrive at a well-behaved leaf with probability at least
$(1-\tfrac{1}{e}) \cdot \Big( (1-\tfrac{1}{e})^2 \nwspace \tfrac{1}{4} \Big)^h
  \ge \tfrac{1}{2} \cdot 11^{-h}$.
Since there are $I = c \cdot 11^h \ln n$ independent trees,
the query algorithms fails for any fixed query with probability at most 
$(1-\frac{1}{2 \cdot 11^h})^I \le n^{-c/2}$.
We choose the constant $c > 0$ large enough to ensure a high success probability over all 
$O(n^2 m^f) = O(n^{2+2f})$ possible queries.

\section{Sublinear Query Time for Long Paths}
\label{sec:sublinear_query_time}

Let $0 < \alpha < 1/2$ be a constant,
where the approximation parameter $\varepsilon > 0$ may depend on $m$ and $n$.
As a warm-up, we construct a
distance sensitivity oracle with the same stretch and space as in \Cref{thm:oracle_long_paths},
but only a sublinear query time of the form $O_{\varepsilon}(n^{1-g(\alpha,f)})$,
for some function $g$.
In \Cref{sec:improved_query_time}, we then show how to reduce the query time to $\Otilde_\eps(n^\alpha)$.
The intermediate solution serves to highlight many of the key ideas 
needed to implement the classical FT-trees in subquadratic space,
but does not yet involve the granularity $\lambda$. 
Recall that we assume that, for every two vertices $u$ and $v$ of $G$,
there is a unique shortest path from $u$ to $v$ in $G$.
Since the short replacement paths are handled by~\Cref{thm:oracle_short_paths},
we focus on long paths.
The structure of this section is as follows.
We first describe the interface of an abstract data structure \emph{FT}
and show how to use it to get a $(3{+}\eps)$-approximation of the replacement distances.
We then implement the data structure \emph{FT} using FT-trees.

\begin{lemma}
\label{lem:long_paths_sublinear_query}
  Let $f$ be a positive integer and $0 < \alpha < 1/2$ a constant.
  For any  undirected, unweighted graph
  with unique shortest paths
  and any $\eps > 0$,
  there exists a $(3{+}\eps)$-approximate \mbox{$f$-DSO}
  that takes space $\Otilde(n^{2-\frac{\alpha}{f+1}}) \cdot O(\log n/\varepsilon)^{f+1}$,
  has query time $n^{1-\frac{\alpha}{f+1}+o(1)}/\varepsilon$,
  and preprocessing time 
  $\Otilde(mn^{2-\frac{\alpha}{f+1}}) \cdot O(\log n/\varepsilon)^{f}$.
\end{lemma}

\subsection{Trapezoids and Expaths}
\label{subsec:sublinear_trapezoids}

For the interface of \emph{FT}, 
we need a bit of terminology from the work by Chechik et al.~\cite{ChCoFiKa17}.
Recall the high-level description of the original FT-trees in \Cref{sec:overview}.
We now make precise what we mean by all failures in $F$ being ``far away'' from a given path.
Let $0< \eps < 3$; moreover, we assume it to be bounded away from $3$.
(Recall that $\varepsilon$ may depend on the input.)
We use $V(F)$ for the set of endpoints of failing edges.

\begin{definition}[$\frac{\eps}{9}$-trapezoid]
\label{def:trapezoid}
  Let $F \subseteq E$ a set of edges, $u,v \in V$,
  and $P$ a $u$-$v$-path in $G-F$.
  The $\frac{\eps}{9}$\emph{-trapezoid} around $P$ in $G-F$ is
  \begin{equation*}
    \tr^{\eps/9}_{G-F}(P) = \big\lbrace\nwspace z \in V{\setminus}\{u,v\} \mid
      \exists y \in V(P) \colon d_{G-F}(y,z)
      \le \frac{\eps}{9}
        \cdot \min(\nwspace \length{P[u..y]}, \length{P[y..v]} \nwspace) \big\rbrace. 
  \end{equation*}
  \noindent
  $P$ is \emph{far away}\footnote{%
      \Cref{def:trapezoid} relaxes the notion of ``far away'' compared to \cite{ChCoFiKa17}
      in that we allow the case
      $\tr^{\eps/9}_{G-F}(P) \cap \{s,t\} \neq \emptyset$
      if $s,t \notin V(F)$.
      This makes the definition independent of the vertices $s$ and $t$ in the query.
      The proof of \Cref{lem:not_too_far_off} remains the same
      using a vertex $z \in V(F)$ instead of $z \in V(H^F) = V(F) \cup \{s,t\}$.
  }
  from $F$ if it exists in $G-F$
  and $\tr^{\eps/9}_{G-F}(P) \cap V(F) = \emptyset$.
\end{definition}

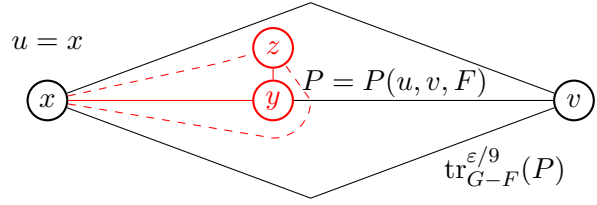
\begin{figure}
\centering
\begin{tikzpicture}
\begin{scope}[scale=1.5]
    \coordinate(left) at (0, 0);
    \coordinate(right) at (7, 0);
    \coordinate(up) at (3.5, 1.3);
    \coordinate(down) at (3.5, -1.3);
    
	\node[vert] (l) [label={above:{$u=x$}}] at (left) {$x$};
	\node[vert] (r) at (right) {$v$};
	\node[vert] (y) [red] at (3,0) {$y$};
	\node[vert] (z) [red] at (3,0.7) {$z$};
	\draw[-]  (l) -- (up) -- (r);
	\draw[-]  (l) -- (down) to node[xshift=2.5em, yshift=-0.5em] {$\tr^{\eps/9}_{G-F}(P)$} (r);
	\draw[-] [red] (l) -- (y) -- (z);
	\draw[-] (y) -- node[xshift=-1.0em, yshift=0.6em] {$P = P(u,v,F)$} (r);
	\draw[-] [red,dashed] (l) -- (3, -0.5) to[in=270, out=0] (3.5, 0) -- (z);
	\draw[-] [red,dashed] (l) -- (2.3, 0.5) -- (z);
 \end{scope}
\end{tikzpicture}
\caption{ A visualization of the trapezoid $\tr^{\eps/9}_{G-F}(P)$ 
	in \Cref{lem:not_too_far_off} for the case $u=x$.
	The vertices $u,v$ are endpoints of failing edges in $F$ or the query vertices $s$ or $t$,
	they are not part of $\tr^{\eps/9}_{G-F}(P)$.
	Vertex $y$ lies on the path $P$ and vertex $z$ is in $V(F)$.
	The replacement path from $y$ to $z$ has length 
	at most $\frac{\varepsilon}{9} \nwspace d_{G-F}(u,y)$.
	The smaller trapezoid around $P[u..y] \circ P(y,z,F)$ (red dashed line) 
	does not contain any vertex from $V(F)$.}
\label{fig:trapezoid}
\end{figure}

The endpoints $u,v$ of $P$ are removed from the trapezoid
to exclude trivialities when applying it to paths
between vertices contained in the failing edges.
Finally, note that, due to $\eps/9 < 1$, 
the distance from $u$ to any vertex in the trapezoid is strictly smaller than $d_{G-F}(u,v)$
(by symmetry, this also holds for $v$).
The idea is that either the path $P$
is already far away from all failures, 
or we can reach our destination via a vertex $z \in \tr^{\eps/9}_{G-F}(P) \cap V(F)$
such that the shortest $u$-$z$-path in $G-F$ is far away from $F$
and only a slight detour.
An illustration is given in \Cref{fig:trapezoid}.

\begin{lemma}[Lemma~2.6 in~\cite{ChCoFiKa17}]
\label{lem:not_too_far_off}
  Let $u,v \in V(F) \cup \{s,t\}$ be endpoints of failing edges or query vertices
  and $P = P(u,v,F)$ their replacement path.
  If $\tr^{\eps/9}_{G-F}(P) \cap V(F) \neq \emptyset$,
  then there are vertices $x \in \{u,v\}$, $y \in V(P)$,
  and $z \in \tr^{\eps/9}_{G-F}(P) \cap V(F)$ satisfying the following statements.
  \begin{enumerate}[(i)]
    \item $\length{P[x..y]} \le \length{P}/2$;
    \item $d_{G-F}(y,z) \le \frac{\eps}{9} \cdot d_{G-F}(x,y)$;
    \item $\tr^{\eps/9}_{G-F}(P[x..y] \circ P(y,z,F)) \cap V(F) = \emptyset$.
  \end{enumerate}
  Thus, the path $P[x..y] \circ P(y,z,F)$ is far away from $F$
  and has length at most $(1+ \tfrac{\eps}{9}) \cdot d_{G-F}(x,y)$.
\end{lemma}

We now turn to expaths.
Afek et al.~\cite{Afek02RestorationbyPathConcatenation_journal} 
showed that shortest paths in $G-F$ are $f$-\emph{decomposable},
that is, each of them is obtained by concatenating at most $f+1$ shortest paths in~$G$ 
(for weighted $G$ those shortest paths may be interleaved with up to $f$ edges).
One would like to represent replacement paths by the $O(f)$ endpoints of those shortest paths (and edges),
but during the construction of the FT-trees much more than $f$ edges may fail,
so this is not directly possible.
We will see that expaths offer a suitable alternative.

\begin{definition}[$\ell$-decomposable path]
\label{definition:h-decomposable}
Let $A \subseteq E$ be a set of edges and $\ell$ a positive integer. 
An $\ell$\emph{-decomposable} path in $G - A$ is
a concatenation of at most $\ell+1$ shortest paths of $G$.
\end{definition}

\begin{definition}[$\ell$-expath]
\label{def:l-expath}
  Let $A \subseteq E$ be a set of edges and $\ell$ a positive integer.
  An $\ell$-\emph{expath} in $G-A$ is a concatenation of $(2\log_2(n)+1)$
  $\ell$-decomposable paths
  such that, for every $0 \le i \le 2\log_2 n$,
  the length of the $i$-th path is at most $\min(2^i, 2^{2\log_2(n) - i})$.
\end{definition}

Since $n\,{-}\,1$ is an upper bound on the diameter of any connected subgraph of $G$,
the middle level $i \,{=}\, \log_2 n$ is large enough to accompany any (decomposable) path.
Levels may be empty.
Therefore, for any $\ell' \ge \ell$, an $\ell$-decomposable path is also both $\ell'$-decomposable
and an $\ell'$-expath.
Also, an arbitrary subpath of an $\ell$-decomposable path 
(respectively, \mbox{$\ell$-expath}) is again $\ell$-decomposable (respectively, an \mbox{$\ell$-expath}).
This gives the following intuition why it is good enough to work with expaths.
Suppose some replacement path $P(u,v,F)$ survives in $G{-}A$
albeit $A \supseteq F$ may be much larger than $F$, 
then the shortest $u$-$v$-path in $G-A$ is indeed $P(u,v,F)$ and thus $f$-decomposable.
The length of the shortest $(2f{+}1)$-expath
between $u$ and $v$ in $G{-}A$ is the actual replacement distance $|P(u,v,F)| = d_{G-F}(u,v)$.
The reason for the choice $\ell = 2f+1$ will become apparent in the proof of~\Cref{lem:sublinear_query}.
The difficulties of working merely with $(2f{+}1)$-\emph{decomposable} paths are described in~\Cref{lem:prefix_optimality}.

Finally, we define a set $B$ of special vertices of $G$ that we call \emph{pivots}.
Recall that we are mainly interested in paths with more than $L$ edges.
Suppose $L = \omega(\log n)$.
We construct the set $B$ by sampling any vertex from $V$ independently with probability
$C' f \log_2(n)/L$ for some sufficiently large constant $C' > 0$.
With high probability, we have $|B| = \Otilde(n/L)$
and any replacement path with more than $L/2$ edges
in any of the graphs $G-F$ with $|F| \le f$
contains a pivot as can be seen by standard Chernoff bounds,
see e.g.~\cite{GrandoniVWilliamsFasterRPandDSO_journal,RodittyZwick12kSimpleShortestPaths,WY13}.
\vspace*{.5em}

\noindent
\textbf{Interface of Data Structure \emph{FT}.}
For a positive integer $\ell$ and vertices $u,v \in V$,
define $d^{(\ell)}_{\eps/9}(u,v,F)$
to be the minimum length over all $\ell$-\emph{decomposable} paths
between $u$ and $v$ in $G-F$ that are far away from $F$.
If there are no such paths, we set $d^{(\ell)}_{\eps/9}(u,v,F) = +\infty$.
The data structure \emph{FT} can only be queried with triples $(u,v,F)$
for which $u$ or $v$ is a pivot in $B$.
Its returned value satisfies
$d_{G-F}(u,v) \le FT(u,v,F) \le 3 \cdot d^{(2f+1)}_{\eps/9}(u,v,F)$.
We let $q_{FT}$ denote its query time.

\subsection{Querying the Distance Sensitivity Oracle}
\label{subsec:sublinear_algorithm}

We show how to use the black box \emph{FT} to get a $(3{+}\eps)$-approximate $f$-DSO.
As an additional data structure, we instantiate the $f$-DSO for short paths described in \Cref{thm:oracle_short_paths} with parameter $k = 2$.
It thus gives a $3$-approximation, whenever the replacement path in question
has at most $L$ edges.

Fix a query $(s,t,F)$ that we want to answer on the top level.
We build the weighted complete graph $H^F$ 
on the vertex set $V(H^F) = \{s,t\} \cup V(F)$.
For a pair $\{u,v\} \in \binom{V(H^F)}{2}$, let $w_{H^F}(u,v)$ denote the weight of the edge $\{u,v\}$.
Since $G$ is undirected, $w_{H^F}(\cdot,\cdot)$ is symmetric.
To simplify notation, we allow possibly infinite edge weights instead of removing the respective edge.
The weight $w_{H^F}(u,v)$ is intended to roughly equal  to the replacement distance $d_{G-F}(u,v)$.
While we do not achieve this exactly,
most of the time we can ensure that 
$d_{G-F}(u,v) \le w_{H^F}(u,v) \le 3 \cdot d_{G-F}(u,v)$.
Either way,
we show in \Cref{lem:sublinear_query} that we can find weights that are good enough
such that for the query vertices $s$ and $t$,
the $s$-$t$-distance in $H^F$ is a $(3{+}\varepsilon)$-approximation
of $d_{G-F}(s,t)$.
It is thus reported as the answer to the query $(s,t,F)$.

Recall that we use $d_{G-F}^{\le L}(u,v)$ for the minimum length
over all short $u$-$v$-paths in the graph $G{-}F$, 
and $\widehat{d^{\le L}}(u,v,F)$ for its $3$-approximation 
from the $f$-DSO for short paths.
The time to obtain that estimate is $\Otilde(L^{o(1)})$.

If $u$ or $v$ is a pivot, we set $w_{H^F}(u,v)$ 
to the minimum of $\widehat{d^{\le L}}(u,v,F)$ and $FT(u,v,F)$.
Otherwise, if $\{u,v\} \cap B = \emptyset$,
we set it to the minimum of $\widehat{d^{\le L}}(u,v,F)$ and 
\begin{equation*}
  w_{H^F}'(u,v)=\min_{b \in B}\left\{FT(u,b,F)+FT(b,v,F)\right\}.
\end{equation*}
If $u$ and $v$ have a short replacement path,
$\widehat{d^{\le L}}(u,v,F)$ is a good estimate of $d_{G-F}(u,v)$.
Otherwise, the replacement path is long.
The computation of the auxiliary weight $w_{H^F}'(u,v)$ then searches
for a pivot that lies on this long path and uses the FT-trees
to obtain the distance.

\begin{lemma}
\label{lem:sublinear_query}
  With high probability over all queries, the query time is
  $\Otilde(L^{o(1)} + \frac{n}{L} \cdot q_{FT})$
  and it holds that
  $d_{G-F}(s,t) \le d_{H^F}(s,t) \le (3{+}\eps) \nwspace d_{G-F}(s,t)$. 
\end{lemma}

\begin{proof}
  The graph $H^F$ has $O(f^2) = O(1)$ edges, 
  and assigning a weight takes $\Otilde(L^{o(1)} + |B| \cdot q_{FT})$ per edge.
  The distance from $s$ to $t$ can be computed using Dijkstra's algorithm in time $O(f^2)$.
  
  We prove the seemingly stronger assertion that for each pair $u,v \in V(H^F)$,
  we have $d_{G-F}(u,v) \le d_{H^F}(u,v) \le (3{+}\eps) \nwspace d_{G-F}(u,v)$.
  The first inequality is immediate from the fact that the values 
  $\widehat{d^{\le L}}(u,v,F)$, $FT(u,v,F)$, and $FT(u,b,F) + FT(b,v,F)$ for any $b \in B$
  are all at least $d_{G-F}(u,v)$.
  
  We prove the second inequality by induction over $d_{G-F}$.
  If $u = v$ (i.e., $d_{G-F}(u,v) = 0$), there is nothing to prove.
  Assume the inequality holds for all pairs of vertices with replacement distance
  strictly smaller than $d_{G-F}(u,v)$.
  We distinguish three cases.
  In the first case, the (unique)
  replacement path $P = P(u,v,F)$ has at most $L$ edges.
  \Cref{thm:oracle_short_paths} then implies
  \begin{equation*}
    d_{H^F}(u,v) \le w_{H^F}(u,v) \le \widehat{d^{\le L}}(u,v,F)
      \le 3 \cdot d_{G-F}^{\le L}(u,v) = 3 \cdot \length{P},
  \end{equation*}
  which is $3 \nwspace d_{G-F}(u,v)$ as $P$ is a replacement path.

  If the path $P$ is long instead,
  it contains a pivot $b \in B$ w.h.p.\ (possibly $u = b$ or $v = b$).
  For the second case, assume $P$ has more than $L$ edges and is far away from all failures in $F$.
  Note that then the subpaths $P[u..b]$ and $P[b..v]$ are the replacement paths
  for their respective endpoints,
  and therefore both $f$-decomposable (and also $(2f{+}1)$-decomposable).
  Moreover, they are far away from all failures
  as their trapezoids are subsets of $\tr^{\eps/9}_{G-F}(P)$.
  It holds that
  \begin{align*}
    d_{H^F}(u,v) &\le w_{H^F}(u,v) \le FT(u,b,F) + FT(b,v,F)\\
      &\le 3 \cdot d^{(2f+1)}_{\eps/9}(u,b,F)
        + 3 \cdot d^{(2f+1)}_{\eps/9}(b,v,F)\\
      &= 3 \cdot \length{P[u..b]}
        + 3 \cdot \length{P[b..v]}
      = 3 \cdot d_{G-F}(u,v). 
  \end{align*}  
  
  Finally, for the third case suppose the replacement path $P$ is long but \emph{not} far away from $F$.
  \Cref{lem:not_too_far_off} states the existence of three vertices $x \in \{u,v\}$, $y \,{\in}\, V(P)$,
  and $z \,{\in}\, \tr^{\eps/9}_{G-F}(P) \cap V(F)$
  such that $d_{G-F}(z,y) \le \tfrac{\eps}{9} \,{\cdot}\, d_{G-F}(x,y)$.
  The path $P' = P[x..y] \circ P(y,z,F)$ is far away from all failures
  and has length at most \mbox{$(1+ \tfrac{\eps}{9}) \cdot d_{G-F}(x,y)$}.
  In the remainder, we assume $x=u$; the argument for $x = v$ is symmetric.
  If the concatenation $P'$ has at most $L$ edges, we get
  \begin{equation*}
    w_{H^F}(u,z) \le \widehat{d^{\le L}}(u,z,F) \le 3 \nwspace \length{P'}
      \le 3 \left(1+\frac{\eps}{9}\right) \nwspace d_{G-F}(u,y) = \left(3+\frac{\eps}{3} \right) \, d_{G-F}(u,y).
  \end{equation*}
  Note that we do mean $d_{G-F}(u,y)$ here and not $d_{G-F}(u,z)$.
  
  The subpath $P[u..y]$ is in fact the unique replacement path $P(u,y,F)$.
  So, if $P'$ has more than $L$ edges,
  one of its subpaths $P[u..y]$ or $P(y,z,F)$ has more than $L/2$ edges.
  Thus, there exists a pivot $b$ on $P'$.
  Here, we actually use the uniqueness of shortest paths in $G$
  since replacing, say, $P[u..y]$  with another shortest $u$-$y$-path in $G-F$
  to ensure a pivot
  may result in a concatenation that is no longer far away from all failures.
  Similar to the second case, we arrive at
  \begin{align*}
    w_{H^F}(u,z) &\le FT(u,b,F) + FT(b,z,F)\\
      &\le 3 \cdot d^{(2f+1)}_{\eps/9}(u,b,F)
        +3 \cdot d^{(2f+1)}_{\eps/9}(b,z,F)\\
      &\le 3 \cdot \length{P'[u..b]} + 3 \cdot \length{P'[b..z]}
       = 3 \nwspace \length{P'}
       \le \left(3+\frac{\eps}{3}\right) d_{G-F}(u,y).
  \end{align*}
  It is important that \emph{FT} approximates $d^{(2f+1)}_{\eps/9}$
  since $P'$ may not be \mbox{$f$-decomposable}.
  As the concatenation of two $f$-decomposable paths,
  $P'$ is $(2f{+}1)$-decomposable; so are $P'[u..b]$ and $P'[b..z]$.
  
  Now that we have an upper bound on $w_{H^F}(u,z)$
  we can conclude the third case.
  Since $\tfrac{\eps}{9} < 1$ and $z \in \tr^{\eps/9}_{G_F}(P)$ 
  (where $P$ is the $u$-$v$-replacement path), the distance $d_{G-F}(z,v)$
  is strictly smaller than $d_{G-F}(u,v)$.
  By induction, $d_{H^F}(z,v) \le (3{+}\eps) \cdot d_{G-F}(z,v)$.
  Recall that vertex $y$ lies on~$P$, whence $d_{G-F}(u,y) + d_{G-F}(y,v) = d_{G-F}(u,v)$.
  Due to $\eps \le 3$, we have 
  $(2{+}\tfrac{\eps}{3}) \tfrac{\eps}{9} \le \tfrac{\eps}{3}$.
  Also, recall that $d_{G-F}(z,y) \le \tfrac{\eps}{9} \nwspace d_{G-F}(u,y)$ 
  by the definition of $z$ and $x = u$.
  Putting everything together, we estimate the $u$-$v$-distance in the graph $H^F$.

 \begin{align*}
    d_{H^F}(u,v) &\le w_{H^F}(u,z) + d_{H^F}(z,v)
      \le \left(3+\frac{\eps}{3}\right)  d_{G-F}(u,y)
        + (3+\eps)  d_{G-F}(z,v)\\[.33em]
      &\le 3 \left(\left(1+\frac{\eps}{9}\right) d_{G-F}(u,y)
        + \left(1+\frac{\eps}{3} \right) \left(d_{G-F}(z,y) + d_{G-F}(y,v) \right) \right)\\[.33em]
      &\le 3 \left(\left(1+\frac{\eps}{9}\right) d_{G-F}(u,y)
        + \left(1+\frac{\eps}{3} \right) \left( \frac{\eps}{9}  d_{G-F}(u,y) 
          + d_{G-F}(y,v) \right) \right)\\[.33em]
      &= 3 \Big( d_{G-F}(u,y) + d_{G-F}(y,v) 
        + \left(2+\frac{\eps}{3} \right) \!\frac{\eps}{9}  d_{G-F}(u,y) 
        + \frac{\eps}{3}  d_{G-F}(y,v) \!\Big)\\[.33em]
      &\le  3 \Big(d_{G-F}(u,v) + \frac{\eps}{3}  d_{G-F}(u,y) 
        + \frac{\eps}{3}  d_{G-F}(y,v) \Big)\\[.33em]
      &= 3 \left(1+\frac{\eps}{3} \right)  d_{G-F}(u,v)
       = (3 {+} \eps) \nwspace d_{G-F}(u,v). \qedhere
  \end{align*}
\end{proof}

\subsection{Fault-Tolerant Trees}
\label{subsec:sublinear_FT-trees}

We now describe the implementation of the \emph{FT} data structure
via fault-tolerant trees.
We compute all-pairs shortest distances in the original graph $G$ 
(slightly perturbing edge weights for unique shortest paths if required),
and, for each pivot $b \in B$, a shortest path tree of $G$ rooted in $b$
in $\Otilde(mn)$ time.
We turn each of those trees into a data structure that reports
the lowest common ancestor (LCA) in constant time
with the algorithm of Bender and Farach-Colton~\cite{BenderFarachColton00LCARevisited}.
This takes time and space $O(|B|n) = \Otilde(n^2/L)$ w.h.p.

We also assume that we have access to a procedure
that, given any set $A \subseteq E$ of edges 
(which may have much more than $f$ elements) and a pair of vertices $u,v \in V$,
computes the shortest $(2f{+}1)$-expath between $u$ and $v$
in $G{-}A$.
This expath is labeled with its structure,
that means, (a) the start and endpoints of the $2\log_2(n)+1$
constituting $(2f{+}1)$-decomposable subpaths,
and (b) inside each decomposable path the start and endpoint of the constituting shortest paths
(and possibly interleaving edges). 
The explanation of how to achieve this in time $\Otilde(fm)$ is deferred to \Cref{sec:improved_preproc}.
This is also the key ingredient of the proof of \Cref{thm:improved_preprocessing}.

We build the FT-trees only for pairs of vertices $(u,b)$ for which $b \in B$ is a pivot.
On a high level, $FT(u,b)$ is a tree of depth $f$ 
that stores in each node the shortest $(2f{+}1)$-expath between $u$ and $b$ in some graph $G{-}A$.
We first describe the information that we hold in a single node $\nu$.
Let $P_\nu$ be the stored expath.
It is partitioned first into segments and those are partitioned further into parts.
To define the segments, we need the notion of netpoints.

\begin{definition}[Path netpoints]
\label{def:sparse_netpoints}
  Let $P = (u=v_1, \dots, v_\ell = b)$ be a path.
  Define $p_{\text{left}}$ to be all pairs of consecutive vertices $v_j,v_{j+1} \in V(P)$,
  for which there is an integer $i \ge 0$, 
  such that $\length{P[u..v_j]} < (1+ \tfrac{\eps}{36})^i \le \length{P[u..v_{j+1}]}$.
  Let $p_{\text{right}}$ be all vertices $v_j,v_{j-1} \in V(P)$ 
  such that $\length{P[v_j..b]} < (1+ \tfrac{\eps}{36})^i \le \length{P[v_{j-1}..b]}$
  for some $i$.
  The \emph{netpoints} of $P$ are the vertices in 
  $p_{\text{left}} \cup p_{\text{right}} \cup \{u,b\}$.
\end{definition}

\begin{figure}
\centering
\begin{tikzpicture}
\begin{scope}[scale=1.23]
     \coordinate (start) at (1,0);
    \coordinate (end) at (14,0);

    \draw[dashed,-] (start) -- ++(0,1) node [above] {$0$};
    \draw[dashed,-] (end) -- ++(0,1) node [above right] {distance};
    \draw[dashed,-] (2.3, 0) -- ++(0,1) node [above, cred] {$(1+ \tfrac{\eps}{36})^1$};
    \draw[dashed,-] (5.6, 0) -- ++(0,1) node [above, cblue] {$(1+ \tfrac{\eps}{36})^2$};
    \draw[dashed,-] (10.5, 0) -- ++(0,1) node [above, cyellow] {$(1+ \tfrac{\eps}{36})^3$};

    \draw[dashed, lightgray,-] (end) -- ++(0,-1) node [below] {$0$};
    \draw[dashed, lightgray,-] (start) -- ++(0,-1) node [below left] {distance};
    \draw[dashed, lightgray,-] (15-2.3, 0) -- ++(0,-1) node [below] {$(1+ \tfrac{\eps}{36})^1$};
    \draw[dashed, lightgray,-] (15-5.6, 0) -- ++(0,-1) node [below] {$(1+ \tfrac{\eps}{36})^2$};
    \draw[dashed, lightgray,-] (15-10.5, 0) -- ++(0,-1) node [below] {$(1+ \tfrac{\eps}{36})^3$};
    
    \node[vert, fill=white] (v1) at (start) {$u$};
    \node[vert, fill=cred] (v2) at (2,0) {};
    \node[vert, fill=cred] (v3) at (3,0) {};
    \node[vert] (v4) at (4,0) {};
    \node[vert, fill=cblue] (v5) at (5,0) {};
    \node[vert, fill=cblue] (v6) at (6,0) {};
    \node[vert] (v7) at (7,0) {};
    \node[vert] (v8) at (8,0) {};
    \node[vert] (v9) at (9,0) {};
    \node[vert, fill=cyellow] (v10) at (10,0) {};
    \node[vert, fill=cyellow] (v11) at (11,0) {};
    \node[vert] (v12) at (12,0) {};
    \node[vert] (v13) at (13,0) {};
    \node[vert, fill=white] (v0) at (end) {$b$};
    \draw[-]  (v1) -- (v2) -- (v3) -- (v4) -- (v5) -- (v6) -- (v7) -- (v8) -- (v9) -- (v10) -- (v11) -- (v12) -- (v13) -- (v0);
    \draw (1, 0.9) -- (14.5, 0.9);
    \draw[lightgray] (14, -0.9) -- (0.5, -0.9);

    \node[cred] at (2.5, -0.5) {$i=1$};
    \node[cblue] at (5.5, -0.5) {$i=2$};
    \node[cyellow] at (10.5, -0.5) {$i=3$};
    \end{scope}
\end{tikzpicture}
\caption{The red, blue and yellow vertices are some of the netpoints contained in $p_\text{left}$ for the path $(u,b)$. The set $p_\text{right}$ is created in the same way but with reversed distance (hinted below in grey), i.e., from $b$ to $u$.}
\label{fig:netpoints}
\end{figure}

The netpoints can equivalently be seen as the result of the following process, 
illustrated in \Cref{fig:netpoints}.
Given the $u$-$b$-path $P$, start from the endpoint $u$.
For every power of $1+\tfrac{\eps}{36}$, 
mark the vertex that is the furthest away from $u$ but whose path length along $P$
is still less than this power.
Additionally, mark its immediate successor on the path.
These are exactly the vertices in $p_{\text{left}}$.
When doing the same thing from the other endpoint $b$, one obtains $p_{\text{right}}$.

A \emph{segment} of the path $P$ is the subpath between consecutive netpoints.
Note that the bounding netpoints may stem from both $p_{\text{left}}$ and $p_{\text{right}}$.
For an edge $e \in E(P)$, let $\seg(e,P)$ denote the segment of $P$ containing $e$.
There are segments that contain only a single edge,
e.g., the ones between a marked vertex and its immediate successor in the process above.
The others have exponentially increasing lengths,
with $1+ \tfrac{\eps}{36}$ as the base of the exponential.
However, since we define the segments from \emph{both} ends of $P$, they do not grow too large.
This is made precise in~\Cref{lem:segments_not_too_large} below.

We only ever store expaths in the FT-trees,
there the segments need to be subdivided into parts.
Recall that an expath $P$ is made up of logarithmically many decomposable subpaths.
The decomposable subpaths, in turn, consist of $O(f)$ shortest paths (and interleaving edges) in $G$.
These building blocks may not be aligned with the segments.
This can cause problems as we want to use the structure of an expath but store it segment-wise.
We thus define a \emph{part} as a maximal subpath of $P$
that is completely contained in one segment 
and, at the same time, in one of the constituting shortest paths.
We can find all parts by a linear scan over the labels of the expath
that are provided by the $\Otilde(fm)$-time procedure that computes $P$.

Note that each part is a shortest path/edge in $G$
(they are defined as subpaths of shortest paths).
We assume that the shortest paths  in $G$ are unique.
It is therefore enough to represent a part by its endpoints.
With any part $[v,w]$, for $v,w \in V(P)$,
we store pointers to the closest netpoint before $v$ and after $w$,
possibly $v$ and $w$ themselves.
This allows us to quickly find the segment in which the part lies.
We also store the original graph distance $d(v,w)$.
If the part contains more than $L$ edges,
we mark this fact and store a pivot $p \in B$ that lies in $[v,w]$.

We now describe the FT-tree $FT(u,b)$ recursively.
In some node $\nu$, let $A_{\nu}$ be the set of all edges
that were failed in the path from the root to $\nu$;
with $A_{\nu} = \emptyset$ in the root itself.
We compute the shortest $(2f{+}1)$-expath $P_{\nu}$ in $G-A_{\nu}$
and store the information for all its parts.
For each of its segments $S$, we create a child node $\mu$ in which we set
$A_{\mu} = A_{\nu} \cup E(S)$.
That means, the transition from a parent to a child corresponds to failing the \emph{whole segment}.
Note that the sets $A_{\nu}$ are only used during preprocessing and never actually stored.
We continue the recursive construction until depth $f$ is reached;
if in a node $\nu$ the vertices $u$ and $b$ become disconnected,
we mark this as a leaf node not storing any path.
We build one FT-tree for each pair of (distinct) vertices in $V \times B$
and additionally store the LCA data structure for each pivot.

The number of segments of any simple path in a subgraph of $G$ is at most
$2\log_{1+ \tfrac{\eps}{36}}(n) +1$.
Therefore, there exists a constant $c > 0$ such that
the maximum number of segments of one path is at most $c \nwspace \log_2(n)/\eps$.
This is an upper bound on the degree of any node,
so there are at most $2(c \nwspace \log_2(n)/\eps)^f$ nodes in each tree.
Moreover, an $(2f{+}1)$-expath consists of $O(f \log n)$
shortest paths.
So there are $O(f\log n+ \log(n)/\eps) = O(f \log (n)/\varepsilon)$ parts in one node,
for each of which we store a constant number of machine words.
The combined space of the FT-trees and LCA data structures is
\begin{equation*}
 |B|n \cdot O\!\left(\frac{f\log n}{\varepsilon}\right) \cdot O\!\left(\frac{\log n}{\varepsilon}\right)^{f} + O(|B|n)
 	= \Otilde\!\left(\frac{n^2}{L}\right)
 		\cdot O\!\left(\frac{\log n}{\varepsilon}\right)^{f+1}.
\end{equation*}
The time spent in each node is dominated by computing the $(2f{+}1)$-expath.
The total time to precompute \emph{FT} is
  $|B|n \cdot \Otilde(fm) \cdot O\!\left(\frac{\log n}{\varepsilon}\right)^f + O(|B|n)
    = \Otilde\!\left(\frac{n^2}{L} m \right) 
    	\cdot O\!\left(\frac{\log n}{\varepsilon}\right)^f$.

\subsection{Querying the Data Structure \texorpdfstring{\emph{FT}}{FT}}
\label{subsec:sublinear_FT_query}

We used in~\Cref{lem:sublinear_query} that the value $FT(u,b,F)$
is between $d_{G-F}(u,b)$ and $3 \, d^{(2f+1)}_{\eps/9}(u,b,F)$,
three times the minimum length of an $(2f{+}1)$-decomposable between $u$ and $b$ in $G{-}F$
that is far away from all failures in $F$.
We now show how to do this.

The main challenge when traversing the FT-tree is to utilize the little information that is stored
in a node $\nu$ to solve the following problem.
We must either find the segment $\seg(e,P_{\nu})$ for some failing edge $e \in F$
or verify that $F \cap E(P_{\nu}) = \emptyset$.
The original solution in~\cite{ChCoFiKa17} was to determine for each 
shortest path/interleaving edge $[v,w]$ on $P_{\nu}$ and edge $e = \{x,y\} \in F$
whether the minimum of $d(v,x) + w(x,y) + d(y,w)$ and $d(v,y) + w(x,y) + d(x,w)$
is equal to $d(v,w)$.
If so, $e$ must lie on the shortest path $P_{\nu}[v..w]$.
Computing the actual segment $\seg(e,P_{\nu}) \supseteq [v,w]$ of the edge,
then merely has to find the closest netpoints before $v$ and after $w$
(including $v$ and $w$ themselves).
The problem is that this approach requires storing all $\Omega(n^2)$ original graph distances in $G$,
which we cannot afford.
We first prove that we can get a weaker guarantee.

\begin{lemma}
\label{lem:weaker_guarantee}
  Let $\nu$ be a node of $FT(u,b)$.
  There exists an algorithm to check that
  there is a path between $u$ and $b$ in $G{-}F$
  that has length at most $3 \nwspace \length{P_{\nu}}$ 
  or find the segment $\seg(e,P_{\nu})$ 
  for some $e \in F \cap E(P_{\nu})$.
  The computation time is $\Otilde(L^{o(1)}/{\eps})$.
\end{lemma}

\begin{proof}
  Note that one of the alternatives must occur
  for if $F \cap E(P_{\nu}) = \emptyset$, then $P_{\nu}$ exists in $G{-}F$.
  Consider a part $[v,w]$ of $P_{\nu}$.
  If it has more than $L$ edges,
  then we stored a pivot $p$ in $[v,w]$.
  More precisely, $[v,w]$ is the concatenation of the unique shortest path between $v$ and $p$
  and the one between $p$ and $w$ in $G$.
  We have access to a shortest path tree rooted in $p$.
  So, for each edge $e = \{x,y\} \in F$,
  we can check with a constant number of LCA queries involving $p$, $v$, $w$, $x$, and $y$
  whether edge $e$ is in that concatenation
  in time $O(f)$ per part.
  If all checks fail, we have $d_{G-F}(v,w) = d(v,w) = \length{P_{\nu}[v..w]}$.

  If $[v,w]$ is short, 
  the oracle from~\Cref{thm:oracle_short_paths} is queried with the triple $(v,w,F)$.
  That oracle was preprocessed anyway and answers in time $\Otilde(L^{o(1)})$.
  The return value $\widehat{d^{\le L}}(v,w,F)$ is compared with the original distance $d(v,w)$
  that was stored with the part.
  If the former is more than $3$ times the latter, it must be that $d_{G-F}(v,w) > d(v,w)$,
  so the part contains some edge of $F$.

  We either find a part that has a failing edge in time 
  $\Otilde(L^{o(1)} \cdot f \nwspace \frac{\log^2 n}{\eps})
    = \Otilde(L^{o(1)}/{\eps})$
  or verify that $d_{G-F}(v,w) \le 3 \cdot d(v,w)$ holds for \emph{all} parts.
  In the latter case, swapping each part $[v,w]$ for the path $P(v,w,F)$
  shows the existence of a path in $G{-}F$ of length at most $3 \length{P_{\nu}} =  \sum_{[v,w]} 3 \cdot \nwspace d(v,w)$.

  Finally, let $[v,w]$ be a part for which we determined that it contains a failing edge.
  The query algorithm does not need to know which edges are in $E([v,w]) \cap F$
  since for all of them $[v,w]$ is completely contained in the segment $\seg(e,P_{\nu})$.
  It is thus enough to find the last netpoint on the subpath $P_{\nu}[u..v]$ and 
  the first on $P_{\nu}[w..b]$ by following the pointers.
\end{proof}

We use the lemma to compute $FT(u,b,F)$.
The tree transversal starts at the root.
Once it enters a node $\nu$,
it checks whether there is a path in $G{-}F$
of length at most $3 \nwspace \length{P_{\nu}}$.
If so, this length is returned.
Otherwise, the algorithm obtains a segment $\seg(e,P_{\nu})$ for some $e \in F \cap E(P_{\nu})$
and recurses on the corresponding child.
Once a leaf $\nu^*$ is encountered, the length $\length{P_{\nu^*}}$ is returned;
or $+\infty$ if the leaf does not store a path.
This takes total time $q_{FT} = \Otilde(L^{o(1)}/\varepsilon)$
since at most $f{+}1 = O(1)$ nodes are visited.
The main argument for the correctness of this procedure 
is to show that if a $(2f{+}1)$-expath $P$
in $G{-}F$ is far away from all failures,
it survives in $G{-}A_{\nu^*}$.

\begin{lemma}
\label{lem:crucial_survival}
  Let $P$ be the shortest $(2f{+}1)$-decomposable path between $u$ and $b$ in $G-F$ 
  that is far away from all failures in $F$.
  Let $\nu^*$ be the node of $FT(u,b)$ in which a value is returned when queried with $F$, 
  and let $A_{\nu^*}$ be the set of edges that were failed from the root to $\nu^*$.
  Then, $P$ exists in the graph $G-A_{\nu^*}$.
  Moreover, it holds that 
  $d_{G-F}(u,b) \le FT(u,b,F) \le 3 \cdot d^{(2f+1)}_{\eps/9}(u,b,F)$.
\end{lemma}

We need the following two lemmas for the proof.
The first one states that the segments of a path are not too long,
or even only contain a single edge.
The second lemma verifies a certain prefix optimality of expaths.
This is the crucial property that decomposable paths are lacking.
For some edge set $A \subseteq E$,
let $d^{(\ell)}(u,v,A)$ be the length of the shortest \mbox{$\ell$-\emph{decomposable}} path in $G{-}A$.
Compared to $d^{(\ell)}_{\eps/9}(u,v,F)$,
this definition allows for larger failure sets and
drops the requirement of the path being far away from the failures.

\begin{lemma}[Lemma~3.2 in~\cite{ChCoFiKa17}]
\label{lem:segments_not_too_large}
  Let $u \in V$ and $b \in B$,
  $P$ be any path between $u$ and $b$, $e \in E(P)$,
  and $y$ a vertex of the edge $e$.
  Then, $E(\seg(e,P)) = \{e\}$ or 
  $\length{\seg(e,P)}  \le \tfrac{\eps}{36} \min(\length{P[u..y]}, \length{P[y..b]})$. 
\end{lemma}

\begin{lemma}[Lemma~3.1 in~\cite{ChCoFiKa17}]
\label{lem:prefix_optimality}
  Let $u \in V$ and $b \in B$, $A \subseteq E$ a set of edges, 
  $\ell$ a positive integer, and $P$ the shortest $\ell$-expath between $u$ and $b$ in $G-A$.
  Then, for every $y \in V(P)$,
  $\length{P[u..y]} \le 4 \cdot d^{(\ell)}(u,y,A)$ 
  and $\length{P[y..v]} \le 4 \cdot d^{(\ell)}(y,v,A)$ both hold.
\end{lemma}

\begin{proof}[Proof of~Lemma~\ref{lem:crucial_survival}]
  The second assertion is an easy consequence of the first.
  $P$ is the shortest $(2f{+}1)$-decomposable $u$-$b$-path in $G-F$ 
  that is far away from all failures in $F$.
  If $P$ also exists in $G-A_{\nu^*}$,
  then $\length{P_{\nu^*}} \le \length{P}$ by the definition of $P_{\nu^*}$
  as the shortest $(2f{+}1)$-expath between $u$ and $v$ in $G-A_{\nu^*}$
  and $P$ being $(2f{+}1)$-decomposable (and thus a $(2f{+}1)$-expath).
  The query algorithm guarantees
  $FT(u,b,F) \le 3 \nwspace \length{P_{\nu^*}} 
    \le 3 \nwspace \length{P} 
    = 3 \cdot d^{(2f+1)}_{\eps/9}(u,b,F)$.
  It is clear that we never underestimate the true distance $d_{G-F}(u,b)$.
  
  We show the existence of the path $P$ in $G-A_{\nu}$ for \emph{every} visited node $\nu$
  by induction over the parent-child transitions of the tree transversal.
  It is true for the root where $A_{\nu} = \emptyset$.
  When going from $\nu$ to a child, $A_{\nu}$ gets increased by
  the edges $E(\seg(e,P_{\nu}))$
  of a segment for some $e_F \in F \cap E(P_{\nu})$.
  It is enough to prove that $P$ does not contain an edge of $\seg(e,P_{\nu})$.
  Intuitively, we argue that the segments are too short
  for their removal to influence a path far away from $F$.
  
  The claim is immediate if $E(\seg(e_F,P_{\nu})) = \{e_F\}$, because $P$ exists in $G-F$.
  For the remainder, suppose $\seg(e_F,P_{\nu})$ consists of more than one edge.
  To reach a contradiction, assume $e_P \in E(P) \cap E(\seg(e_F,P_{\nu}))$
  is an edge in the intersection.
  If $\seg(e_F,P_{\nu})$ contains multiple edges from~$F$,
  we let $e_F$ be the one closest to $e_P$.
  This ensures that the subpath of $P_{\nu}$ between the closest vertices in $e_F$ and $e_P$ 
  does not contain any other failing edges.
  More formally, there are vertices $y \in e_P$ and $z \in e_F$
  such that neither $y$ nor $z$ are the endpoints $u$ or $b$ and
  the subpath $P_{\nu}[y..z]$ lies entirely both in $\seg(e,P_{\nu})$ and the graph $G-F$.
  Since $y \in V(P)$, $z \in V(F)$, 
  and the path $P$ is far away from all failures, 
  $z$ must be outside the trapezoid $\tr^{\eps/9}_{G-F}(P)$, that is,
  \begin{equation*}
    \length{\seg(e_F,P_{\nu})} \ge \length{P_{\nu}[y..z]} 
    \ge d_{G-F}(y,z) > \frac{\eps}{9} \min(\length{P[u..y]},\length{P[y..b]}).
  \end{equation*}
  Conversely, we combine~\Cref{lem:segments_not_too_large,lem:prefix_optimality},
  edge $e_P$ lying both on $P$ and $P_{\nu}$,
  and $P_{\nu}$ (with its subpaths) being a $(2f{+}1)$-expath to arrive at
  \begin{align*}
    \length{\seg(e_F,P_{\nu})}
      &\le \frac{\eps}{36} \min(\length{P_{\nu}[u..y]}, \length{P_{\nu}[y..b]})
       \le \frac{\eps}{36} \cdot 
        \min\!\left(4 \cdot d^{(2f+1)}(u,y,F), \ 4 \cdot d^{(2f+1)}(y,b,F)\right)\\[.33em]
      &= \frac{\eps}{9} \cdot 
        \min\!\left(d^{(2f+1)}(u,y,F), \ d^{(2f+1)}(y,b,F)\right)
       \le \frac{\eps}{9} \cdot \min(\length{P[u..y]}, \length{P[y..b]}). \qedhere
  \end{align*}
\end{proof}

\subsection{Proof of \texorpdfstring{\Cref{lem:long_paths_sublinear_query}}{Lemma 5.1}}
\label{subsec:sublinear_wrapup}

We derive here the parameters of the $f$-DSO with sublinear query time. 
The preprocessing consists of two main parts.
First, the oracle for short paths is computable in time 
$\Otilde(L^{f+o(1)} \nwspace m \sqrt{n})$ (\Cref{thm:oracle_short_paths}).
Secondly, $FT$ has preprocessing time
$\Otilde(mn^2/L) \cdot O(\log n/\varepsilon)^f$,
assuming that we can compute expaths in time $\Otilde(m)$.
We set $L = n^{\alpha/(f+1)}$ for a constant $0 < \alpha < 1/2$.
The preprocessing is dominated by the FT-trees giving a total time of
$\Otilde(mn^{2-\frac{\alpha}{f+1}}) \cdot O(\log n/\varepsilon)^f$.
By~\Cref{lem:sublinear_query} with $q_{FT} = \Otilde(L^{o(1)}/\varepsilon)$
the query time of the resulting oracle is
$\Otilde(n/(\varepsilon L^{1-o(1)}))
  = n^{1-\frac{\alpha}{f+1}+o(1)}/\varepsilon$.
The data structure from~\Cref{thm:oracle_short_paths} requires space $\Otilde(L^{f+o(1)} \nwspace n^{3/2})$, 
and \emph{FT} takes $\Otilde(n^2/ L) \cdot O(\log n/\varepsilon)^{f+1}$.
Inserting our choice of $L$ gives
 $n^{\frac{f}{f+1}\alpha + \frac{3}{2} + o(1)} + 
  	\Otilde\!\left(n^{2-\frac{\alpha}{f+1}}\right)
    \cdot O\!\left(\frac{\log n}{\varepsilon}\right)^{f+1}$.
Since $\alpha < 1/2$ is a constant, the second term dominates.

\section{Reducing the Query Time}
\label{sec:improved_query_time}

We now reduce the query time to $O_{\eps}(n^\alpha)$.
The bottleneck of the query answering is computing the (auxiliary) weight 
$w'_{H^F}(u,v)$ of the edge $\{u,v\}$ in the graph $H^F$,
see the beginning of \Cref{subsec:sublinear_algorithm}.
Minimizing $FT(u,b,F) + FT(b,v,F)$ over all pivots $b$ 
takes linear time in $|B|$.
Let $\lambda = \lambda(L,\varepsilon) \le L$ be a parameter to be fixed later.
We define  $\ball_{G-F}(x, \lambda) = \{z \in V \mid d_{G-F}(x,u) \le  \lambda \}$.
If we had access to the graph $G-F$ at query time,
we could run breath-first searches from $u$ and from $v$ to scan
$\ball_{G-F}(u, \lambda)$ and $\ball_{G-F}(v, \lambda)$ of radius $\lambda$,
and only consider the pivots that are inside these balls. 
By carefully adapting the sampling probability of the pivots to $\Otilde(1/\lambda)$,
we can ensure at least one of them hits the shortest expath (replacement path) for from $u$ to $v$,
more details are given below. 
The problem is that these balls may still contain too many pivots.
In the worst case, we have, say, $\ball_{G-F}(u,\lambda) \cap B = B$
degenerating again to scanning \emph{all} pivots.
Furthermore, we cannot even afford to store all balls as there are $\Omega(nm^f)$ different ones,
a ball for each pair $(x,F)$.
Finally, the assumption of access to $G-F$ itself is problematic in the subquadratic-space regime.

To handle all these issues, we split the computation of $w'_{H^F}(u,v)$ into two cases. 
That of \emph{sparse balls}, where at least one of $\ball_{G-F}(u,\lambda)$ and $\ball_{G-F}(v,\lambda)$ contains fewer than $L^f$ vertices;
and the case of \emph{dense balls} where the two sets both contain more than $L^f$ vertices.
For the sparse ball case, we can reuse the ideas from the previous section
since there are only a few pivots to scan.
The only issue is that we cannot access $G-F$ directly,
instead we use the $(L,f)$-replacement path covering.
However, if we were to apply the same technique also to the dense balls,
the query time would again rise to $\Omega(|B|) = \Omega(n/L)$.
Instead, we introduce FT-trees \emph{with granularity} to handle those.
Unfortunately, those are much larger than the original FT-trees.
We can only keep the total space subquadratic by using fewer of them.
For this, we exploit the fact that the dense balls can also be hit by much fewer than $|B|$ pivots.

\subsection{The Case of Sparse Balls}
\label{subsec:improved_query_sparse}

Consider the same setup as in \Cref{subsec:sublinear_trapezoids},
only that the pivots for $B$ are now sampled with probability $C'' f \log_2(n)/\lambda$
for some $C'' > 0$.
By making the constant $C''$ slightly larger than $C'$ in the original sampling probability
(see the end of \Cref{subsec:sublinear_trapezoids}),
we ensure that w.h.p.\ every path that is a \emph{concatenation} of at most two
replacement paths and has more than $\lambda$ edges contains a pivot.
(Previously, we only had this for ordinary replacement paths with at least $L/2$ edges.)
Note that all statements from \Cref{sec:sublinear_query_time}
except for the space, preprocessing and query time in \Cref{lem:long_paths_sublinear_query} remain true.
Further, observe that in the case of sparse balls, w.h.p.\ there are $\Otilde(L^{f}/\lambda)$ pivots in 
$\ball_{G-F}(u,\lambda)$ or in $\ball_{G-F}(v,\lambda)$.
In this case, it is sufficient to scan those in the same way as we did above.
The only issue is that we do not have access to $\ball_{G-F}(u,\lambda)$ at query time,
so we precompute a proxy.

Let $G_1, \ldots, G_\kappa$ be all the subgraphs of $G$ in the leaves of the sampling trees
introduced in~\Cref{subsec:short_paths_tree_sampling}.
Recall that they form an $(L,f)$-replacement path covering w.h.p.
During preprocessing, we compute and store the sets 
$B_{G_i}(x, \lambda) = B \cap \ball_{G_i}(x,\lambda)$ 
for all the sparse balls $\ball_{G_i}(x, \lambda)$, 
that is, if $|\ball_{G_i}(x, \lambda)| \le L^f$.
Otherwise, we store a marker that $\ball_{G_i}(x, \lambda)$ is dense.\footnote{
	This marker is made more precise in \Cref{subsec:improved_query_dense}.
} 
As $\kappa = L^{f+o(1)}$ and w.h.p.\ $|B_{G_i}(x, \lambda)| = \Otilde(L^{f}/\lambda)$
for sparse balls, storing all of these sets requires 
$\Otilde(nL^{2f+o(1)}/\lambda)$ space.
One can compute $B_{G_i}(x, \lambda)$ by running Dijkstra from $x$ in $G_i$ until at most $L^{f}$ vertices are discovered in time $\Otilde(L^{2f})$.
In total, this takes $\Otilde(nL^{3f+o(1)})$ time.

Suppose we want to compute the weight $w_{H^F}(u,v)$ in the sparse balls case,
meaning that there is an $x \in \{u,v\}$ such that the true set $\ball_{G-F}(x,\lambda)$ is sparse.
If this holds for both $u$ and $v$ the choice of $x$ is arbitrary.
We use $y$ to denote the remaining vertex in $\{u,v\}{\setminus}\{x\}$.
Let $i_1, \ldots, i_r$ be the indices of the graphs $G_{i_j}$ that exclude $F$ 
as computed by \Cref{alg:query_short_paths}. 
We showed in \Cref{subsec:short_paths_tree_sampling} 
that $r =\Otilde(L^{o(1)})$ and that the indices can be found in time proportional to their number.
By definition of $x$, all the proxies $\ball_{G_{i_j}}(x,\lambda)$ for $1 \le j \le r$ are sparse as well.
Departing from \Cref{subsec:sublinear_algorithm}, we redefine the auxiliary weight as
\begin{equation}
\label{eq:sparse-pivots}
	w'_{H^F}(u,v) = \min_{
		\substack{
			1\le j \le r \\
			b \in B_{G_{i_j}}(x,\lambda)} 
		}\left(\widehat{d^{\le L}}(x,b,F) + FT(b,y,F)\right).
\end{equation}
Now that we have precomputed the sets $B_{G_{i_j}}(x,\lambda)$,
it is sufficient to limit the search to pivots that are close to the endpoint $x \in \{u, v\}$.

The actual weight is again
$w_{H^F}(u,v) = \min(\widehat{d^{\le L}}(u,v,F), w'_{H^F}(u,v))$.
Its computation takes time $\Otilde(L^{f+o(1)}/\varepsilon\lambda)$
as there are $\Otilde(L^{o(1)})$ balls, each with $\Otilde(L^f/\lambda)$ pivots,
the values $\widehat{d^{\le L}}$
can be evaluated in time $L^{o(1)}$ (\Cref{thm:oracle_short_paths}),
and we navigate through $\Otilde(L^{f}/\lambda)$ FT-trees
with a query time of $q_{FT} = \Otilde(L^{o(1)}/\varepsilon)$ each. 

Recall the proof of the $(3{+}\varepsilon)$-approximation 
in \Cref{lem:sublinear_query}.
Clearly, if the replacement path $P(u,v,F)$ is short,
then $d_{H^F}(u,v) \le 3 \cdot d_{G-F}(u,v)$ still holds,
the argument was independent of $w'_{H^F}(u,v)$.
We make the next step in recovering what we called the ``second case'' for sparse balls.
The proof of the following lemma motivates the transition from $\Otilde(n/L)$ to $\Otilde(n/\lambda)$ pivots.

\begin{restatable}{lemma}{approxsparseballs}
\label{lem:approx_sparse_balls}
	Let $u,v \in V$ be such that $|\ball_{G-F}(u,\lambda)| \le L^f$
	or $|\ball_{G-F}(v,\lambda)| \le L^f$, and the replacement path $P(u,v,F)$
	is long and far away from all failures in $F$.
	Then, with high probability $w_{H^F}(u,v) \le 3 \cdot d_{G-F}(u,v)$ holds.
\end{restatable}

\begin{proof}
  Let $P = P(u,v,F)$.
  Without losing generality, we assume $\ball_{G-F}(u,\lambda)$ is sparse,
  the other case is symmetric.
  Note that $P$ has at least $L \ge \lambda$ edges.
  Let $u' \in V(P)$ be the vertex on $P$ at distance exactly $\lambda$ from $u$.
  There exists a (regular) pivot $b^*$ on $P[u..u']$ w.h.p.
  Here, we used the adapted sampling probability for set $B$ in \Cref{sec:improved_query_time}.
  Note that the pivot is in $B_{G-F}(u,\lambda) = B \cap \ball_{G-F}(u,\lambda)$.
  The graphs $G_1, \dots, G_\kappa$ are an $(L,f)$-replacement path covering,
  and \Cref{alg:query_short_paths} finds the right indices $i_1, \dots, i_r$.
  \Cref{eq:sparse-pivots} thus gives
  \begin{equation*}
    w_{H^F}(u,v) \le w'_{H^F}(u,v) 
    	 = \min_{
		\substack{
			1\le j \le r \\
			b \in B_{G_{i_j}}(u,\lambda)} 
		}\left(\widehat{d^{\le L}}(u,b,F) + FT(b,v,F)\right)
    	\le \widehat{d^{\le L}}(u,b^*,F) + FT(b^*,v,F).
  \end{equation*}
  
  Recall that \emph{FT} approximates the length $d^{(2f+1)}_{\eps/9}$
  of the shortest $(2f{+}1)$-decomposable path that is far away from all failures.
  As in the proof of \Cref{lem:sublinear_query}, since $P$ is far away from all failures,
  $P[b^*..v]$ is $(2f{+}1)$-decomposable and far away itself.
  It holds that
  \begin{align*}
     w'_{H^F}(u,v) &\le 3 \cdot d_{G-F}(u,b^*)
        + 3 \cdot d^{(2f+1)}_{\eps/9}(b^*,v,F)\\
      &= 3 \nwspace \length{P[u..b^*]}
        + 3 \nwspace \length{P[b^*..v]} = 3 \nwspace |P|
      = 3 \nwspace d_{G-F}(u,v). \qedhere
  \end{align*}  
\end{proof}

\subsection{The Case of Dense Balls}
\label{subsec:improved_query_dense}

\Cref{lem:sublinear_query} showed the $3+\varepsilon$ stretch in \Cref{sec:sublinear_query_time}.
To directly transfer its proof, the inequality $w_{H^F}(u,v) \le 3 \, d_{G-F}(u,v)$ would have to hold
also if both $\ball_{G-F}(u,\lambda)$ and $\ball_{G-F}(v,\lambda)$ are dense
and $P(u,v,F)$ is far away from all failures.
\Cref{eq:sparse-pivots} ensures that but requires a query time of $\Omega(n/\lambda)$
since a dense ball may contain many pivots. 
We provide a more efficient query algorithm via FT-trees with granularity.
Besides the larger space requirement of those trees,
the new query algorithm only gives a $(3{+}\delta)$-approximation for some small $\delta > 0$ (see \Cref{lem:crucial_survival_granularity}).
Therefore, we also have to adapt the proof of \Cref{lem:sublinear_query}.
This is done in \Cref{subsec:improved_wrapup}.

Our changes to the construction of FT-trees are twofold.
We define a set $\B$  of \emph{new} pivots,
polynomially sparser than $B$, by sampling each vertex independently with probability 
$C' f\log_2(n)/\lambda L^{f-1}$.
By a Chernoff bound and $\lambda L^{f-1} \le L^f$, 
it holds that w.h.p.\ $|\B| = \Otilde(n/\lambda L^{f-1})$ and
all sets $\ball_{G-F}(x,\lambda)$ with $|\ball_{G-F}(x,\lambda)| > L^f$
contain a new pivot.
We build an FT-tree with granularity $\lambda$ for each pair in $\B^2$.
\vspace*{.5em}

\noindent
\textbf{FT-Trees with Granularity.}
Given two new pivots $b_u,b_v \in \B$, let $FT_\lambda(b_u,b_v)$ be the \emph{fault-tolerant tree of $b_u$ and $b_v$ with granularity $\lambda$}. 
Granularity affects the netpoints, segments and expaths.

\begin{definition}[Path netpoints with granularity $\lambda$]
\label{def:dense_netpoint}
  Let $P = (b_u = v_1, v_2, \ldots, v_\ell= b_v)$ be a path.
  If $|P| \leq 2\lambda$, then the \emph{netpoints of $P$ with granularity $\lambda$}
  are all vertices $V(P)$ of the path.
  Otherwise, define $p_{\text{left}}$ to be all vertices $v_j,v_{j+1} \in V(P)$ 
  with $\lambda \le j \le \ell-\lambda$ 
  such that $\length{P[v_\lambda..v_j]} < (1{+}\frac{\eps}{36})^i \le \length{P[v_\lambda..v_{j+1}]}$
  for some integer $i \ge 0$.
  Analogously, let $p_{\text{right}}$ be all vertices $v_j,v_{j-1} \in V(P)$ 
  such that $\length{P[v_j..v_{\ell-\lambda}]} < (1{+}\frac{\eps}{36})^i \le \length{P[v_{j-1}..v_{\ell-\lambda}]}$
  for some $i$.
  The \emph{netpoints of P with granularity $\lambda$} are all vertices in
  $\{v_1,\ldots,v_{\lambda}\} \cup p_{\text{left}} \cup p_{\text{right}} 
  	\cup \{v_{\ell-\lambda},\ldots, v_\ell\}$.
\end{definition}

For $\lambda = 0$, this is the same as \Cref{def:sparse_netpoints}.
We denote by $\seg_\lambda(e,P)$ for $e\in P$
the set of \emph{segments} w.r.t.\ to the new netpoints that contains $e$.
Any path has $O(\lambda) + O(\log_{1+\eps}n)= O(\lambda)+ O(\log n/\eps)$ 
netpoints with granularity $\lambda$ and thus so many segments.
The number of nodes per tree is now $(O(\lambda)+O(\log n/\eps))^f = O(\lambda^f) + O(\log n/\eps)^f $
The crucial change is that the first and last $\lambda$ edges are in their own segment
and the exponential length increase happens only in the middle part.

\begin{definition}[$\ell$-expath with granularity $\lambda$]
\label{definition:l-expath_granularity}
  Let $A \subseteq E$ be a set of edges and $\ell$ a positive integer.
  An $\ell$-\emph{expath with granularity} $\lambda$ in $G{-}A$ is a path
  $P_a \circ P_b \circ P_c$
  such that $P_a$ and $P_c$ contain at most $\lambda$ edges each, 
  while $P_b$ is a concatenation of $(2\log_2(n){+}1)$
  $\ell$-decomposable paths
  such that, for every $0 \le i \le 2\log_2 n$,
  the length of the $i$-th such path is at most $\min(2^i, 2^{2\log_2(n) - i})$.
\end{definition}

The \emph{parts} of an $\ell$-expath with granularity $\lambda$ are defined as before.
In each node $\nu$ of $FT_\lambda(b_u,b_v)$,
we store the shortest $(2f{+}1)$-expath $P_\nu$ with granularity $\lambda$ 
from $b_u$ to $b_v$ in $G_\nu$.
Note that $P_\nu$ now has
$O(f\log(n) + \lambda+ \log(n)/\varepsilon) = O(\lambda+ f\log(n)/\varepsilon)$ many parts.
\vspace*{.5em}

\noindent
\textbf{Space and Preprocessing Time.}
Recall the analysis at the end of \Cref{subsec:sublinear_FT-trees},
and also that (compared to that) we changed the size of $|B|$ to $\Otilde(n/\lambda)$
in \Cref{subsec:improved_query_sparse}.
The number of nodes in $FT_{\lambda}(b_u,b_v)$ is $O(\lambda^f) + O(\log n/\varepsilon)^f$
and we only need $|\B|^2 = \Otilde(n^2/\lambda^2 L^{2f-2})$ new trees.
With $\lambda \le L$, this makes for
	$\Otilde\!\left(\frac{n^2}{L^{f}} \right) + 
		\Otilde\!\left(\frac{n^2}{\lambda^2 L^{2f-2}}\right)  
			O\!\left(\frac{\log n}{\varepsilon} \right)^f$
nodes in all new trees.
These are fewer than the $\Otilde(n^2/\lambda) \cdot O(\log n/\varepsilon)^f$
we had for the original FT-trees (that we still need to preprocess).
The more efficient expath computation transfers to expaths with granularity,
see \Cref{subsec:improved_preproc_granularity}.
We can compute such a path in asymptotically the same time $\Otilde(fm) = \Otilde(m)$.
So the preprocessing time of the new trees is dominated by the one for the old trees.
Also, the additional $\Otilde(nL^{3f+o(1)})$ term for the sparse/dense balls
will turn out to be negligible, see \Cref{subsec:proof_main_theorem}.
More importantly, the number of nodes
gets multiplied by the new number of parts $O(\lambda+ f \log(n)/\eps)$ to get their size.
The result is
\begin{equation*}
	\Otilde\!\left(\frac{n^2}{L^{f-1}}\right)
		+ \Otilde\!\left(\frac{n^2}{\lambda L^{2f-2}}\right) 
			\cdot O\!\left(\frac{\log n}{\varepsilon} \right)^f +
		\Otilde\!\left(\frac{n^2}{\varepsilon L^{f}}\right)
		+ \Otilde\!\left(\frac{n^2}{\lambda^2 L^{2f-2}}\right) 
			\cdot O\!\left(\frac{\log n}{\varepsilon} \right)^{f+1}\!.
\end{equation*}
Due to $f \ge 2$ (see \Cref{thm:oracle_long_paths}),
all terms are bounded by the $\Otilde(n^2/\lambda) \cdot O(\log n/\varepsilon)^{f+1}$
from the old FT-trees.
The $\Otilde(nL^{2f+o(1)}/\lambda)$ space to store the regular pivots
of the sparse balls will be shown to be irrelevant in comparison.
\vspace*{.5em}

\noindent
\textbf{Query time.}
A straightforward generalization of \Cref{lem:weaker_guarantee}
shows that evaluating $FT_{\lambda}(b_u,b_v)$ with query set $F$
takes time $\Otilde(L^{o(1)} (\lambda + 1/\varepsilon))$.
To compute $w'_{H^F}(u,v)$,
let again $G_1, \ldots, G_\kappa$ be the graphs in the leaves of the sampling trees (\Cref{subsec:short_paths_tree_sampling}).
For every $G_i$ and vertex $x \in V$ for which $|\ball_{G_i}(x,\lambda)| > L^f$,
we said we store a marker.
More precisely, we associate with $(G_i,x)$
a \emph{single} new pivot $b_x \in \mathcal{B} \cap \ball_{G_i}(x,\lambda)$. 
As before, let $i_1, \dots, i_r$ be the indices of graphs $G_{i_j}$
that are relevant for the query $(u,v,F)$.
Even if $\ball_{G-F}(u,\lambda)$ and $\ball_{G-F}(v,\lambda)$
are dense, it might be that all the $\ball_{G_{i_j}}(u,\lambda)$ are sparse
or all $\ball_{G_{i_j}}(v,\lambda)$ are sparse.
If so, we compute the auxiliary weight $w'_{H^F}(u,v)$
(and in turn $w_{H^F}(u,v) = \min(w'_{H^F}(u,v), d^{\le L}(u,v,F))$) 
via \Cref{eq:sparse-pivots}.
Otherwise, there are indices $i_u,i_v \in \{i_1, \ldots, i_r\}$ such that both
$|\ball_{G_{i_u}}(u,\lambda)| > L^f$ and $|\ball_{G_{i_v}}(v,\lambda)| > L^f$.
If there are multiple such indices, the choice is arbitrary.
Let $b_u \in \mathcal{B} \cap \ball_{G_{i_u}}(u, \lambda)$ and let $b_v \in \mathcal{B} \cap \ball_{G_{i_v}}(v, \lambda)$ be the stored new pivots.
We define
\begin{equation} 
\label{eq:dense-pivots}
w_{H^F}'(u,v)= FT_\lambda(b_u, b_v, F) + 2\lambda.
\end{equation}
Computing this auxiliary weight takes time $\Otilde(L^{o(1)} (\lambda + 1/\varepsilon))$,
much less than for the sparse balls.
It is not immediately obvious why $w_{H^F}'(u,v)$ is still a good estimate of $d_{G-F}(u,v)$.
We prove this in the following section, namely, in \Cref{lem:crucial_survival_granularity}.

\subsection{Approximation Guarantee}
\label{subsec:improved_wrapup}

We have already proven the space and preprocessing time stated in \Cref{thm:oracle_long_paths}
for they are dominated by the original FT-trees,
when accounting for the slightly larger set of regular pivots with $|B| = \Otilde(n/\lambda)$.
We also argued the query time.
The plan to prove the approximation guarantee is the same as in \Cref{sec:sublinear_query_time}.
Recall the structure of the proof of \Cref{lem:sublinear_query}.
It involved an induction over the distances in $G{-}F$.
The induction step had three cases:
first, that the replacement path $P =P(u,v,F)$ has at most $L$ edges;
secondly, that $P$ is in fact long but far away from all failures in $F$;
the third case is that $P$ is long but close to $F$.
The first case is not affected by the introduction of granularity.
We already discussed the second case in \Cref{lem:approx_sparse_balls}
if both $\ball_{G-F}(u,\lambda)$ and $\ball_{G-F}(v,\lambda)$ are sparse.
In the remainder, we prove the second case also if the balls are dense,
and show how to transfer the induction to the new setting.

As a first step, we generalize \Cref{lem:segments_not_too_large,lem:prefix_optimality}
to FT-trees with granularity $\lambda > 0$.

\begin{restatable}{lemma}{segmentsize}
\label{lemma:segment-size} 
  Let $b_u, b_v \in \B$,
  $P$ be any path between $b_u$ and $b_v$, $e \in E(P)$,
  and $y \in e$ a vertex of that edge.
  Then, $E(\seg_{\lambda}(e,P)) = \{e\}$ or 
  $\length{\seg_{\lambda}(e,P)}  
  	\le \frac{\eps}{36} \big(\min(\length{P[b_u..y]}, \length{P[y..b_v]}) - \lambda \big)$. 
\end{restatable}

\begin{proof}
If $|P| \le 2\lambda$ then $\seg_\lambda(e,P) = \{e\}$ for every edge $e$ of $P$ by definition.
We thus assume that $P$ has more than $2\lambda$ edges.
Let $u'$ and $v'$ be the two vertices of $P$ such that $\length{P[b_u..u']}=\lambda$ and $\length{P[v'..b_v]} = \lambda$, respectively. 
Let $e$ be an edge of $P$ such that $\seg_\lambda(e,P) \supsetneq \{e\}$ and $y \in e$.
Note that $y$ must lie on $P[u'..v']$ since $e$ is not among the first or last $\lambda$
edges of $P$.
It is enough to show
\begin{equation*}
	\length{\seg_\lambda(e,P)} \le \frac{\eps}{36} \cdot \min(\length{P[u'..y]}, \length{P[y..v']})
\end{equation*}
since 
$\min(\length{P[u'..y]}, \length{P[y..v']}) = \min(\length{P[b_u..y]}, \length{P[y..b_v]}) - \lambda$.

Let $z$ be the vertex of the edge $e$ that is not $y$.
First, assume that $y$ is closer to $u'$ along $P$ than~$z$,
that is, $|P[u'..y]| < |P[u'.. z]|$.
Let $i$ be the maximal integer such that $(1+\eps/36)^i \le \length{P[u'..y]}$,
whence $(1+\eps/36)^{i+1}>\length{P[u'..y]}$.
Since $\seg_\lambda(e,P)$ contains more edges than just $e$,
the endpoints $y,z$ cannot both be netpoints of the path $P$
with granularity $\lambda$.
Therefore, we even have $(1+\eps/36)^{i+1} > \length{P[u'..z]}$,
which gives
\begin{equation*}
	|\seg_\lambda(e,P)| \le \left(1+\frac{\eps}{36}\right)^{i+1} - \left(1+ \frac{\eps}{36}\right)^i
	 = \frac{\eps}{36} \cdot \left(1+\frac{\eps}{36}\right)^i
	  \le \frac{\eps}{36} \cdot |P[u'..y]|.
\end{equation*}

We can also show $\length{\seg_\lambda(e,P)} \le \eps/36 \cdot \length{P[z..v']}$ with a symmetric argument.
Together with the assumption $|P[u'..y]| < |P[u'.. z]|$ (whence, $\length{P[z..v']} < \length{P[y..v']}$),
it follows that $\length{\seg_\lambda(e,P)} < \eps/36 \cdot \length{P[y..v']}$.

If conversely $|P[u'..y]| > |P[u'.. z]|$ holds (and thus $\length{P[z..v']} > \length{P[y..v']}$),
then the exact same argument as above where $y$ and $z$ switch places
shows $|\seg_\lambda(e,P)| \le \eps/36 \cdot |P[u'..z]| < \eps/36 \cdot |P[u'..y]|$
and $\length{\seg_\lambda(e,P)} \le \eps/36 \cdot \length{P[y..v']}$.
In summary, we get
$\length{\seg_\lambda(e,P)} \le \eps/36 \cdot \min(\length{P[u'..y]}, \length{P[y..v']})$
in both cases.
\end{proof}

Recall that $d^{(\ell)}(u,v,A)$, $A \subseteq E$, is the length
of the shortest $\ell$-decomposable path in~$G{-}A$.

\begin{restatable}{lemma}{prefixgranularity}
\label{lemma:expath_with_granularity_lambda}
  Let $u,v \in V$ be two vertices, $A \subseteq E$ a set of edges,
  and $b_u \in \B \cap \ball_{G-A}(u,\lambda)$ and $b_v \in \B \cap \ball_{G-A}(u,\lambda)$.
  Let further $\ell$ be a positive integer, 
  and $P$ the shortest $\ell$-expath with granularity $\lambda$ between $b_u$ and $b_v$ in $G-A$.
  Then, for every $y \in V(P)$ with $|P[b_u..y]|,|P[y..b_v]| > \lambda$, it holds that
  $\length{P[b_u..y]} \le 4 \cdot d^{(\ell)}(u,y,A) + \lambda$ 
  and $\length{P[y..b_v]} \le 4 \cdot d^{(\ell)}(y,v,A) + \lambda$.
\end{restatable}

\begin{proof}
We only show that $|P[b_u..y]|, \length{P[y..b_v]} > \lambda$ implies 
$\length{P[b_u..y]} \le 4 \nwspace d^{(\ell)}(u,y,A) + \lambda$,
the proof of the other inequality is symmetric.
Let $P^{(\ell)}_{G-A}(u,y)$ be the shortest $\ell$-decomposable $u$-$y$-path
in $G-A$, that is, $d^{(\ell)}(u,y,A) = \length{P^{(\ell)}_{G-A}(u,y)}$.
For the sake of contradiction, assume $\length{P[b_u..y]}-\lambda > 4\length{P^{(\ell)}_{G-A}(u,y)}$.
Let $P_{b_u}$ be the shortest path in $G-A$ from $b_u$ to $u$.
Since $b_u \in \B \cap \ball_{G-A}(u,\lambda)$, it holds that $\length{P_{b_u}} \leq \lambda$.

We first prove that $P'=P_{b_u} \circ P^{(\ell)}_{G-A}[u..y]\circ P[y..b_v]$ is an $\ell$-expath with granularity $\lambda$.
Let $P=P_a \circ P_b \circ P_c$ be the constituting decomposition of $P$ as an expath with granularity.
That means $P_a$ and $P_c$ contain at most $\lambda$ edges each, while $P_b$ is the concatenation $P_0\circ \ldots \circ P_{2\log_2 n}$ of $2\log_2(n)+1$, $\ell$-decomposable paths 
such that $\length{P_i} \leq \min(2^i,2^{2\log_2 n - i})$.
To show that $P'$ is an $\ell$-expath with granularity $\lambda$, we define $\ell$-decomposable paths $P'_0,\ldots, P'_{2\log n}$ in $G-A$ such that $\length{P'_i} \leq \min(2^i,2^{2\log_2 n - i})$ and $P'$ is the concatenation of $P_{b_u} \circ P'_0 \circ \ldots \circ P'_{2\log n} \circ P_c$.

We have $\length{P[b_u..y]} > \lambda$ and thus $\length{P[b_u..y]}-\length{P_a} \ge 1$.
Let $j_0= \lfloor \log_2(\length{P[b_u..y]}-\length{P_a})\rfloor -1$.
Be aware that $j_0  = -1$ is possible.
Since $\length{P_i} \leq 2^i$, we have
\begin{equation*}
	\length{P_a} + \sum_{i=0}^{j_0} \length{P_i} < \length{P_a} + 2^{j_0+1} \le
	\length{P_a} + \length{P[b_u..y]}-\length{P_a}
	= \length{P[b_u..y]}.
\end{equation*}
This implies that either $y$ is contained in a subpath $P_{j_1}$ of $P_b$ for some $j_1 > j_0$ 
or that $y$ is a vertex of $P_c$. 
The latter case is impossible 
as $\length{P_c} \leq \lambda$ while $\length{P[y..b_v]} > \lambda$. 
So $y$ is indeed on $P_{j_1}$.

We are now ready to define the new subpaths $P'_0, P'_1, \ldots, P'_{2\log_2 n}$.
For every $0 \le i <j_0$, we define $P'_i$ as the empty path,
and set $P'_{j_0} = P^{(\ell)}_{G-A}(u,y)$. 
For every $j_0 < i < j_1$, we define $P'_i$ as the empty path again,
and $P'_{j_1}$ is the suffix of $P_{j_1}$ starting at $y$.
Finally, for every $i > j_1$, we set $P'_i = P_i$.
Clearly all the $P'_i$ are $\ell$-decomposable paths in $G - A$.
The only index where the length bound is possibly in doubt is $j_0$.
We need to prove that $\length{P^{(\ell)}_{G-A}(u,y)} \le \min(2^{j_0},2^{2\log_2 n - {j_0}})$.
Note that $j_0 < \log_2 n$ as otherwise $\length{P[b_u..y]}-\length{P_a} > n$,
thus
\begin{multline*}
	\min(2^{j_0},2^{2\log_2 n - {j_0}}) = 2^{j_0}
 	 = 2^{\lfloor\log(\length{P[b_u..y]}-\length{P_a})\rfloor -1}\\
 	\geq \frac{\length{P[b_u..y]} - \length{P_a}}{4} 
 	 \geq \frac{\length{P[b_u..y]} - \lambda}{4}
 	> \length{P^{(\ell)}_{G-A}(u,y)}.
\end{multline*}
The last step uses the assumption
$\length{P[b_u..y]}-\lambda > 4 \nwspace \length{P^{(\ell)}_{G-A}(u,y)}$.

We have established that $P'=P_{b_u} \circ P^{(\ell)}_{G-A}[u..y]\circ P[y..b_v]$
is some $\ell$-expath with granularity~$\lambda$ from $b_u$ to $b_v$ in $G-A$.
However, its length is
\begin{align*}
	\length{P'} &=\length{P_{b_u}} + \length{P^{(\ell)}_{G-A}(u,y)} + \length{P[y..b_v]}
		< \lambda + \frac{\length{P[b_u..y]}-\lambda}{4} + \length{P[y..b_v]}\\
		&= \frac{3\lambda + \length{P[b_u..y]}}{4} + \length{P[y..b_v]}
		< \length{P[b_u..y]}+\length{P[y..b_v]}=\length{P},
\end{align*}
where the last proper inequality follows from $|P[b_u..y]| > \lambda$.
This is a contradiction to $P$ being the shortest 
$\ell$-expath with granularity $\lambda$ from $b_u$ to $b_v$ in $G-A$.
\end{proof}

We use the results to show that also \Cref{lem:crucial_survival}
transfers to non-vanishing granularity, but with a slight loss in the approximation.
Again, $d^{(2f+1)}_{\eps/9}(u,b,F)$ is the length of the shortest $(2f{+}1)$-de\-com\-pos\-able
$u$-$v$-path in $G-F$ that is \emph{far away} from all failures.

\begin{restatable}{lemma}{crucialsurvivalgranularity}
\label{lem:crucial_survival_granularity}
  Define $\delta = 8\lambda/L$.
  Let $u,v \in V$ be such that both
  $|\ball_{G-F}(u,\lambda)|, |\ball_{G-F}(v,\lambda)| > L^f$,
  and  $b_u,b_v \in \B$ the associated new pivots.
  Let $P$ be any $(2f{+}1)$-de\-com\-pos\-able path between $u$ and~$v$ in $G-F$
  that is far away from $F$.
  Then,
  $d_{G-F}(u,v) \le FT_{\lambda}(b_u,b_v,F) + 2\lambda \le 3 \length{P} + \delta L$.
  Moreover, if the \emph{shortest} $(2f{+}1)$-de\-com\-pos\-able path between $u$ and $v$ in $G-F$ 
  that is far away from $F$ has more than $L$ edges,
  then,
  $FT_{\lambda}(b_u,b_v,F) + 2\lambda \le (3+\delta) \cdot d^{(2f+1)}_{\eps/9}(u,v,F)$.
\end{restatable}

\begin{proof}
We prove the survival of $P$ all the way to the output node $\nu^*$ of $FT_{\lambda}(b_u,b_v)$
when queried with set $F$, as in \Cref{lem:crucial_survival}.
We have to take care of the fact that $P$ and $P_{\nu^*}$
may have different endpoints.
In fact, we argue about a longer path.
Let $P(b_u,u,F)$ be the replacement path between $b_u$ and $u$ in $G{-}F$.
It has at most $\lambda$ edges by the choice $b_u \in \ball_{G-F}(u,\lambda)$,
same for $P(v,b_v,F)$.
Also $P$ is $(2f{+}1)$-decomposable, thus $Q = P(b_u,u,F) \circ P \circ  P(v,b_v,F)$
is an $(2f{+}1)$-expath with granularity $\lambda$.
We argue by induction that $Q$ exists in the graph $G{-}A_{\nu}$ for every visited node $\nu$.
This is clear for the root.
For a non-output node $\nu \neq \nu^*$, let $\nu'$ be its visited child.

To reach a contradiction, assume $Q$ does not exist in $G-A_{\nu'}$
Thus, there is a segment of
$P_{\nu}$ that contains both a failing edge of $F$ and an edge of $Q$.
Without losing generality,
we choose $e_F \in F$ and $e_Q \in E(Q)$ such that both $e_F$ and $e_Q$ are in $P_\nu$
and the subpath of $P_\nu$ containing both edges contains no other failing edge.
Let $y \in e_Q$ the endpoint closer to $e_F$ along $P_\nu$,
and let $z \in e_F$ the endpoint closer to $e_Q$. 
The subpath $P_\nu[y..z]$ is entirely in $G-F$.

The edges $e_F \ne e_Q$ must be different since $Q$ lies in $G{-}F$.
Segments with more than one edge only appear in the middle part of the stored expath, that is,
$\length{P_\nu[b_u..y]},\length{P_\nu[y..b_v]} > \lambda$.
By \Cref{lemma:segment-size,lemma:expath_with_granularity_lambda}, this implies
\begin{align*}
	\length{\seg_{\lambda}(e_F,P)}  
  		&\le \frac{\eps}{36} \Big(\min\!\big(\length{P_\nu[b_u..y]}, \length{P_\nu[y..b_v]} \big) - \lambda \Big)\\
  		&\le \frac{\eps}{36} \Big(\min\!\big(4 \nwspace d^{(2f+1)}(u,y,A_{\nu})+\lambda,
  			4 \nwspace d^{(2f+1)}(y,v,A_{\nu}) + \lambda \big) - \lambda \Big)\\
  		&= \frac{\varepsilon}{9} \min\!\big(d^{(2f+1)}(u,y,A_{\nu}), d^{(2f+1)}(y,v,A_{\nu})\big).
\end{align*}

We argue next that $y$ lies on the middle part of $Q$, that is, on $P$.
To reach a contradiction, assume that $y$ is on the prefix $P(b_u,u,F)$
and consider the replacement path $P(b_u,y,F) = P(b_u,u,F)[b_u..y]$,
which has at most $\lambda$ edges since $b_u \in \ball_{G-F}(u,\lambda)$.
The concatenation $P(b_u,y,F) \circ P_\nu(y,b_v)$ is some $(2f{+}1)$-expath from
$b_u$ to $b_v$ with granularity $\lambda$,
but due to $|P_\nu[b_u..y]| > \lambda$ it is strictly shorter than $P_\nu$, a contradiction.
Likewise, $y$ being on $P(v,b_v,F)$ contradicts $|P_{\nu}[y..b_v]| > \lambda$.

As $P$ is some $(2f{+}1)$-decomposable path from $u$ to $v$ in $G-A_{\nu}$,
subpaths of decomposable paths are decomposable and $y$ is on $P$,
we get $d^{(2f+1)}(u,y,A_{\nu}) \le \length{P[u,y]}$ 
and $d^{(2f+1)}(y,v,A_{\nu}) \le \length{P[y..v]}$.
Since $P$ is also far away from $e_F$ the vertex $z$ is outside the trapezoid 
$\tr^{\eps/9}_{G-F}(P)$, we finally arrive at the contradiction
\begin{align*}
	\length{\seg_{\lambda}(e_F,P)} &\ge \length{P_{\nu}[y..z]} \ge d_{G-F}(y,z)\\
		&> \frac{\eps}{9} \min(\length{P[u..y]},\length{P[y..b]}) 
			\ge \frac{\varepsilon}{9} \min\!\big(d^{(2f+1)}(u,y,A_{\nu}), d^{(2f+1)}(y,v,A_{\nu})\big).
\end{align*}

For the approximation, we focus on proving 
\begin{equation*}
	d_{G-F}(u,v) \leq FT_{\lambda}(b_u,b_v,F) + 2\lambda 
	\leq (3{+}\delta) \cdot d^{(2f+1)}_{\eps/9}(u,v,F)
\end{equation*}
with $\delta = 8\lambda/L$
if $P$ is the \emph{shortest} $(2f{+}1)$-decomposable path from $u$ to $v$ in $G-F$
and has more than $L$ edges;
in particular, if $\length{P} = d^{(2f+1)}_{\eps/9}(u,v,F)$.
The other claim is established in passing.

Recall that $FT_{\lambda}(b_u,b_v,F) \le 3 \length{P_{\nu^*}}$ for the output node $\nu^*$,
for which we determined that $3 \length{P_{\nu^*}} > d_{G-F}(b_u,b_v)$.
By the triangle inequality, it holds that
$FT_{\lambda}(b_u,b_v,F) + 2\lambda \ge d_{G-F}(b_u,b_v) + d_{G-F}(u,b_u) + d_{G-F}(b_v,v) 
	\ge d_{G-F}(u,v)$.
We have seen that $Q$ survives until $\nu^*$ and that $P_{\nu^*}$ is not longer than $Q$.
In summary,
\begin{multline*}
	FT_{\lambda}(b_u,b_v,F) + 2\lambda \le 3 \length{Q} + 2\lambda
		 \le 3 (\length{P} + 2\lambda) + 2\lambda\\
		 \le 3 \length{P} + 8\lambda =  3 \length{P} + \delta L 
		 < (3+\delta) \, d^{(2f+1)}_{\eps/9}(u,v,F). \qedhere
\end{multline*}
\end{proof} 

Before formally proving the $3+\varepsilon$ stretch of the new query algorithm,
we sketch the necessary changes to \Cref{lem:sublinear_query}.
Recall that we assume $\varepsilon >0$ to be bounded away from $3$,
thus $\Delta = 3-\varepsilon > 0$ is a constant.
We define $\lambda =\frac{\Delta}{96} \varepsilon L$,
which in turn implies $\delta = \frac{\Delta}{12} \varepsilon$.
As it turns out, any $\delta \le \frac{3-\varepsilon}{9+\varepsilon} \varepsilon$ would do
as this ensures 
$\delta +  (6 + \delta + \varepsilon) \frac{\varepsilon}{9} \le \varepsilon$.
In \Cref{lem:sublinear_query} we had $w_{H^F}(u,v) \le 3\, d_{G-F}(u,v)$
if the path was short (``first case'') or long but far away from all failures (``second case'').
We now only have the weaker inequality $w_{H^F}(u,v) \le (3{+}\delta)\, d_{G-F}(u,v)$ due to the dense ball case.
For the ``third case'', we are going to use the $x$-$y$-$z$-argument of \Cref{lem:not_too_far_off} again. 
A similar reasoning as before gives
$w_{H^F}(u,z) \le (3{+}\delta)(1{+}\frac{\varepsilon}{9}) d_{G-F}(u,y)$
(instead of $(3{+}\frac{\varepsilon}{3}) \, d_{G-F}(u,y)$).
The crucial part is to carefully adapt the chain of inequalities to show that also this slightly higher factor gives the desired stretch of $3+\varepsilon$.

\begin{lemma}
\label{lem:improved_query_correctness}
  With the changes made in \Cref{sec:improved_query_time}, 
  and when setting $\lambda = \frac{3-\varepsilon}{96} \varepsilon L$ 
  and $\delta = \frac{8\lambda}{L}$,
  the inequalities
  $d_{G-F}(s,t) \le d_{H^F}(s,t) \le (3{+}\eps) \nwspace d_{G-F}(s,t)$
  hold with high probability over all queries.
\end{lemma}

\begin{proof}
	The structure of the proof is like the one for \Cref{lem:sublinear_query}.
	Recall the graph $H^F$, which depends on the query $(s,t,F)$.
	It has an edge for every pair in $\binom{V(F) \cup \{s,t\}}{2}$,
	  the weight $w_{H^F}(u,v)$ is the minimum of $\widehat{d^{\le L}}(u,v,F)$
	and $w'_{H^F}(u,v)$, where the computation of the latter depends on whether
	we are in the sparse ball case or the dense ball case.
	The details are given in \Cref{eq:sparse-pivots,eq:dense-pivots}
	
	The first inequality $d_{G-F}(s,t) \le d_{H^F}(s,t)$ follows from the fact that all oracle calls used to compute $w_{H^F}(u,v)$ do not underestimate
	the true replacement distance $d_{G-F}(u,v)$.
	The second inequality will be implied by $d_{H^F}(u,v) \le (3{+}\eps) \nwspace d_{G-F}(u,v)$
	holding for all pairs $u,v \in V(H^F)$.
    We prove this by induction over the replacement distance.
 
    Fix a pair $(u,v)$ 
    and assume the statement holds for all pairs with smaller distance in $H^F$.
	We distinguish the same three cases as before,
 	beginning with the replacement path $P = P(u,v,F)$ having at most $L$ edges.
 	The same argument as before, via \Cref{thm:oracle_short_paths}, gives
    $d_{H^F}(u,v) \le 3 \nwspace d_{G-F}^{\le L}(u,v) = 3 \nwspace d_{G-F}(u,v) 
    	\le (3+\varepsilon) \nwspace d_{G-F}(u,v)$ w.h.p.

	In the second case, $P$ is long and far away from all failures.
	If $\ball_{G-F}(u,\lambda)$ or $\ball_{G-F}(v,\lambda)$ contain at most $L^f$ vertices,
	then \Cref{lem:approx_sparse_balls} also shows that
	$d_{H^F}(u,v) \le w_{H^F}(u,v) \le 3 \cdot d_{G-F}(u,v)$.
	For the subcase that $|\ball_{G-F}(u,\lambda)|, |\ball_{G-F}(v,\lambda)| > L^f$,
	recall that $\delta = 8\lambda/L = \frac{3-\varepsilon}{12} \varepsilon$
	and that $d_{\varepsilon/9}^{(2f+1)}(u,v,F)$ is the length of the shortest $(2f{+}1)$-decomposable
	$u$-$v$-path in $G-F$ that is far away from $F$.
	The replacement path $P$ is $f$-decomposable and therefore also $(2f{+}1)$-decomposable.
	It is far away by assumption, so we get
	$d_{\varepsilon/9}^{(2f+1)}(u,v,F) = \length{P} = d_{G-F}(u,v)$.
	The second part of \Cref{lem:crucial_survival_granularity} now implies that
	$d_{H^F}(u,v) \le w'_{H^F}(u,v) \le (3+\delta) \cdot d_{G-F}(u,v) 
		\le (3+\varepsilon) \cdot d_{G-F}(u,v)$.
  
	The main part of the proof consists in recovering the third case,
	where the replacement path $P$ is long but not far away from $F$.
	By \Cref{lem:not_too_far_off}, there are vertices $x \in \{u,v\}$, $y \,{\in}\, V(P)$,
	and $z \,{\in}\, \tr^{\eps/9}_{G-F}(P) \cap V(F)$
	such that $d_{G-F}(z,y) \le \tfrac{\eps}{9} \,{\cdot}\, d_{G-F}(x,y)$
	and the path $P' = P[x..y] \circ P(y,z,F)$ is far away.
	$P(y,z,F)$ denotes the shortest path from $y$ to $z$ in $G-F$.
  	The length of $P'$ is bounded from above by \mbox{$(1+ \tfrac{\eps}{9}) \cdot d_{G-F}(x,y)$}.

	We again assume $x=u$ for simplicity.
	Just as before, if $P'$ has at most $L$ edges, we have
	$w_{H^F}(u,z) \le \widehat{d^{\le L}}(u,z,F) \le 3 \nwspace \length{P'}
      \le 3(1+\frac{\eps}{9}) \, d_{G-F}(u,y)$.
  	If $P'$ has more than $L$ edges, we have to distinguish 
  	the sparse ball and dense ball subcases again.
  	First, note that $P[u..y]$ is a subpath of the replacement path $P = P(u,v,F)$,
  	so it is itself the unique replacement path $P(u,y,F)$.
  	Therefore, $P' = P[u..y] \circ P(y,z,F)$ and \emph{all its subpaths} are a concatenation of at most two replacement paths.
  	Moreover, since replacement paths are $f$-decomposable,
  	all subpaths of $P'$ are $(2f{+}1)$-decomposable.
  	Finally, recall that all subpaths are far away from all failures in $F$.

  	For the sparse balls case, let $a \in \{u,z\}$ be a vertex such that $\ball_{G-F}(a,\lambda)$ is sparse.
  	Consider the vertex $a'$ that is exactly $\lambda$ steps away from $a$ on the path $P'$.
  	Then, $P'[a...a']$ is a concatenation of two replacement paths and has $\lambda$ edges.
  	We adjusted the sampling probability for $B$ to ensure that there is a (regular) pivot
  	$b^* \in B \cap \ball_{G-F}(a,\lambda)$ on $P'$.
  	The subpath $P'[b^*..z]$ is $(2f{+}1)$-decomposable and far away from all failures,
  	so the exact same argument as in \Cref{lem:approx_sparse_balls}
  	gives $w_{H^F}(u,z) \le 3 \nwspace \length{P'} \le 3(1+ \tfrac{\eps}{9}) \cdot d_{G-F}(u,y)$.
  	
  	Regarding the dense ball case, the whole path $P'$ is $(2f{+}1)$-decomposable.
	The first part of \Cref{lem:crucial_survival_granularity}
	together with $L < \length{P'}$ gives
	$w_{H^F}(u,z) \le 3 \length{P'} + \delta L < (3+\delta) \length{P'} \le
	(3+\delta)(1+ \tfrac{\eps}{9}) \cdot d_{G-F}(u,y)$.
	In summary, we have $w_{H^F}(u,z) \le (3+\delta)(1+ \tfrac{\eps}{9}) \cdot d_{G-F}(u,y)$
	in all cases.
	Vertex $z$ lies in the trapezoid associated with the path $P = P(u,v,F)$,
	so $d_{G-F}(z,v) < d_{G-F}(u,v)$.
	By induction, we get $d_{H^F}(z,v) \le (3+\varepsilon) \nwspace d_{G-F}(z,v)$.
	Recall that our choices of $\lambda$ and $\delta$ imply
	$\delta +  (6 + \delta + \varepsilon) \frac{\varepsilon}{9} \le \varepsilon$.
	Combining all this, we arrive at
	\begin{align*}
    d_{H^F}(u,v) &\le w_{H^F}(u,z) + d_{H^F}(z,v)
       \le (3{+}\delta)\! \left(1{+}\frac{\varepsilon}{9}\right)  d_{G-F}(u,y)
        + (3{+}\eps)  d_{G-F}(z,v)\\[.33em]
      &\le (3{+}\delta)\! \left(1{+}\frac{\varepsilon}{9}\right)  d_{G-F}(u,y)
        + (3{+}\eps)  d_{G-F}(z,y) + (3{+}\eps)  d_{G-F}(y,v)\\[.33em]
      &\le (3{+}\delta)\! \left(1{+}\frac{\varepsilon}{9}\right)  d_{G-F}(u,y)
        + (3{+}\eps)\frac{\varepsilon}{9}  d_{G-F}(u,y) + (3{+}\eps)  d_{G-F}(y,v)\\[.33em]
      &= 3 d_{G-F}(u,y) + \delta \cdot d_{G-F}(u,y) +
       (6 + \delta + \varepsilon) \frac{\varepsilon}{9} \cdot d_{G-F}(u,y)\, +
      	(3{+}\eps) \nwspace d_{G-F}(y,v)\\[.33em]
      &\le 3 d_{G-F}(u,y) + \varepsilon \cdot  d_{G-F}(u,y) + (3{+}\eps) \nwspace d_{G-F}(y,v)
       = (3{+}\eps) \nwspace d_{G-F}(u,v). \qedhere
\end{align*}
\end{proof}

\subsection{Proof of \texorpdfstring{\Cref{thm:oracle_long_paths}}{Theorem 1.1}}
\label{subsec:proof_main_theorem}

We can now complete the proof of \Cref{thm:oracle_long_paths}, which we restate below.
It follows in the same fashion as \Cref{lem:long_paths_sublinear_query}
(see \Cref{subsec:sublinear_wrapup}),
but takes into account the changes made in this section.
The main difference is the transition from $L$ to $\lambda$,
giving an extra $1/\varepsilon$ factor in the space and preprocessing time,
and, of course, the improved query time.
\RestateGo{\restateoraclelongpaths}
\oraclelongpaths*

\begin{proof}
	The stretch of $3+\varepsilon$ is treated in \Cref{lem:improved_query_correctness},
	requiring $\lambda = O(\varepsilon L)$.
	The derivation of the other parameters of the theorem is very similar to the proof of \Cref{lem:long_paths_sublinear_query}.
	The main difference is that now the vertices for $B$ are sampled with probability
	$\Otilde(f/\lambda) = \Otilde(1/\varepsilon L)$ (as opposed to $\Otilde(1/L)$ before).
	We again choose $L = n^{\alpha/(f+1)}$.
	
	The total size of the FT-trees with and without granularity
	are discussed in \Cref{subsec:improved_query_dense}.
	There, we claimed that the $\Otilde(nL^{2f+o(1)}/\lambda)$ space needed
	to store all the pivot sets $B_{G_i}(x,\lambda)$ 
	for graphs $G_i$ of the $(L,f)$-replacement path covering is immaterial.
	\begin{equation*}
		\Otilde\!\left(\frac{nL^{2f+o(1)}}{\lambda}\right) 
			= \Otilde\!\left(\frac{nL^{2f-1+o(1)}}{\varepsilon}\right)
			= \Otilde\left(\frac{n^{1+ \frac{\alpha(2f-1+o(1))}{f+1}}}{\varepsilon}\right)
			= \frac{n^{1+\alpha\frac{2f-1+o(1)}{f+1}}}{\varepsilon}.
	\end{equation*}
	The last transformation uses $\alpha = O(1)$.
	We compare this with the space for the FT-trees.
	\begin{equation*}
		\Otilde\!\left(\frac{n^2}{\lambda}\right) \cdot O\!\left(\frac{\log n}{\varepsilon}\right)^{f+1}
			= \Otilde\!\left(\frac{n^{2-\frac{\alpha}{f+1}}}{\varepsilon}\right) \cdot O\!\left(\frac{\log n}{\varepsilon}\right)^{f+1}
			= \Otilde(n^{2-\frac{\alpha}{f+1}}) \cdot O\!\left(\frac{\log n}{\varepsilon}\right)^{f+2}.
	\end{equation*}
	Indeed, due to $\alpha < \frac{1}{2} \le \frac{f+1}{2f+o(1)}$,
	the latter part dominates.
	
	The main effort when answering a query is computing the edge weights in the auxiliary graph $H^F$
	in the sparse ball case.
	There, we have to scan all pivots in $B \cap \ball_{G-F}(x,\lambda)$.
	(The dense ball case only takes time 
	$\Otilde(L^{o(1)}(\lambda + 1/\varepsilon)) 
		= \Otilde(\varepsilon L^{1+o(1)} + L^{o(1)}/\varepsilon)$.)
	Recall that $H^F$ is of size $O(f^2) = O(1)$.
	The query time is
	\begin{equation*}
		\Otilde\!\left(\frac{L^{f+o(1)}}{\varepsilon \lambda}\right)
			= \Otilde\!\left(\frac{L^{f-1+o(1)}}{\varepsilon^2}\right)
			= \Otilde\!\left(\frac{n^{\frac{\alpha(f-1+o(1))}{f+1}}}{\varepsilon^2}\right)
			= \frac{n^{\alpha\frac{f-1}{f+1}+o(1)}}{\varepsilon^2} 
			= O\!\left(\frac{n^\alpha}{\varepsilon^2}\right).
	\end{equation*}

	For the preprocessing time, we assume that we can compute an expath with granularity~$\lambda$
	in time $\Otilde(fm + \lambda) = \Otilde(m)$.
	Even though the FT-trees with granularity are much larger, we need fewer of them.
	It still takes longer to construct all the regular FT-trees.
	Compared to \Cref{lem:long_paths_sublinear_query}, we get an additional $1/\varepsilon$
	factor from replacing $L$ by $\lambda = \Theta(\varepsilon L)$,
	yielding 
	$\Otilde(mn^{2-\frac{\alpha}{f+1}}/\varepsilon)
  	\cdot O(\log n/\varepsilon)^{f}$.
	It takes $\Otilde(nL^{3f+o(1)}) = n^{1+3\alpha \frac{f}{f+1}+o(1)}$ additional time
	to prepare the sets $B_{G_i}(x,\lambda)$.
	This is negligible compared to the 
	$mn^{2-\frac{\alpha}{f+1}}/\varepsilon = \Omega(n^{3-\frac{\alpha}{f+1}})$ term.
	We style the final preprocessing time as 
	$\Otilde(mn^{2-\frac{\alpha}{f+1}}/\varepsilon)
  	\cdot O(\log n/\varepsilon)^{f} =  	
  	\Otilde(mn^{2-\frac{\alpha}{f+1}})
  	\cdot O(\log n/\varepsilon)^{f+1}$
\end{proof}

\section{Computing Shortest \texorpdfstring{$(2f{+}1)$}{(2f+1)}-Expaths in \texorpdfstring{$\Otilde(fm)$}{Õ(fm)} Time}
\label{sec:improved_preproc} 

We finally turn to computing shortest $(2f{+}1)$-expaths in $\Otilde(fm)$ time.
We assume that we are given access to the all-pairs distances in the original graph $G$.
Since the latter data can be obtained in time $\Otilde(mn)$, this completes the proof of the preprocessing time in
\Cref{lem:long_paths_sublinear_query} and \Cref{thm:oracle_long_paths}.
It also allows us to improve the time needed to construct the (superquadratic-space) $f$-DSO with stretch $(1+\eps)$ by Chechik et al.~\cite{ChCoFiKa17}.
More precisely, the preprocessing time, which was $O(n^{5+o(1)}/\varepsilon^f)$~\cite{ChCoFiKa17},
is now reduced to $O(mn^{2+o(1)}/\varepsilon^f)$ (\Cref{thm:improved_preprocessing}),
improving the complexity by a factor of $n^3/m$.

We present our algorithm to compute $(2f{+}1)$-expaths in \emph{weighted} undirected graphs
and with a sensitivity of up to $f = o(\log n/\log\log n)$.
Any edge $\{a,b\} \in E$ carries a weight $w(a,b)$ between $1$ and some maximum weight $W = \poly(n)$.
The reason for this generalization is to fit the framework of~\cite{ChCoFiKa17}.
The definition of graph distances is adjusted accordingly.
This also has an effect on the definition of decomposable paths.
For any non-negative integer~$\ell$, 
Afek et al.~\cite[Theorem~1]{Afek02RestorationbyPathConcatenation_journal}
showed that in unweighted graphs after at most $\ell$ edge failures
shortest paths are the concatenation of up to $\ell+1$ shortest paths in $G$;
if $G$ is weighted, this changes to a
concatenation of up to $\ell+1$ shortest paths and $\ell$ interleaving edges.
We mean the latter whenever we speak of $\ell$-decomposable paths below.
Let $D = n \cdot W$ be an upper bound of the diameter of $G$.
An $\ell$-expath is now a concatenation of $1+2\log_2 D$ $\ell$-decomposable paths
such that, for every $0 \le i \le 2\log_2 D$,
the length of the $i$-th \mbox{$\ell$-decomposable} path is at most $\min(2^i, 2^{2\log_2(D) - i})$.

\subsection{Shortest \texorpdfstring{$(2f{+}1)$}{(2f+1)}-Decomposable Paths}
\label{subsec:improved_preproc_decomposable}

As a warm-up, we first describe how to obtain 
$(2f{+}1)$-decomposable paths efficiently.
We later extend our approach to $(2f{+}1)$-expaths,
with or without granularity.
Our main tool is a modification of Dijkstra's algorithm~\cite[Chapter 22.3]{clrs}.
Given an edge weighted graph with a distinguished source vertex $s$,
the algorithms determines in time $\Otilde(m)$ the distance from $s$ to any other vertex.
It maintains a priority queue of vertices whose distance estimate
is not yet final.
The \texttt{DecreaseKey} operation updates the estimate.
We mainly adapt this operation.

Let $A \subseteq E$ be a set of edges in $G$ and $s,t \in V$ two vertices.
We denote the shortest $\ell$-decomposable path between $u$ and $v$ in $G-A$ by $P^{(\ell)}(s,t,A)$,
its length is $d^{(\ell)}(s,t,A)$.
Recall that $d^{(\ell)}(s,t,A)$ may be larger than $d_{G-A}(s,t)$ if $|A| > \ell + 1$.

\begin{lemma}
\label{lem:f-decomposable}
Given the original distances $d_G(u,v)$ for all $u,v \in V$
and an edge set $A \subseteq E$, 
the distance $d^{(2f+1)}(s, t, A)$ is computable in time $\Otilde(fm)$.
Moreover, one can compute $d^{(2f+1)}(s, v, A)$
for \emph{all} vertices $v \in V$ within the same time bound.
\end{lemma}

\begin{proof}
We prove the lemma by induction over $\ell$ from $0$ to $2f{+}1$,
computing  $d^{(\ell)}(s, v, A)$ for all targets $v \in V$ in the $\ell$-th step.
For the base case, 
note that $d^{(0)}(s,v,A) = d_G(s,v)$ if the shortest $s$-$v$-path in $G$
does not use any edge in $A$ (that is, if it also exists in $G-A$);
and $d^{(0)}(s,v,A) = +\infty$ otherwise.
We use a modified version of Dijkstra's algorithm in the graph $G-A$ from the source~$s$.
Let $d'(s, a)$ be the distance from $s$ to some vertex $a$
computed so far by our algorithm.
During relaxation of an edge $e = \{a, b\}$, we check if the current path is also the shortest path in $G$
by testing if $d'(s,a) + w(a, b) = d_G(s,b)$, with the right-hand side being precomputed.
If this fails, we do not decrease the key of vertex $b$.

We now argue that $d'(s,v) = d^{(0)}(s,v,A)$.
Note that if the shortest $s$-$v$-path in $G$ also lies in $G{-}A$,
then all its edges are relaxed at one point 
and the corresponding checks in the modification succeed.
Indeed, the last key of $v$ in the priority queue (i.e., $d'(s,v)$) 
then is $d_{G}(s,v) = d^{(0)}(s,v,A)$.
Otherwise, due to the uniqueness of shortest paths,
\emph{every} $s$-$v$-path in $G-A$ has length larger than $d_G(s,v)$.
Therefore, the key of $v$ is never decreased, it remains at $+\infty$.

For the induction step, we construct a new \emph{directed} graph $G^*= (V^*, E^*)$ with $V^* = \{s_0\} \cup V_1 \cup V_2$, where $V_1$ and $V_2$ are two copies of $V$. 
For a given $v \in V$, we denote by $v_1$ and $v_2$ the copies of $v$ that are contained in $V_1$ and $V_2$, respectively.
The set $E^*$ contains the following edges.
\begin{itemize}
    \item $(s_0,v_1)$ of weight $w_{G^*}(s_0,v_1)=d^{(\ell-1)}(s, v, A)$
    	for every $v \in V$ with $d^{(\ell-1)}(s, v, A) \neq +\infty$;
    \item $(v_1,v_2)$ of weight $w_{G^*}(v_1,v_2)=0$ for every $v \in V$;
    \item $(u_1,v_2)$, $(v_1,u_2)$ both of weight $w(u,v)$ 
    	for every edge $\{u,v\} \in E {\setminus} A$;
    \item $(u_2,v_2)$, $(v_2,u_2)$ of weight $w(u,v)$ for every $\{u,v\} \in E {\setminus} A$.
\end{itemize}

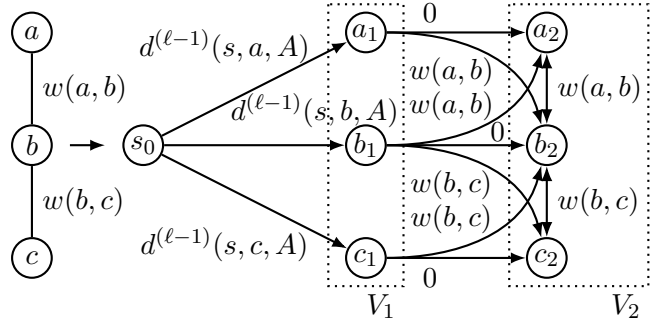
\begin{figure}[t]

\centering
\begin{tikzpicture}
\begin{scope}[scale=1.25]
    \node[vert] (a) at (-.5*\gwidth,  1*\gdist) {$a$};
    \node[vert] (b) at (-.5*\gwidth,  0*\gdist) {$b$};
    \node[vert] (c) at (-.5*\gwidth, -1*\gdist) {$c$};
    \draw (a) edge[-] node[left] {$w(a,b)$} (b);
    \draw (c) edge[-] node[left] {$w(b,c)$} (b);
    
    \node (larrow) at (-.35*\gwidth,  0*\gdist) {};
    \node (rarrow) at (-.15*\gwidth,  0*\gdist) {};
    \draw (larrow) edge (rarrow);
    
    \node[vert] (s0) at (0, 0) {$s_0$};
    \node[vert] (a1) at (1*\gwidth,  1*\gdist) {$a_1$};
    \node[vert] (b1) at (1*\gwidth,  0*\gdist) {$b_1$};
    \node[vert] (c1) at (1*\gwidth, -1*\gdist) {$c_1$};
    \node[vert] (a2) at (2.2*\gwidth,  1*\gdist) {$a_2$};
    \node[vert] (b2) at (2.2*\gwidth,  0*\gdist) {$b_2$};
    \node[vert] (c2) at (2.2*\gwidth, -1*\gdist) {$c_2$};
    
    \coordinate(box1a) at (1*\gwidth - \gpad, 1.25*\gdist);
    \coordinate(box1b) at (1*\gwidth + \gpad,-1.25*\gdist);
    \coordinate(box2a) at (2.2*\gwidth - \gpad, 1.25*\gdist);
    \coordinate(box2b) at (2.2*\gwidth + 2.7*\gpad,-1.25*\gdist);
    \node [anchor=north east] at (box1b) {$V_1$};
    \node [anchor=north east] at (box2b) {$V_2$};
    \begin{scope} [every path/.style = {thick, dotted}]
        \draw (box1a) rectangle (box1b);
        \draw (box2a) rectangle (box2b);
    \end{scope}
    
    \draw
      (a2) edge[Latex-Latex] node[right] {$w(a,b)$} (b2)
      (c2) edge[Latex-Latex] node[right] {$w(b,c)$} (b2)
      ;
    \draw
      (s0) edge["{$d^{(\ell-1)}(s,a,A)$}"] (a1)
           edge["{$d^{(\ell-1)}(s,b,A)$}" xshift=.24cm] (b1)
           edge["{$d^{(\ell-1)}(s,c,A)$}" below] (c1)
      (a1) edge node[above, xshift=1.85cm] {$0$} (a2)
      (b1) edge node[below, xshift=1.85cm] {$0$} (b2)
      (c1) edge node[below, xshift=1.85cm] {$0$} (c2)
      (a1) edge[out=0,in=120] node[below, xshift=-.08cm] {$w(a,b)$} (b2)
      (b1) edge[out=0,in=120] node[below, xshift=-.08cm] {$w(b,c)$} (c2)
      (b1) edge[out=0,in=240] node[above, xshift=-.08cm] {$w(a,b)$} (a2)
      (c1) edge[out=0,in=240] node[above, xshift=-.08cm] {$w(b,c)$} (b2)
      ;
      \end{scope}
  \end{tikzpicture}
  \caption{Example construction of the graph $G^*$ (right) from $G-A$ (left) for step $\ell$ of the algorithm to compute shortest $(2f{+}1)$-decomposable paths (\Cref{lem:f-decomposable}).}
  \label{fig:f-decomposable-construction}
\end{figure}

See \Cref{fig:f-decomposable-construction} for an example of this construction.
The intuition behind $G^*$ is the following.
The edge from the source $s_0$ to some vertex $v_1 \in V_1$ represents the shortest $(\ell{-}1)$-decomposable path from $s$ to $v$ in $G{-}A$.
They become $\ell$-decomposable paths when concatenated with some shortest path in $G$ which only use edges from $E{\setminus}A$.
The latter are represented by the edges from $V_1$ to $V_2$ and those within $V_2$.
If the shortest path from $u$ to $v$ in $G$ is the edge $\{u,v\}$ (respectively, if $u = v$),
this is modeled by the edge $(u_1,v_2) \in V_1 \times V_2$ in $G^*$(respectively, by $(u_1,u_2)$ with weight $0$).
If the shortest $u$-$v$-path in the weighted graph $G$ has multiple edges, 
this is modeled by first following $(u_1,u_2)$ and then the respective edges through $V_2$.

We use $s_0$ as the source and compute the values $d'(s_0,v_i)$ in $G^*$, for $i \in {1,2}$ and all $v \in V$,
with a similar Dijkstra modification as in the base case.
The relaxation of any out-edge of $s_0$ or any edge from $V_1$ to $V_2$ remains unchanged.
For the relaxation of  $e = (a_2,b_2)$ in $G^*[V_2]$,
let $w_2$ be the first vertex from $V_2$ on the shortest path from $s_0$ to $a_2$ that we found. 
We check whether $d'(s_0, a_2) - d'(s_0, w_2) + w_{G^*}(a_2,b_2) = d_G(w,b)$
from the original graph $G$.
If not, we do not decrease the key of $b_2$.
This makes sure that the subpath from $w_2$ to $b_2$
corresponds to the shortest $w$-$b$-path in $G$ (not only in $G-A$).
To have access to $w_2$ in constant time, we store the \emph{entry vertex} $w_2$ for each $a_2$ at the time the key of $a_2$ is decreased.
If this happens using an outgoing edge from $V_1$, then we set the entry vertex of $a_2$ to $a_2$ itself.
Otherwise, we set it to be equal to the entry vertex of its predecessor.

The main part of the proof is to show that the computed distance $d'(s_0,v_2)$ 
is indeed $d^{(\ell)}(s, v, A)$.
This inductively implies that the algorithm eventually produces $d^{(2f+1)}(s,v,A)$.
Any path in $G^*$ from $s_0$ to some vertex $v_2 \in V_2$ has at least $3$ vertices
and its prefix has the form $(s_0, u_1, w_2)$.
By construction, we have $w_{G^*}(s_0,u_1)=d^{(\ell-1)}(s, u, A)$,
which corresponds to the shortest $(\ell{-}1)$-de\-com\-pos\-able path $P^{(\ell-1)}(s,u,A)$ in $G$.
Next, note that there is no edge leaving $G^*[V_2]$,
so the rest of the path from $w_2$ to $v_2$ exclusively uses vertices from $V_2$.
Let $P$ denote the corresponding path in $G$ (meaning it uses the corresponding vertices from $V$).
Our checks in the modification ensure that $P$ is the shortest $w$-$v$-path in $G$.
Slightly abusing notation, let $e  = \{u,w\} \in E{\setminus}A$ be an edge in case $u \neq w$;
and $e = \emptyset$ otherwise.
In summary, the computed path through $G^*$
corresponds to the path $P_{\ell} = P^{(\ell-1)}(s, u, A) \circ e \circ P$ in $G$.
Since $P^{(\ell-1)}(s, u, A)$ is $(\ell{-}1)$-decomposable, and $P$ is a shortest path,
$P_{\ell}$ is $\ell$-decomposable.
It lies entirely in $G{-}A$.

We now prove that $P_{\ell}$ is also the \emph{shortest} $\ell$-decomposable path from $s$ to $v$
in $G{-}A$.
To reach a contradiction, let $Q_{\ell}$ be a strictly shorter $\ell$-decomposable $s$-$v$-path.
$Q_{\ell}$ is not $(\ell{-}1)$-decomposable
since otherwise the two-edge path $(s_0,v_1,v_2)$ in $G^*$ of length $\length{Q_{\ell}}$ 
would be strictly shorter than $P_\ell$.
Our algorithm would have found that path instead (even without modifications).

There exists a decomposition $Q_\ell = Q_{\ell-1} \circ e' \circ Q$,
where $Q_{\ell-1}$ is an $(\ell{-}1)$-decomposable path in $G-A$ ending in some vertex $a$,
$e'$ is either empty or a single edge $\{a,b\} \in E{\setminus}A$,
and $Q$ is a shortest path in $G$ from $b$ to $v$ that only uses edges from $E{\setminus}A$.
Let $Q = (b,x^{(2)}, \dots, x^{(i)},v)$.
Since $d^{(\ell-1)}(s, a, A) \le \length{Q_{\ell-1}}$,
the path $(s_0,a_1,b_2,x^{(2)}_2, \dots, x^{(i)}_2, v_2)$ through $G^*$
has length at most $\length{Q_\ell} < \length{P_\ell}$.
It would have been preferred by our algorithm,
a contradiction.

Concerning the runtime, we have $O(f)$ steps. 
In each of them, we build the graph $G^*=(V^*,E^*)$ with $O(n)$ vertices and $O(m)$ edges 
and run (the modified) Dijkstra's algorithm.
It requires $\Otilde(m)$ time,
so the overall time of our algorithm is $\Otilde(fm)$.
\end{proof}

\begin{figure*}[t]
  \centering
  \begin{tikzpicture}
  \begin{scope}[scale=1.25]
    \node[vert] (a) at (-.5*\gwidth,  1*\gdist) {$a$};
    \node[vert] (b) at (-.5*\gwidth,  0*\gdist) {$b$};
    \node[vert] (c) at (-.5*\gwidth, -1*\gdist) {$c$};
    \draw (a) edge[-] node[left] {$w(a,b)$} (b);
    \draw (c) edge[-] node[left] {$w(b,c)$} (b);
    
    \node (larrow) at (-.35*\gwidth,  0*\gdist) {};
    \node (rarrow) at (-.15*\gwidth,  0*\gdist) {};
    \draw (larrow) edge (rarrow);
    
      \node[vert] (s0) at (0, 0) {$s_0$};
      \node[vert] (a0) at (1*\gwidth,  1*\gdist) {$a_0$};
      \node[vert] (b0) at (1*\gwidth,  0*\gdist) {$b_0$};
      \node[vert] (c0) at (1*\gwidth, -1*\gdist) {$c_0$};
      \node[vert] (a1) at (2.2*\gwidth,  1*\gdist) {$a_1$};
      \node[vert] (b1) at (2.2*\gwidth,  0*\gdist) {$b_1$};
      \node[vert] (c1) at (2.2*\gwidth, -1*\gdist) {$c_1$};
      \node[vert] (ak) at (3*\gwidth,  1*\gdist) {$a_\ell$};
      \node[vert] (bk) at (3*\gwidth,  0*\gdist) {$b_\ell$};
      \node[vert] (ck) at (3*\gwidth, -1*\gdist) {$c_\ell$};
      \node (dots) at (2.77*\gwidth,  0*\gdist) {$\dots$};
      
    \coordinate(box1a) at (1*\gwidth - \gpad, 1.25*\gdist);
    \coordinate(box1b) at (1*\gwidth + \gpad,-1.25*\gdist);
    \coordinate(box2a) at (2.2*\gwidth - \gpad, 1.25*\gdist);
    \coordinate(box2b) at (2.2*\gwidth + 3.4*\gpad,-1.25*\gdist);
    \coordinate(boxka) at (3*\gwidth - \gpad, 1.25*\gdist);
    \coordinate(boxkb) at (3*\gwidth + 3.4*\gpad,-1.25*\gdist);
    \node [anchor=north east] at (box1b) {$V_{0}$};
    \node [anchor=north east] at (box2b) {$V_{1}$};
    \node [anchor=north east] at (boxkb) {$V_{\ell}$};
    \begin{scope} [every path/.style = {thick, dotted}]
      \draw (box1a) rectangle (box1b);
      \draw (box2a) rectangle (box2b);
      \draw (boxka) rectangle (boxkb);
    \end{scope}
    
    \begin{scope} [every path/.style = {{Latex-Latex}}]
      \draw
      (a1) edge node[right] {$w(a,b)$} (b1)
      (c1) edge node[right] {$w(b,c)$} (b1)
      (ak) edge node[right] {$w(a,b)$} (bk)
      (ck) edge node[right] {$w(b,c)$} (bk)
      ;
    \end{scope}
      \draw
      (s0) edge["{$d^{(0)}(s,a,A)$}"] (a0)
           edge["{$d^{(0)}(s,b,A)$}" xshift=.27cm] (b0)
           edge["{$d^{(0)}(s,c,A)$}" below] (c0)
           
      (a0) edge node[above, xshift=1.85cm] {$0$} (a1)
      (b0) edge node[below, xshift=1.85cm] {$0$} (b1)
      (c0) edge node[below, xshift=1.85cm] {$0$} (c1)
      (a0) edge[out=0, in=120] node[below, xshift=-.08cm] {$w(a,b)$} (b1)
      (b0) edge[out=0, in=120] node[below, xshift=-.08cm] {$w(b,c)$} (c1)
      (b0) edge[out=0, in=240] node[above, xshift=-.08cm] {$w(a,b)$} (a1)
      (c0) edge[out=0, in=240] node[above, xshift=-.08cm] {$w(b,c)$} (b1)
      ;
      \end{scope}
  \end{tikzpicture}
  \caption{The exploded graph for computing shortest $\ell$-decomposable paths.}
  \label{fig:exploded graph}
\end{figure*}
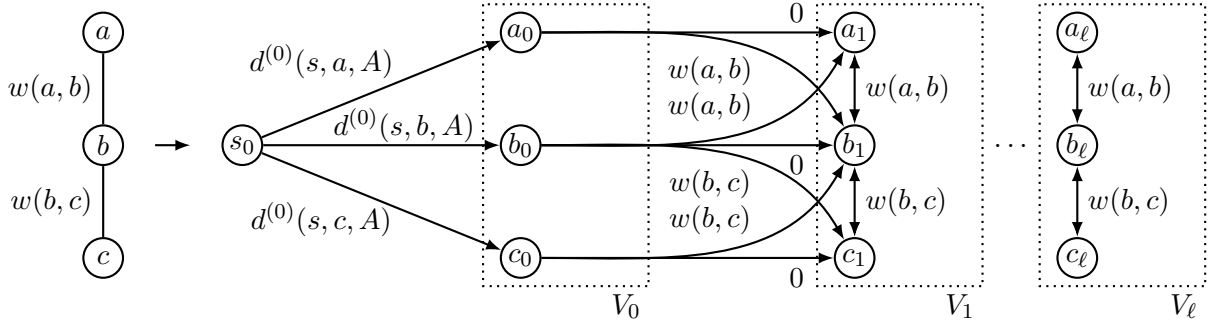

Unrolling the inductive computation of $(2f{+}1)$-decomposable distances
discussed above leads to a graph of $2f+2$ layers as follows.
As the edges from $s_0$ to $v_1 \in V_1$ in the graph $G^*$ are modeling $(\ell{-}1)$-decomposable paths,
we could substitute those by the construction graph used to compute the $(\ell{-}1)$-decomposable paths and proceed recursively.
This way, we obtain an \emph{exploded graph} with a source $s_0$ and $2f+2$ additional layers $V_0,V_1,\ldots,V_{2f+1}$, where layer $V_\ell$ models the graph $G-A$ as above 
and is used to compute shortest $\ell$-decomposable distances.
\Cref{fig:exploded graph} gives an overview.

Running the modified version of Dijkstra's algorithm in the exploded graph guarantees that subpaths entirely contained in one layer are also shortest paths in $G$ (while using only edges from $E{\setminus}A$).
The decomposable distances may be computed out of order;
for example, we may compute the value $d^{(\ell)}(s,v,A)$ before $d^{(\ell-1)}(s,u,A)$ 
provided that $d^{(\ell)}(s,v,A) \le d^{(\ell-1)}(s,u,A)$).
Notwithstanding, for each vertex $v_\ell \in V_\ell$, the computed distance is $d^{(\ell)}(s,v,A)$.

We can further modify Dijkstra's algorithm to limit the length of partial paths to some upper bound $\delta_{i}$, by only allowing edges to be relaxed that do not increase the length of a subpath above that threshold.
For this, we need to store information about the start of a subpath,
e.g., the entry point into the $i$-th layer, for each of its descendant nodes.
This can be propagated during edge relaxation as we did with $w_2$ in the proof above.
In the next subsection, we formally describe and combine these two ideas,
exploded graphs and length restrictions, to compute $\ell$-expaths efficiently.

\subsection{Shortest \texorpdfstring{$(2f{+}1)$}{(2f+1)}-Expaths} 
\label{sec:preprocessing-expaths}

\begin{figure*}[t]

  \centering
  \begin{tikzpicture}
  \begin{scope}[scale=1.23]
    \node[vert] (a) at (-.5*\gwidth,  1*\gdist) {$a$};
    \node[vert] (b) at (-.5*\gwidth,  0*\gdist) {$b$};
    \node[vert] (c) at (-.5*\gwidth, -1*\gdist) {$c$};
    \draw (a) edge[-] node[left] {$w(a,b)$} (b);
    \draw (c) edge[-] node[left] {$w(b,c)$} (b);
    
    \node (larrow) at (-.35*\gwidth,  0*\gdist) {};
    \node (rarrow) at (-.15*\gwidth,  0*\gdist) {};
    \draw (larrow) edge (rarrow);
    
    \node[vert] (s0) at (0, 0) {$s^*$};
    \node[vert] (a0) at (1*\gwidth,  1*\gdist) {$a_0$};
    \node[vert] (b0) at (1*\gwidth,  0*\gdist) {$b_0$};
    \node[vert] (c0) at (1*\gwidth, -1*\gdist) {$c_0$};
    \node[vert] (a1) at (2.2*\gwidth,  1*\gdist) {$a_1$};
    \node[vert] (b1) at (2.2*\gwidth,  0*\gdist) {$b_1$};
    \node[vert] (c1) at (2.2*\gwidth, -1*\gdist) {$c_1$};
    \node[vert] (ak) at (3*\gwidth,  1*\gdist) {$a_\ell$};
    \node[vert] (bk) at (3*\gwidth,  0*\gdist) {$b_\ell$};
    \node[vert] (ck) at (3*\gwidth, -1*\gdist) {$c_\ell$};
    
    \node (dots) at (2.77*\gwidth,  0*\gdist) {$\dots$};
      
    \coordinate(box1a) at (1*\gwidth - \gpad, 1.25*\gdist);
    \coordinate(box1b) at (1*\gwidth + 3.4*\gpad,-1.25*\gdist);
    \coordinate(box2a) at (2.2*\gwidth - \gpad, 1.25*\gdist);
    \coordinate(box2b) at (2.2*\gwidth + 3.4*\gpad,-1.25*\gdist);
    \coordinate(boxka) at (3*\gwidth - \gpad, 1.25*\gdist);
    \coordinate(boxkb) at (3*\gwidth + 3.4*\gpad,-1.25*\gdist);
    \node [anchor=north east] at (box1b) {$V_{i,0}$};
    \node [anchor=north east] at (box2b) {$V_{i,1}$};
    \node [anchor=north east] at (boxkb) {$V_{i,\ell}$};
    \begin{scope} [every path/.style = {thick, dotted}]
      \draw (box1a) rectangle (box1b);
      \draw (box2a) rectangle (box2b);
      \draw (boxka) rectangle (boxkb);
    \end{scope}
      
    \begin{scope} [every path/.style = {{Latex-Latex}}]
      \draw
      (a0) edge node[right] {$w(a,b)$} (b0)
      (c0) edge node[right] {$w(b,c)$} (b0)
      (a1) edge node[right] {$w(a,b)$} (b1)
      (c1) edge node[right] {$w(b,c)$} (b1)
      (ak) edge node[right] {$w(a,b)$} (bk)
      (ck) edge node[right] {$w(b,c)$} (bk)
      ;
    \end{scope}
    \draw
      (s0) edge["{$d^{i-1,\ell}(s,a,A)$}"] (a0)
           edge["{$d^{i-1,\ell}(s,b,A)$}" xshift=.27cm, red] (b0)
           edge["{$d^{i-1,\ell}(s,c,A)$}" below] (c0)
           
      (a0) edge node[above, xshift=1.85cm] {$0$} (a1)
      (b0) edge node[above, xshift=1.85cm] {$0$} (b1)
      (c0) edge node[below, xshift=1.85cm] {$0$} (c1)
      (a0) edge[out=0, in=120] node[below, xshift=-.08cm] {$w(a,b)$} (b1)
      (b0) edge[out=0, in=120] node[below, xshift=-.08cm] {$w(b,c)$} (c1)
      (b0) edge[out=0, in=240, red] node[above, xshift=-.08cm] {$w(a,b)$} (a1)
      (c0) edge[out=0, in=240] node[above, xshift=-.08cm] {$w(b,c)$} (b1)
      (a1) edge[red, dashed] (ak)
      ;
  \end{scope}
  \end{tikzpicture}
  \caption{Example construction of graph $G_i$ (right) from $G$ (left) 
  	for the algorithm to compute shortest $i$-partial $\ell$-expaths (\Cref{lem:l-expath}).
  Each layer $V_{i,j}$ is connected to its neighboring layers in the same way as the first two layers.
  If the red path is the shortest path from $s_0$ to $a_\ell$, then the entry nodes for $a_\ell$ are $w_{a_\ell} = b_0$ and $x_{a_\ell} = a_\ell$.}
  \label{fig:f-expath-construction}
\end{figure*}
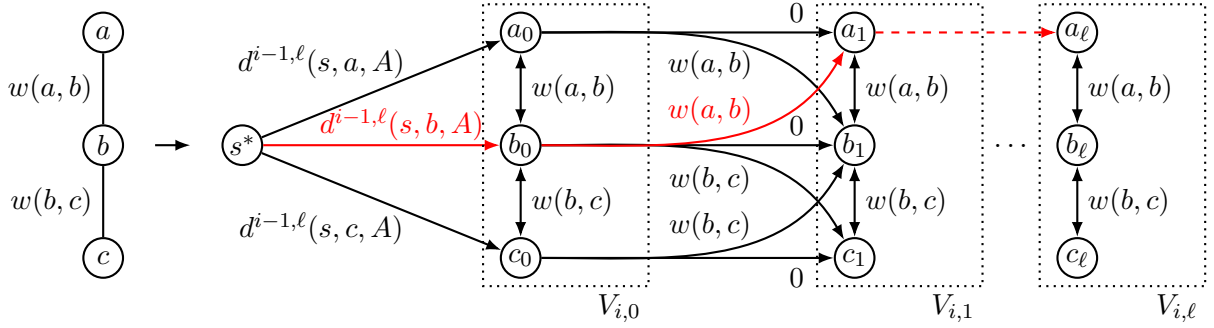

We show how to compute the shortest $(2f{+}1)$-expath in time $\Otilde(fm \log(nW))$,
when given access to all-pairs distances in $G$. 
We define an $i$-partial $\ell$-expath as follows.

\begin{definition}[$i$-partial $\ell$-expath] \label{Definition:k-decomposable}
	Let $A \subseteq E$ be a set of edges and $i,\ell$ non-negative integers.
	An \emph{$i$-partial $\ell$-expath} in $G{-}A$ is a concatenation of $i+1$ $\ell$-decomposable paths in $G{-}A$
	such that, for every $0 \le j \le i$ the length of the $j$-th $\ell$-decomposable path
	is at most $\delta_j=\min( 2^j, 2^{2\log_2 (nW) - j} )$. 
\end{definition}

We write $P^{i} = P_0 \circ P_1 \circ \ldots \circ P_i$ for the constituting subpaths 
of an $i$-partial $\ell$-expath, that is, $P_j$ is an $\ell$-decomposable path and $|P_j| \le \delta_j$.
The definition interpolates between $\ell$-decomposable paths (i.e., $0$-partial $\ell$-expaths) 
and ordinary $\ell$-expaths (i.e., $(2 \log_2(nW))$-partial $\ell$-expaths).

The shortest $\ell$-expath from $s$ to $t$ in $G - A$ is computed in $2\log_2(nW)+1$ phases.
At the start of the $i$-th phase, we already have the shortest $(i{-}1)$-partial $\ell$-expath 
$P^{i-1}$ from~$s$ to $v$ in $G - A$ for each $v \in V$.
We extend them to the shortest $i$-partial $\ell$-expath 
$P^i = P^{i-1} \circ  P_{i}$ from $s$ to $v$ in $G - A$, again for each $v \in V$. 

Let $d^{i,\ell}(s,v,A)$ denote the length of the shortest $i$-partial $\ell$-expath from $s$ to $v$ in $G{-}A$. We define a directed graph $G_i$ as follows. We set $G_i= (V_i, E_i)$ with $V_i = \{s^*\} \cup \bigcup_{j=0}^{\ell}{V_{i,j}}$, where all the $V_{i,j}$ are copies of $V$.
In the remaining description, for every $v \in V$ and every $0 \leq j \leq \ell$, we denote by $v_{j}$ the copy of $v$ contained in $V_{i,j}$. 
The graph $G_i$ has the following edges $E_i$.
\begin{itemize}
    \item $(s^*,v_{0})$ of weight $w_{G_i}(s^*,v_{0}) = d^{i-1,\ell}(s,v,A)$
    	for every $v \in V$ with $d^{i-1,\ell}(s,v,A) \neq +\infty$;
    \item $(v_{j-1},v_{j})$ of weight $w_{G_i}(v_{j-1},v_{j}) = 0$
    	for every $v \in V$ and every $1 \leq j \leq \ell$;
    \item $(u_{j-1},v_j),(v_{j-1},u_j)$ both of weight $w(u,v)$
    	for every edge $\{u,v\} \in E{\setminus}A$ and every $1 \leq j \leq \ell$;
    \item $(u_{j},v_{j}), (v_j,u_j)$ of weight $w(u,v)$
    	for every \mbox{$\{u,v\} \in E{\setminus}A$} and every $0 \leq j \leq \ell$.
\end{itemize}
The construction is visualized in~\Cref{fig:f-expath-construction}. 
Intuitively, this adds a new $\ell$-decomposable path to the previously computed partial expath, by using the same exploded graph as in \Cref{fig:exploded graph}.
Regarding the weights of the out-edges of $s^*$ in the first graph $G_0$,
note that $d^{-1,\ell}(s,v,A)$, by definition, is $1$ if $\{s,v\} \in E{\setminus}A$;
and $+\infty$ otherwise.
Let $d'(s^*,v_j)$ be the values we compute in $G_i$ by running Dijkstra's algorithm
from the source $s^*$ with the following modifications.
\begin{enumerate}
    \item For each vertex $v_j \in V_i$, we store \emph{entry vertices} $w_{v_j}$ and $x_{v_j}$.
    \item If \texttt{DecreaseKey} is called on $v_0 \in V_{i,0}$ upon relaxation 
    	of an edge $(s^*, v_0)$, $w_{v_0}$ is set to $v_0$.
    \item If \texttt{DecreaseKey} is called on a vertex $v_{j} \in V_{i,j}$ upon relaxation
    	of an edge $(u_{j-1}, v_{j})$ from the previous layer, then $x_{v_{j}}$ is set to $v_{j}$.
    \item For all other calls of \texttt{DecreaseKey}, when relaxing
    edge $(u_j,v_j)$, set $w_{v_j}=w_{u_j}$ and $x_{v_j}=x_{u_j}$.
    \item \label[modification]{mod:shortest-in-g} To relax an edge $(u_{j}, v_{j})$,
    	we require 
    	$d'(s^*, u_{j}) - d'(s^*, x_{u_{j}}) + w_{G_i}(u_{j}, v_{j}) = d_G(u,v)$.
    \item \label[modification]{mod:limit-length} To relax an edge $(u_{j'},v_{j})$, 
    	$j' \in \{j-1,j\}$, we require 
    	$d'(s^*, u_{j'}) - d'(s^*, w_{u_{j'}}) + w(u_{j'}, v_j) \leq \delta_i$.
\end{enumerate}
Finally, we set $d^{i,\ell}(s,v,A) = d'(s^*, v_{\ell})$ for each $v \in V$.

The modifications are such that $w_{v_j}$ marks the entry point of the current shortest path from $s^*$ to $v_j$ into the graph induced by $V_i{\setminus}\{s^*\}$, 
while $x_{v_j}$ marks the entry point into the layer $V_{i,j}$.
\Cref{mod:shortest-in-g} further ensures that a path entirely contained within one layer
corresponds to a shortest path in $G$.
This implies that a shortest path from $s^*$ to a vertex $v_\ell \in V_{i,\ell}$
corresponds to a composition of an $(i{-}1)$-partial $\ell$-expath, via the edge $(s^*,w_{v_\ell})$,
and an $\ell$-decomposable path.
\Cref{mod:limit-length} enforces that the $\ell$-decomposable path we append
has length at most $\delta_i$.

\begin{lemma}
\label{lem:l-expath}
	Given the original distances $d_G(u,v)$ for all $u,v \in V$, 
	a set $A \subseteq E$, and the distances $d^{i-1,\ell}(s,v,A)$ for all $v \in V$,
	the distances $d^{i,\ell}(s, v, A)$ for all $v$ are computable in total time $\Otilde(\ell m)$.
\end{lemma}

\begin{proof}
First, we show that, for an arbitrary vertex $t \in V$ with $d'(s^*,t_\ell) \neq +\infty$,
there exists an $i$-partial $\ell$-expath in $G-A$ of length $d'(s^*,t_\ell)$.
Consider the $s^*$-$t_\ell$-path $Q$ through $G_i$ computed by our algorithm.
For each layer $j$, let $P_{i,j}$ be the path in $G$ corresponding 
to the subpath of $Q$ within $V_{i,j}$. 
That means, for every $(u_j, v_j) \in E(Q)$, $P_{i,j}$ contains the edge $\{u, v\}$.
Due to~\Cref{mod:shortest-in-g}, we only relax edges $(u_j, v_j)$,
if the distance to the current entry vertex $x_{v_j}$ into the $j$-th layer,
corresponding to some $x \in V$, is equal to the $d_G(x,v)$.
Thus, each $P_{i,j}$ is a shortest path in $G$.

Recall that we use symbol $P_i$ for the $(i{+}1)$th constituting $\ell$-de\-com\-pos\-able subpath of 
the $i$-partial $\ell$-expath $P^{i}$ we aim to construct. 
We define $P_i$ by interleaving all the paths $P_{i,j}$ 
with the edges corresponding to the layer transitions. 
In more detail, let $(u_j,v_{j+1})$ be the edge leaving $V_{i,j}$.
Note that $Q$ never returns to $V_{i,j}$.
We add the corresponding edge $\{u,v\}$ between $P_{i,j}$ and $P_{i,j+1}$ if $u \neq v$;
otherwise, we concatenate $P_{i,j}$ and $P_{i,j+1}$ directly.
Since $P_i$ consists of $\ell+1$ shortest paths in $G$ possibly interleaved with single edges,
$P_i$ is indeed an $\ell$-decomposable path.
By the definition of $E_i$, path $P_i$ exclusively uses edges from $E{\setminus}A$.

The path $P_i$ starts in the vertex $w \in V$ corresponding to $w_{t_\ell} \in V_{i,0}$, 
hence the length of $P_i$ is exactly $d'(s^*, t_\ell) - d'(s^*, w_{t_\ell})$
as our transformation preserves edge weights
and the edges $(v_{j-1},v_j)$ between vertices corresponding to the same $v$ have weight $0$ in $G_i$. 
By~\Cref{mod:limit-length}, the length of $P_i$ is bounded by $\delta_i$
as otherwise the last edge would not have been relaxed.
Let $P^{i-1}$ be the $(i{-}1)$-partial $\ell$-expath corresponding to the edge $(s^*, w_{t_\ell})$
with length $d^{i,\ell}(s,w,A)$.
In summary, $P^{i} = P^{i-1} \circ P_i$ is an $i$-partial $\ell$-expath
that has length $\length{Q} = d'(s^*, t_\ell)$.

It remains to prove that $P^{i}$ is the shortest such path in $G - A$.
Assume there is a shorter $i$-partial $\ell$-expath $P' = P'_0 \circ \ldots \circ P'_i$.
Let $x'$ be the first vertex of $P'_i$, then $w_{G_i}(s^*, x'_0) = d^{i,\ell}(s,x,A) \le w( P'_0 \circ \ldots \circ P'_{i-1})$ as $P'_0 \circ \ldots \circ P'_{i-1}$ is an $(i{-}1)$-partial $\ell$-expath
from $s$ to $x$ in $G-A$.
Also, $P'_i$ is an $\ell$-decomposable path from $x'$ to $t$ of length $\length{P'_i} \le \delta_i$.

Let $Q'_i$ be the corresponding path through $G_i$ from $x'_0$ to $t_{\ell}$.
Then, the path $Q' = (s^*, x'_0) \circ Q'_i$ is an $s^*$-$t_{\ell}$-path in $G_i$ 
that is shorter than the path $Q$ that our algorithm found.
Since Dijkstra's (original) algorithm is correct and $Q'_i$ has length $\length{P'_i} \le \delta_i$, this can only happen due to \Cref{mod:shortest-in-g}.
During the computation, some edge $(u_j,v_j) \in E(Q'_i)$ 
satisfies $d'(s^*,u_j) - d'(s^*,x_{u_j}) + w_{G_i}(u_j,v_j) > d_G(u,v)$.
Thus, the subpath of $Q'_i$ between $x_{u_j}$ and $v_j$ is 
entirely contained in $V_{i,j}$ but not a shortest path in $G$.
This is a contradiction to $P'_i$ being $\ell$-decomposable.
Since $G_i$ has $O(\ell n)$ vertices and $O(\ell m)$ edges the $i$-th phase of
our modified algorithm runs in time $\Otilde(\ell m)$.
\end{proof}

If desired, the computed $i$-partial $\ell$-expath can be reconstructed by storing the parent of the relaxed vertex whenever \texttt{DecreaseKey} is called.
Additionally, we can label the start and endpoints of the $\ell$-decomposable paths 
as well as the shortest paths within them, by inserting labels for the first vertex after $s^*$
and when an edge transitions from one layer to the next.

\Cref{lem:l-expath} implies the following result
that we frequently referenced in \Cref{sec:sublinear_query_time,sec:improved_query_time}.\footnote{%
	Recall that in \Cref{sec:sublinear_query_time,sec:improved_query_time}
	the maximum weight is $W = 1$, whence the running time simplifies to $\Otilde(fm)$,
	and further to $\Otilde(m)$ as $f$ is assumed to be constant.
}

\begin{corollary}
Given vertices $s,t \in V$, the original distances $d_G(u,v)$ for all $u,v \in V$,
and edges $A \subseteq E$, the shortest $(2f{+}1)$-expath from $s$ to $t$ in $G{-}A$
is computable in time $\Otilde(fm \log(nW))$.
\end{corollary}

\subsection{Expaths with Granularity}
\label{subsec:improved_preproc_granularity}

It is also not hard to extend the efficient expath computation to positive granularity $\lambda$
(\Cref{definition:l-expath_granularity}).
The difference is that the path now may have a pre- and suffix of up to $\lambda$ edges each.
Recall that $d_{G-A}^{\le \lambda}(u,v)$ is the minimum length of paths 
between vertices $u,v$ in $G-A$
that have at most $\lambda$ edges; or $+\infty$ if no such path exists.
We first prepare the distances $d_{G-A}^{\le \lambda}(s,v)$
and $d_{G-A}^{\le \lambda}(v,t)$ for all $v \in V$ by running Dijkstra's algorithm
from $s$ and from $t$, respectively.
This takes time $\Otilde(m)$ since $\lambda \le n$,
whence it does not affect the total computation time.

Let $G_0, \ldots, G_{2 \log n}$ be the graphs defined above,
where $G_i$ is used to compute the $i$-partial \mbox{$\ell$-expath}.
We incorporate the prefix of $\lambda$ edges in the Dijkstra run from $s^*$ in $G_0$.
We set the weight of the edges $(s^*, v_{0})$ for every $v \in V$ to $d_{G-A}^{\le \lambda}(s,v)$
(and omit the edge in case $d_{G-A}^{\le \lambda}(s,v) = +\infty$).
For the suffix, we add a final node $t^*$ after the last layer of the last graph $G_{2\log n}$
The weight of the edges $(v_{\ell}, t^*)$ is $d_{G-A}^{\le \lambda}(v, t)$.

It is not difficult to see that the algorithm for computing shortest $(2f{+}1)$-expaths 
(without granularity) from \Cref{sec:preprocessing-expaths} together with these adaptions computes shortest \mbox{$(2f{+}1)$-expaths} with granularity $\lambda$ in $G{-}A$.

\subsection{Improved Preprocessing of the Distance Sensitivity Oracle of \texorpdfstring{\\}{} Chechik, Cohen, Fiat, and Kaplan}
\label{subsec:improved_preproc_DSO}

We now plug our expath computation into the preprocessing algorithm
of the $f$-DSO of Chechik et al.~\cite{ChCoFiKa17}.
The sensitivity $f$ can grow to $o(\log n/\log\log n)$ in their setting
and the underlying graph $G$ is weighted with a polynomial maximum weight $W = \poly(n)$.
They use fault-tolerant trees $FT(u,v)$ for all pairs of vertices $u,v \in V$
incurring super-quadratic space.
Every node $\nu$ in an FT-tree is associated with a specific subgraph $G_\nu \subseteq G$.
To obtain the expath $P_{\nu}$ from $u$ to $v$,
all-pairs shortest paths in $G_\nu$ are computed
and then assembled in time
$\Otilde(fn^3 + n^2\log (nW) + n \log (nW) \log \log (nW)) = \Otilde(fn^3)$ \emph{per node}.
With $O(n^2)$ FT-trees having $O(\log(nW)/\varepsilon)^f$ nodes each,
this makes for a preprocessing time of
$\Otilde(fn^5) \cdot O(\log(nW)/\varepsilon)^f
  = O(1/\varepsilon^f)  \cdot n^{5+o(1)}$.

We have shown above that APSP is only needed in the original graph $G$ to obtain 
the expaths in all relevant subgraphs,
taking only $\Otilde(mn)$ time.
Our algorithm for expaths then reduces the time to construct one node of an FT-tree to $\Otilde(fm)$.
In total, we obtain a preprocessing time of 
$\Otilde(mn) + \Otilde(fmn^2) \cdot O(\log(nW)/\varepsilon)^f 
	= O(1/\varepsilon^f)  \cdot mn^{2+o(1)}$.
The stretch of $1+\varepsilon$, space $\Otilde(fn^2) \cdot O(\log(nW)/\varepsilon)^f$,
and query time $O(f^5 \log n)$ of the DSO remain the same as for Chechik et al.~\cite[Theorem~3.2]{ChCoFiKa17}.
This proves our~\Cref{thm:improved_preprocessing}.

\printbibliography

\end{document}